\documentclass[journal]{IEEEtran}

\usepackage[english]{babel}
\usepackage{booktabs}

\usepackage{ifpdf}

\usepackage{cite} 
\usepackage{url}
\usepackage{hyperref}
\usepackage[framemethod=default]{mdframed}

\ifCLASSINFOpdf
	\usepackage[pdftex]{graphicx}
	\graphicspath{{./figures/}}
\else
	\usepackage[dvips]{graphicx}
	\graphicspath{./figures/}
\fi
\usepackage{color}

\usepackage{pgf, tikz, pgfplots}
\usetikzlibrary{shapes, arrows, automata}
\usetikzlibrary{calc,hobby,decorations}

\usepackage[cmex10]{amsmath}
\usepackage{amsfonts, amssymb, amsthm}
\usepackage{mathrsfs}
\usepackage{tcolorbox}

\usepackage{algorithm} 
\usepackage{algorithmic}
\usepackage{enumerate}
\usepackage{multirow}
\usepackage{rotating}
\usepackage{subcaption}
\usepackage{needspace}

\input{mySymbol.sty}

\definecolor{penndarkestblue}{cmyk}{1,0.74,0,0.77}
\definecolor{penndarkerblue}{cmyk}{1,0.74,0,0.70}
\definecolor{pennblue}{cmyk}{0.99,0.66,0,0.57} 
\definecolor{pennlighterblue}{cmyk}{0.98,0.44,0,0.35}
\definecolor{pennlightestblue}{cmyk}{0.38,0.17,0,0.17} 

\definecolor{penndarkestred}{cmyk}{0,1,0.89,0.66}
\definecolor{penndarkerred}{cmyk}{0,1,0.88,0.55}
\definecolor{pennred}{cmyk}{0,1,0.83,0.42} 
\definecolor{pennlighterred}{cmyk}{0,1,0.6,0.24}
\definecolor{pennlightestred}{cmyk}{0,0.43,0.26,0.12} 

\definecolor{penndarkestgreen}{cmyk}{1,0,1,0.68}
\definecolor{penndarkergreen}{cmyk}{1,0,1,0.57}
\definecolor{penngreen}{cmyk}{1,0,1,0.44} 
\definecolor{pennlightergreen}{cmyk}{1,0,1,0.25}
\definecolor{pennlightestgreen}{cmyk}{0.43,0,0.43,0.13}

\definecolor{penndarkestorange}{cmyk}{0,0.65,1,0.49}
\definecolor{penndarkerorange}{cmyk}{0,0.65,1,0.33}
\definecolor{pennorange}{cmyk}{0,0.54,1,0.24} 
\definecolor{pennlighterorange}{cmyk}{0,0.32,1,0.13}
\definecolor{pennlightestorange}{cmyk}{0,0.15,0.46,0.06}
	
\definecolor{penndarkestpurple}{cmyk}{0,1,0.11,0.86}
\definecolor{penndarkerpurple}{cmyk}{0,1,0.13,0.82}
\definecolor{pennpurple}{cmyk}{0,1,0.11,0.71} 
\definecolor{pennlighterpurple}{cmyk}{0,1,0.05,0.46}
\definecolor{pennlightestpurple}{cmyk}{0,0.35,0.02,0.23}
	
\definecolor{pennyellow}{cmyk}{0,0.20,1,0.05} 
\definecolor{pennlightgray1}{cmyk}{0,0,0,0.05}
\definecolor{pennlightgray3}{cmyk}{0.01,0.01,0,0.18}
\definecolor{pennmediumgray1}{cmyk}{0.04,0.03,0,0.31}
\definecolor{pennmediumgray4}{cmyk}{0.08,0.06,0,0.54}
\definecolor{penndarkgray2}{cmyk}{0.09,0.07,0,0.71}
\definecolor{penndarkgray4}{cmyk}{0.1,0.1,0,0.92}

\def\SO3{\mathrm{SO(3)}}

\newtheorem{assumption}{\hspace{0pt}\bf Assumption \hspace{-0.15cm}}
\newtheorem{lemma}{\hspace{0pt}\bf Lemma}

\newtheorem{theorem}{\hspace{0pt}\bf Theorem}
\newtheorem{corollary}{\hspace{0pt}\bf Corollary}

\newtheorem{definition}{\hspace{0pt}\bf Definition}

\begin{document}

\title{Graph Neural Networks over the Air \\ for Decentralized Tasks in Wireless Networks} 

\author{\IEEEauthorblockN{Zhan Gao$^{\dagger}$ and Deniz~G{\"u}nd{\"u}z$^{\ddagger}$, ~\IEEEmembership{IEEE Fellow}}\\

%
\thanks{Preliminary results were presented in EUSIPCO 2023 \cite{10289977}. $^{\dagger}$Department of Computer Science and Technology, University of Cambridge, Cambridge, UK (Email: zg292@cam.ac.uk). $^{\ddagger }$Department of Electrical and Electronic Engineering, Imperial College London, London, UK (Email: d.gunduz@imperial.ac.uk).}
}

\markboth{}%
{Graph Neural Networks over the Air for Decentralized Tasks in Wireless Networks}

\maketitle

\begin{abstract}
Graph neural networks (GNNs) model representations from networked data and allow for decentralized inference through localized communications. Existing GNN architectures often assume ideal communications and ignore potential channel effects, such as fading and noise, leading to performance degradation in real-world implementation. Considering a GNN implemented over nodes connected through wireless links, this paper conducts a stability analysis to study the impact of channel impairments on the performance of GNNs, and proposes graph neural networks over the air (AirGNNs), a novel GNN architecture that incorporates the communication model. AirGNNs modify graph convolutional operations that shift graph signals over random communication graphs to take into account channel fading and noise when aggregating features from neighbors, thus, improving architecture robustness to channel impairments during testing. We develop a channel-inversion signal transmission strategy for AirGNNs when channel state information (CSI) is available, and propose a stochastic gradient descent based method to train AirGNNs when CSI is unknown. The convergence analysis shows that the training procedure approaches a stationary solution of an associated stochastic optimization problem and the variance analysis characterizes the statistical behavior of the trained model. Experiments on decentralized source localization and multi-robot flocking corroborate theoretical findings and show superior performance of AirGNNs over wireless communication channels.
\end{abstract}


\IEEEpeerreviewmaketitle


\section{Introduction} \label{sec:intro}

Graph neural networks (GNNs) are one of the key tools to extract features from networked data \cite{scarselli2008graph, defferrard2016convolutional, gao2021stochastic, kenlay2021interpretable, xu2018powerful}, and have found wide applications in wireless communications \cite{gao2020resource, eisen2020optimal, gao2023decentralized}, multi-agent coordination \cite{tolstaya2020learning, chen2021graph, gao2023environment, gao2023constrained} and recommendation systems \cite{ying2018graph, wu2019session, 10042031}. GNNs extend convolutional neural networks to the graph setting by employing a multi-layered architecture with each layer comprising a graph filter bank and a pointwise nonlinearity \cite{gama2018convolutional}. With the distributed nature of graph filters, GNNs can compute output features with local neighborhood information. This makes them well-suited candidates for decentralized tasks over networks, where each node takes actions by only sensing its local states and communicating with its neighboring nodes \cite{gama2022synthesizing, gao2022wide}. 

However, when implemented in a decentralized fashion over remote nodes, information exchanges among neighboring nodes of GNNs require wireless transmissions. Existing works often assume ideal communications and ignore channel impairments, such as fading and noise, which impose inevitable effects on all transmissions \cite{belloni2004fading, peppas2011overview, popa2008fading}. In real-world wireless networks, each node receives faded and noisy messages from its neighbors and computes perturbed outputs that mismatch the ones assuming ideal communications during training; hence, resulting in performance degradation during testing. We note here that it is possible to establish reliable communication links among nodes using conventional communication techniques, e.g., error correction codes, power adaption, multiplexing and retransmissions; however, these techniques introduce significant complexity and delay, and it has been shown that separating learning from communication can lead to severe suboptimality \cite{gunduz2019machine, gunduz2020communicate}. Therefore, in this work, our goal is to incorporate wireless channel characteristics into training and inference of GNNs through an end-to-end problem formulation. 

GNNs have been shown to be effective in solving various wireless resource allocation problems by leveraging channel state information (CSI) \cite{he2021overview}. The works in \cite{shen2019graph, chowdhury2021efficient} utilize GNNs to control transmission powers in interference networks based on instantaneous CSI, while \cite{eisen2020optimal, gao2020resource} impose real-world constraints and learn power allocation strategies with GNNs in a model-free manner. Subsequently, \cite{gu2023graph, yang2023implementing} implement GNNs for decentralized power allocation that are scalable to wireless links, while \cite{li2021energy, naderializadeh2023learning} consider energy efficiency and user selection on the basis of power allocation. On the other hand, \cite{he2020resource} proposes GNNs to extract features from CSI and solve channel allocation problems in vehicle-to-everything networks. The work in \cite{gao2023decentralized} considers the decentralized setting and uses GNNs to learn channel management strategies that minimize the mutual interference in wireless local area networks. Other applications include beamforming design \cite{shen2020graph, zhang2021scalable}, privacy-preserving inference \cite{10140175}, link scheduling \cite{9962800}, etc. However, these works only consider CSI as external (input) features of GNNs to design resource allocation strategies and still ignore the impact of channel impairments in the decentralized implementation of GNNs. Moreover, they focus on specific resource allocation problems, which are not applicable to general decentralized tasks (e.g., multi-robot coordination). 

To account for wireless communication channels in the decentralized implementation of GNNs, we take advantage of the structural similarity between graph convolutional operation and over-the-air computation (AirComp) \cite{goldenbaum2013harnessing, amiri2020federated, fiorellino2024topological, amiri2021blind}. In particular, graph convolutional operation aggregates neighborhood information via wireless communications to generate higher-level features, which is challenging with limited spectrum bandwidth especially when the aggregation needs to be dealt in a timely manner. On the other hand, the generated features do not depend on the individual signal of each neighbor but only concern the fusion of all neighbors' signals \cite{gama2018convolutional, ruiz2021graph}. This motivates to leverage AirComp, which is an analog wireless transmission technique that enables efficient data collection from concurrent channels for decentralized implementation. 

AirComp is capable of calculating the weighted sum of neighbors' signals by exploiting the signal-superposition property of wireless medium \cite{goldenbaum2013robust, amiri2020machine, zhu2019broadband}. It allows to transmit signals simultaneously with non-orthogonal wireless resources and effectively integrates communication with computation, which, in turn, reduce required spectrum resources and speed up wireless transmissions. The research on AirComp more focused on developing transmission schemes that minimize the impact of channel impairments on multi-signal aggregation, which include beamforming design \cite{zhu2018mimo, lin2022distributed}, power control \cite{li2019wirelessly, liu2020over} and user scheduling \cite{tang2022node, du2023gradient}. AirComp has been applied in various mobile applications with integrated communication and computation tasks \cite{csahin2023survey}. For example, it is used in federated learning for gradient transmissions among local edge devices and the global edge server \cite{yang2020federated, aygun2022hierarchical, jing2022federated}, or for collaborative inference across an ensemble of models at the edge \cite{yilmaz2022over}.

By leveraging the aforementioned connection between GNN implementation and over-the-air computation, this paper proposes graph neural networks over the air (AirGNNs) that incorporate wireless communication channels directly into the architecture for decentralized implementation. This is inspired by \cite{jankowski2022airnet}, which showed that incorporating channel impairments into training makes deep neural networks more robust when network parameters are delivered via wireless communications. AirGNNs account for channel fading and Gaussian noise when exchanging messages among neighboring nodes during training, where each node generates features based on faded and noisy information aggregated from its neighbors. This yields the learned parameters robust to channel impairments during implementation and improves performance in decentralized tasks. Our contributions are summarized as follows: 

\begin{enumerate}[(i)]
	
	\item \emph{Graph neural networks over the air (Section \ref{sec:AirGNNs}).} We develop the AirGNN framework as a multi-layered architecture where each layer consists of graph filters over the air (AirGFs) and a pointwise nonlinearity. AirGFs shift graph signals over wireless communication channels in an analog fashion to collect multi-hop neighborhood information and generate features by aggregating information with over-the-air computation. 
	
	\item \emph{Stability analysis (Section \ref{sec:stabilityAnalysis}).} We analyze the impact of random channel impairments on the output of AirGNNs, and show that the output variation caused by channel impairments is bounded by factors that are quadratic in the mean and standard deviation of fading and noise. We conduct such analysis in the spectral domain by generalizing the graph filter frequency response to scenarios over wireless communication channels, which allows to establish stability results that hold uniformly for all graphs.
	
	\item \emph{AirGNNs with and without CSI (Sections \ref{sec:airGNNsCSI}-\ref{sec:AirGNNNoCSI}).} We design different training and implementation schemes for AirGNNs with and without CSI. When CSI is available at transmitting nodes, we propose truncated channel-inversion power control for transmission, formulate the cost function w.r.t. truncated channel randomness, and develop a stochastic gradient descent (SGD) based method to train AirGNNs. When CSI is not available, we consider channel randomness directly and train AirGNNs over the entire CSI space.  
	
	\item \emph{Convergence and variance analysis (Section \ref{sec:convergence}).} We show that the training procedure of the AirGNN is equivalent to solving an associated stochastic optimization problem with SGD, and prove its convergence to a stationary solution. We further prove the output variance of the trained AirGNN is bounded by the mean and standard deviation of channel fading and noise, showing that its output does not diverge but varies around the optimal expected solution of SGD. 
	
\end{enumerate}

Numerical results on decentralized tasks of source localization and multi-robot flocking corroborate our theoretical findings in Section \ref{sec:experiments}. The conclusions are drawn in Section \ref{sec:conclusion}, and all proofs are summarized in the supplementary material.

\section{Graph Neural Networks over the Air}\label{sec:AirGNNs}

Consider a graph $\mathcal{G}= ( \mathcal{V}, \mathcal{E} )$ with node set $\mathcal{V} = \{ 1,\cdots,n \}$ and edge set $\ccalE$. The graph structure can be represented by an $n \times n$ matrix $\bbS$, referred to as the graph shift operator, with $(i,j)$th entry $[\bbS]_{ij}=s_{ij}$ non-zero if node $j$ is connected to node $i$, i.e., $(j,i)\in \mathcal{E}$, or $i=j$, e.g., the adjacency matrix $\bbA$ \cite{sandryhaila2013discrete}. The graph signal is a vector defined on the nodes $\bbx = [x_1,...,x_n]^\top$, where the $i$th entry $x_i$ is the signal value associated to node $i$. The graph serves as a mathematical representation of the network topology and the graph signal captures the network states. For example, in a multi-agent system, nodes correspond to agents, edges to communication links, and signal values to agent states \cite{gao2023online}. 

The graph shift operation $\bbS \bbx$ is one of the key operations in graph signal processing. It assigns to each node $i$ the aggregated information from its neighboring nodes $[\bbx^{(1)}]_i = [\bbS\bbx]_i$, which extends signal shifting from the time / space domain to the graph domain. In the decentralized setting over physical networked systems (e.g., sensor networks and multi-robot systems), this corresponds to message exchanges between the neighboring sensors / robots through available communication links. However, $\bbS\bbx$ assumes perfect signal transmissions among nodes and ignores the fact that communication links suffer from channel fading and additive noise in the real world.

\smallskip
\noindent \textbf{Communication channel.} We consider that each node exchanges information wirelessly with the neighboring nodes in its communication range. The wireless channel between nodes $i$ and $j$ is modeled as a slow fading channel
\begin{align}\label{eq:communicationChannelIndividual}
	[\bbz]_i = h_{ij} p [\bbx]_j + n_{i},
\end{align}
where $h_{ij}$ is the channel coefficient from node $j$ to node $i$ with Rayleigh distribution,\footnote{The channel coefficient $h_{ij}$ can follow any random distributions, which depend on specific application scenarios. We consider Rayleigh distribution following the common assumption without loss of generality.} $p$ is the transmitted power of node $j$, and $n_i$ is the zero-mean Gaussian noise at node $i$. The channel coefficients $\{h_{ij}\}$ are assumed independent across communication links, which are fixed during a single transmission but changes from one transmission to the next in an independent and identically distributed (i.i.d.) manner.

\smallskip
\noindent \textbf{Over-the-air computation (AirComp).} Since each node aggregates the received signal values by summing them up, it does not need signal values of individual neighbors but only their sum. This motivates to transmit these signals in an uncoded / analog manner with \textit{over-the-air computation} \cite{amiri2020federated, amiri2021blind}. Specifically, at each communication round, we assume that all nodes transmit signal values simultaneously in a synchronized manner. AirComp allows to collect neighbors' signals through wireless communication channels in a way as 
\begin{align}\label{eq:communicationChannel}
	[\bbx^{(1)}]_i \!=\! \sum_{j \in \ccalN_i}\! h_{ij}^{(1)} p [\bbx]_j + n_{i}^{(1)},~\text{for}~i=1,...,n,
\end{align}
where $\ccalN_i$ is the neighbor set of node $i$, i.e., $\ccalN_i = \{j | (i,j)\in \ccalE\}$, and the superscript $(1)$ represents the $1$-hop communication.

\smallskip
\noindent \textbf{Graph shift operation over the air (AirGSO).} While presenting an efficient integration of communication and computation by leveraging the signal-superposition property, AirComp introduces channel fading and noise as shown in \eqref{eq:communicationChannel}. With these channel impairments, each node receives faded and noisy versions of the messages sent by its neighbors. By multiplying the received signal with $1/p$, the resulting signal is
\begin{align}\label{eq:communicationChannel1}
	[\bbx^{(1)}]_i \!=\! \sum_{j\in\ccalN_i}\! h_{ij}^{(1)} [\bbx]_j + \frac{n_{i}^{(1)}}{p},~\text{for}~i=1,...,n,
\end{align}
and the graph shift operation through wireless communication channels takes the form of
\begin{align}\label{eq:siganlShiftingChannel}
	\bbx^{(1)} = \bbS_{\rm air}^{(1)} \bbx + \bbn^{(1)},
\end{align}
where $\bbS^{(1)}_{\rm air}$ is the graph shift operator with channel coefficients, i.e., 
$[\bbS^{(1)}_{\rm air}]_{ij} = h_{ij}^{(1)}$ for any $(i,j)\in \ccalE$, and $\bbn^{(1)}$ is the Gaussian noise vector with $[\bbn^{(1)}]_i = n^{(1)}_{i}/p$ for $i=1,...,n$. Note that, different from its ideal counterpart $\bbS\bbx$, \eqref{eq:siganlShiftingChannel} depends not only on the underlying graph topology but also on communication effects. We refer to the latter as \emph{graph shift operation over the air} (AirGSO).

\smallskip
\noindent \textbf{Graph filter over the air (AirGF).} AirGF is a linear combination of multi-shifted signals over the air. Specifically, an AirGSO shifts $\bbx$ over $\bbS$ through wireless communication channels and obtains the one-shifted signal $\bbx^{(1)}$ that aggregates $1$-hop neighborhood information [cf. \eqref{eq:siganlShiftingChannel}]. By recursively shifting $\bbx$ in this fashion, the $k$-shifted signal $\bbx^{(k)}$ becomes 
\begin{align}\label{eq:kShiftedSignal}
	&\bbx^{(k)} = \bbS_{\rm air}^{(k)} \Big(\bbS_{\rm air}^{(k-1)} \big( \cdots + \bbn^{(1)} \big) + \bbn^{(k-1)}\Big) + \bbn^{(k)}\\
	&= \prod_{\kappa=1}^k \bbS_{\rm air}^{(\kappa)}\bbx + \Big(\!\sum_{\kappa=1}^{k-1} \prod_{\tau = \kappa+1}^k \!\bbS_{\rm air}^{(\tau)} \bbn^{(\kappa)} \!+\! \bbn^{(k)}\!\Big)\nonumber \!=\! \bbP^{(k)} \!+\! \bbN^{(k)},
\end{align}
where $\bbS_{\rm air}^{(k)}$ and $\bbn^{(k)}$ are, respectively, the AirGSO and Gaussian noise vector at the $k$th graph shift operation,\footnote{For convenience of expression, we assume $\sum_{a}^b = 0$ and $\prod_{a}^b=1$ if $b < a$.} and $\bbP^{(k)}$ and $\bbN^{(k)}$ are, respectively, the signal and noise components of $\bbx^{(k)}$. The $k$-shifted signal $\bbx^{(k)}$ accesses farther nodes and aggregates the $k$-hop neighborhood information with communication impairments. Given a sequence of shifted signals $\{\bbx, \bbx^{(1)},\ldots,\bbx^{(K)}\}$, the AirGF aggregates them with a set of filter parameters $\{\alpha_k\}_{k=0}^K$ as
\begin{align}\label{eq:AirGF}
	&\bbH_{\rm air} (\bbS)\bbx = \sum_{k=0}^K \alpha_k \bbx^{(k)} = \bbP_{\rm air}(\bbS,\bbx) + \bbN_{\rm air}(\bbS,\bbn)\\
	& = \sum_{k=0}^K \alpha_k \prod_{\kappa=1}^k \bbS_{\rm air}^{(\kappa)}\bbx + \sum_{k=1}^{K} \sum_{\kappa=k}^{K} \alpha_{\kappa}\! \prod_{\tau=k+1}^{\kappa} \bbS_{\rm air}^{(\tau)} \bbn^{(k)} 
	\nonumber
\end{align}
with $\bbx^{(0)} = \bbx$. The first term $\bbP_{\rm air}(\bbS,\bbx)$ in \eqref{eq:AirGF} captures the signal information perturbed by channel fading, referred to as the signal component, while the second term $\bbN_{\rm air}(\bbS,\bbn)$ is the accumulated noise, referred to as the noise component. The AirGF is a shift-and-sum operator that extends graph convolution by incorporating wireless communication channels. It aggregates the neighborhood information up to a radius of $K$ while accounting for channel fading and Gaussian noise when shifting signals over the underlying graph. AirGFs reduce to conventional graph filters in ideal scenarios with perfect communication links, i.e., $h_{ij}^{(1)} \!=\! 1$ and $n_i^{(1)} \!=\! 0$ in \eqref{eq:communicationChannel1}.

\smallskip
\noindent \textbf{AirGNN.} AirGNN consists of multiple layers, where each layer comprises a bank of AirGFs and a pointwise nonlinearity. Specifically, at layer $\ell$, we have $F$ input signals $\{\bbx_{\ell-1}^g\}_{g=1}^{F}$ generated by the previous layer $\ell-1$. The latter are processed by a bank of AirGFs $\{\bbH_{{\rm air}, \ell}^{fg}\}_{fg}$, aggregated over the input index $g$, and passed through a pointwise nonlinearity $\sigma(\cdot)$ as
\begin{align}\label{eq:AirGNN}
	\bbx_{\ell}^f = \sigma\left( \sum_{g=1}^{F} \bbH_{{\rm air}, \ell}^{fg}(\bbS) \bbx_{\ell-1}^{g} \right),~\forall~f=1,\ldots,F.
\end{align}
The outputs $\{\bbx_\ell^f\}_{f=1}^F$ are fed into layer $\ell+1$ until the final layer $L$. Without loss of generality, we consider a single input $\bbx_0^1 = \bbx$ and output $\bbx_L^1$. The AirGNN can be represented as a nonlinear mapping $\bbPhi_{\rm air}(\bbx, \bbS, \ccalA)$ of graph signals $\bbx$, where $\ccalA$ are the architecture parameters comprising all filter parameters. AirGNNs reduce to conventional GNNs if assuming ideal communications without channel fading and Gaussian noise among the neighboring nodes. 

\smallskip
\noindent \textbf{Decentralized implementation.} AirGNNs and AirGFs are ready for decentralized implementation, i.e., each node can compute its own output with local neighborhood information. Specifically, the one-shifted signal $[\bbx^{(1)}]_i$ at node $i$ can be computed by receiving signal values $\{[\bbx]_j\}_{j\in \ccalN_i}$ of neighboring nodes transmitted through wireless communication links in an uncoded fashion [cf. \eqref{eq:communicationChannel1}] and summing them up with AirComp. Likewise, $K$ shifted signals $\{[\bbx^{(k)}]_i\}_{k=0}^K$ at node $i$ can be computed with recursive communications. Since the aggregation of the shifted signals $\{[\bbx^{(k)}]_i\}_{k=0}^K$ with the filter parameters $\{\alpha_k\}_{k=0}^K$ does not involve inter-node operations, node $i$ can compute the filter output $[\bbH_{\rm air}(\bbS)\bbx]_i$ locally and the AirGF allows for a decentralized implementation. AirGNN consists of AirGFs and nonlinearities, where the former is decentralized and the latter is pointwise and local; hence, inheriting the decentralized implementation. This implies that a node does not need signal values of all the nodes and the full knowledge of the graph to compute AirGNN / AirGF outputs, but only the ability to communicate local signals in a synchronized manner with its neighbors, in order to allow over-the-air aggregation \cite{amiri2020federated, amiri2021blind}.

\section{Stability Analysis}\label{sec:stabilityAnalysis}

In this section, we analyze the impact of channel impairments, i.e., channel fading and Gaussian noise, on the performance of AirGNNs. Specifically, AirGNNs generate output features based on the impaired neighborhood information, which deviate from those generated based on the ideal information and result in performance degradation. We aim to characterize the output difference caused by channel impairments and identify the effects of channel conditions, filter property and architecture hyper-parameters on the stability of AirGNNs. Since AirGNNs are built upon AirGFs, we first analyze the stability of AirGFs and then extend the result to AirGNNs. 

\subsection{Filter Frequency Response over the Air}

For stability analysis, we consider the channel coefficient $h$ drawn from any random distribution with expectation $\mu$ and standard deviation $\delta$, and the noise $n$ drawn from a Gaussian distribution with zero expectation and standard deviation $\eps$. In this context, the output of an AirGF is a random variable that differs from the deterministic output of a nominal graph filter. 

To characterize such output difference and establish a result that holds uniformly for any underlying graph, we conduct stability analysis in the graph spectral domain. Specifically, let $\bbS = \bbV \bbLambda \bbV^\top$ be the eigendecomposition with eigenvectors $\bbV = [\bbv_1,\ldots,\bbv_n]$ and eigenvalues $\Lambda = {\rm diag}(\lambda_1,\ldots,\lambda_n)$. The graph Fourier transform (GFT) projects the graph signal $\bbx$ on $\bbV$ as $\bbx = \sum_{i=1}^n \hat{x}_i \bbv_i$ and obtains the Fourier coefficients $\hat{\bbx}= [\hat{x}_1,\ldots,\hat{x}_n]^\top$. By substituting the GFT into the nominal graph filter, we have
\begin{align}\label{eq:frequencyResponse}
	\bbu \!=\! \bbH(\bbS)\bbx \!=\!\!\! \sum_{k=0}^K \alpha_k \bbS^k \bbx \!=\!\!\!  \sum_{k=1}^K \!\!\alpha_k \bbS^k \!\!\sum_{i=1}^n\!\! \hat{x}_i \bbv_i \!=\!\!\! \sum_{i=1}^n \!\hat{x}_i\! \sum_{k=0}^K \!\alpha_k \lambda_i^k \bbv_i. 
\end{align}
By further applying the GFT on the filter output $\bbu = \sum_{i=1}^n \hat{u}_i \bbv_i$, we obtain the input-output relation in the spectral domain as $\hat{u}_i = \sum_{k=0}^K \alpha_k \lambda_i^k \hat{x}_i$ for $i=1,\ldots,n$. This motivates to define the filter frequency response as an analytic function
\begin{align}\label{eq:filterResponse1}
	f(\lambda) = \sum_{k=0}^K \alpha_k \lambda^k
\end{align}
of a generic frequency variable $\lambda$. The underlying graph $\bbS$ instantiates specific eigenvalues $\{\lambda_i\}_{i=1}^n$ on the function variable $\lambda$ resulting in specific frequency responses $\{f(\lambda_i)\}_{i=1}^n$, while the filter parameters $\{\alpha_k\}_{k=0}^K$ determine the function shape. However, AirGFs account for channel impairments during message exchanges, such that signals are shifted over a sequence of graphs with random fading $\{\bbS_{\rm air}^{(k)}\}_{k=1}^K$ instead of a deterministic one $\bbS$. The latter requires extending the concept of filter frequency response. 

\smallskip
\noindent \textbf{Filter frequency response over the air.} Denote by $\bbS_{\rm air}^{(k)} = \bbV^{(k)} \bbLambda^{(k)} {\bbV^{(k)}}^\top$ the eigendecomposition of the $k$th AirGSO with eigenvectors $\bbV^{(k)} = [\bbv^{(k)}_1,...,\bbv^{(k)}_n]$ and eigenvalues $\bbLambda^{(k)} = {\rm diag}(\lambda_1^{(k)},...,\lambda_n^{(k)})$ for $k=1,...,K$. By applying the GFT on $\bbx$ over the $1$st AirGSO as $\bbx = \sum_{i_1=1}^n \hat{x}_{i_1}^{(1)} \bbv_{i_1}^{(1)}$, we can represent the signal component $\bbP^{(1)}$ of the one-shifted signal $\bbx^{(1)}$ [cf. \eqref{eq:kShiftedSignal}] as
\begin{equation}\label{eq:AirFrequencyResponse1}
	\bbP^{(1)} = \bbS^{(1)}_{\rm air}\bbx = \bbS^{(1)}_{\rm air} \sum_{i_1=1}^n \hat{x}_{i_1}^{(1)} \bbv_{i_1}^{(1)} = \sum_{i_1=1}^n \hat{x}_{i_1}^{(1)} \lambda_{i_1}^{(1)} \bbv_{i_1}^{(1)}.
\end{equation}
By considering each eigenvector $\bbv_{i_1}^{(1)}$ as an intermediate graph signal and applying the GFT on $\bbv_{i_1}^{(1)}$ over the $2$nd AirGSO as $\bbv_{i_1}^{(1)} = \sum_{i_2=1}^n \hat{x}^{(2)}_{i_1 i_2} \bbv_{i_2}^{(2)}$, the signal component $\bbP^{(2)}$ of the two-shifted signal $\bbx^{(2)}$ can be represented as
\begin{align}\label{eq:AirFrequencyResponse2}
	\bbP^{(2)} &= \bbS^{(2)}_{\rm air}\bbP^{(1)} = \bbS^{(2)}_{\rm air} \sum_{i_1=1}^n \hat{x}_{i_1}^{(1)} \lambda_{i_1}^{(1)} \bbv_{i_1}^{(1)} \\
	&= \sum_{i_1=1}^n \sum_{i_2=1}^n \hat{x}_{i_1}^{(1)}\hat{x}^{(2)}_{i_1 i_2} \lambda_{i_1}^{(1)} \lambda_{i_2}^{(2)} \bbv_{i_2}^{(2)}. \nonumber
\end{align}
Proceeding in a recursive manner, the signal component $\bbP^{(k)}$ of the $k$-shifted signal $\bbx^{(k)}$ can be represented as
\begin{align}\label{eq:AirFrequencyResponse3}
	\bbP^{(k)} = \sum_{i_1=1}^n \cdots \sum_{i_k=1}^n \hat{x}_{i_1}^{(1)} \hat{x}_{i_1i_2}^{(2)} \cdots \hat{x}_{i_{k-1}i_k}^{(k)} \prod_{j=1}^k \lambda_{i_j}^{(j)} \bbv_{i_k}^{(k)}.
\end{align}
We refer to \eqref{eq:AirFrequencyResponse3} as the GFT on $\bbx$ over a sequence of $k$ random AirGSOs $\{\bbS_{\rm air}^{(1)}, \ldots, \bbS_{\rm air}^{(k)}\}$, which depends on all eigenvalues $\bbLambda^{(1)},...,\bbLambda^{(K)}$ and eigenvectors $\bbV^{(1)},...,\bbV^{(K)}$. By aggregating the signal components of $K$ shifted signals $\{\bbP^{(k)}\}_{k=1}^K$ and the graph signal $\bbx$, we can represent the signal component of the AirGF output $\bbP_{\rm air}(\bbS, \bbx)$ [cf. \eqref{eq:AirGF}] as
\begin{align}\label{eq:AirFrequencyResponseFilter}
	\bbP_{\rm air}(\bbS,\! \bbx) \!\!=\!\! \sum_{i_1\!=\!1}^n \!\!\cdots\!\! \sum_{i_K\!=\!1}^n\! \hat{x}_{i_1}^{(1)}\!\hat{x}_{i_1i_2}^{(2)}\!\!\cdots \!\hat{x}_{i_{K\!-\!1}i_K}^{(K)}\!\! \sum_{k\!=\!0}^K\! \alpha_k\! \prod_{j\!=\!1}^k\!\lambda_{i_j}^{(j)} \bbv_{i_K}^{(K)}\!,
\end{align}
where $\{ \hat{x}_{i_1}^{(1)} \}_{i_1=1}^n$ and $\{ \hat{x}_{i_ji_{j+1}}^{(j)} \}_{j=1}^{K-1}$ are the Fourier coefficients of $\bbx$ over a sequence of $K$ random AirGSOs $\{\bbS_{\rm air}^{(k)}\}_{k=1}^K$. The expression in \eqref{eq:AirFrequencyResponseFilter} resembles that in \eqref{eq:frequencyResponse} while taking channel fading effects into account, which motivates to define the filter frequency response over the air of an AirGF as follows.
\begin{definition}[Filter frequency response over the air]\label{def:AirFrequencyResponseFilter1}
	Consider an AirGF defined by filter parameters $\{ \alpha_k \}_{k=0}^K$ over a sequence of $K$ AirGSOs [cf. \eqref{eq:AirGF}]. The filter frequency response over the air is a $K$-dimensional analytic function
	\begin{equation}\label{eq:AirFrequencyResponseFilter1}
		f(\bblambda) = \sum_{k=0}^K \alpha_k \prod_{\kappa=1}^k \lambda^{(\kappa)}
	\end{equation}
	of a generic vector variable $\bblambda = [\lambda^{(1)},\ldots,\lambda^{(K)}]^\top \in \mathbb{R}^K$.
\end{definition}
The filter frequency response over the air $f(\bblambda)$ is a multivariate function, where $\lambda^{(k)}$ is the frequency variable associated to the $k$th AirGSO $\bbS^{(k)}_{\rm air}$. The filter parameters $\{\alpha_k\}_{k=0}^K$ determine the function shape of $f(\bblambda)$, while a specific sequence of AirGSOs $\{\bbS_{\rm air}^{(k)}\}_{k=1}^K$ only instantiates the function variables $\{\lambda^{(k)}\}_{k=1}^K$ on specific eigenvalues. In this context, we can analyze filter behaviors that hold uniformly for any graph by focusing directly on properties of $f(\bblambda)$ -- see Fig. \ref{fig:FrequencyResponseAir} for an example of $f(\bblambda)$ with $K=2$. In ideal scenarios with perfect communication, i.e., $h^{(1)}_{ij} = 1$ and $n^{(1)}_i = 0$ in \eqref{eq:communicationChannel1}, AirGSOs $\{\bbS_{\rm air}^{(k)}\}_{k=1}^K$ boil down to a single GSO $\bbS$, the input-output relation \eqref{eq:AirFrequencyResponseFilter} reduces to \eqref{eq:frequencyResponse}, and the filter frequency response over the air \eqref{eq:AirFrequencyResponseFilter1} recovers the nominal filter frequency response \eqref {eq:filterResponse1}.

\begin{figure}[t]
	\centering
	\includegraphics[width=0.85\linewidth , height=0.575\linewidth, trim=10 10 10 10]{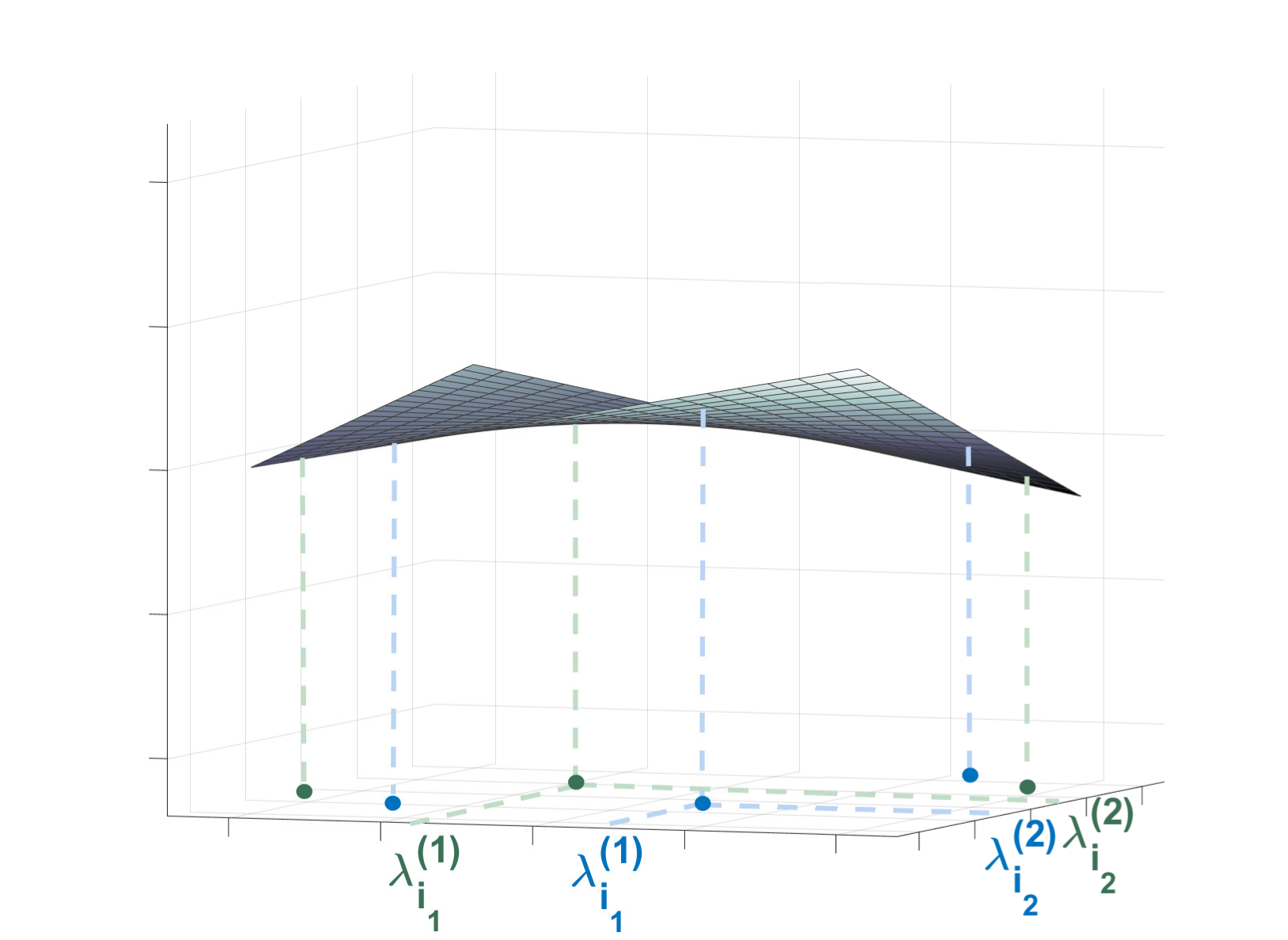}
	\caption{The filter frequency response over the air with $K=2$. The function $f(\bblambda)$ is determined by the filter parameters $\{ \alpha_k \}_{k=0}^K$ and independent of the underlying graph. For two specific sequences of AirGSOs, $h(\bblambda)$ is instantiated on two specific sets of eigenvalues in blue and green.}
	\label{fig:FrequencyResponseAir}\vspace{-4mm}
\end{figure}

\subsection{Stability of AirGFs}

We start by analyzing the stability of AirGFs to channel impairments. Specifically, channel fading effects perturb the underlying GSO $\bbS$ to a sequence of AirGSOs $\{\bbS_{\rm air}^{(k)}\}_{k=1}^K$, generate perturbed frequency variables $\{\lambda^{(k)}\}_{k=1}^K$, and result in random filter frequency responses $h(\bblambda)$. This indicates that the stability depends on the variability of $h(\bblambda)$ and we define the Lipschitz gradient to characterize this spectral variability. 
\begin{definition}[Lipschitz gradient]\label{def:LipschitzGradient}
	Consider the filter frequency response over the air $f(\bblambda)$ of the function variable $\bblambda = [\lambda^{(1)},\ldots,\lambda^{(K)}]^\top$ [cf. \eqref{eq:AirFrequencyResponseFilter1}]. For two specific instantiations $\bblambda_1 = [\lambda^{(1)}_1,\ldots,\lambda_1^{(K)}]^\top$ and $\bblambda_2 = [\lambda^{(1)}_2,\ldots,\lambda_2^{(K)}]^\top$ of the multivariate variable $\bblambda$, denote by $\bblambda_{1:2,k} = [\lambda^{(1)}_1,\ldots,\lambda_1^{(k)}, \lambda_2^{(k+1)}, \ldots, \lambda^{(K)}_2]^\top$ the vector formed by concatenating the first $k$ entries of $\bblambda_1$ and the last $K-k$ entries of $\bblambda_2$. The Lipschitz gradient of $f(\bblambda)$ over two instantiations $\bblambda_1$ and $\bblambda_2$ is defined as
	\begin{align}
		\nabla_L f(\bblambda_1, \bblambda_2) = \Big[ \frac{\partial f(\bblambda_{1:2, 1})}{\partial \lambda^{(1)}}, \ldots, \frac{\partial f(\bblambda_{1:2, K})}{\partial \lambda^{(K)}} \Big],
	\end{align}
	where $\partial f(\bblambda_{1:2, k}) / \partial \lambda^{(k)}$ is the partial derivative of $f(\bblambda)$ w.r.t. the $k$-th frequency variable $\lambda^{(k)}$ at the instantiation $\bblambda_{1:2, k}$.
\end{definition} 
The Lipschitz gradient $\nabla_L f(\bblambda_1, \bblambda_2)$ characterizes the variability of $f(\bblambda)$ between two multivariate frequencies $\bblambda_1$ and $\bblambda_2$, i.e., $f(\bblambda_1) - f(\bblambda_2) = \nabla_L f(\bblambda_1, \bblambda_2) \cdot (\bblambda_1 - \bblambda_2)$ with $\cdot$ the vector product. This is akin to the derivative of $f(\lambda)$ [cf. \eqref{eq:filterResponse1}] and allows to define integral Lipschitz AirGFs following \cite{gama2020stability}.
\begin{definition}[Integral Lipschitz AirGF]\label{def:LipschitzAirGF}
	Consider an AirGF with the filter frequency response over the air [cf. \eqref{eq:AirFrequencyResponseFilter1}] satisfying $|f(\bblambda)|\le 1$. The AirGF is integral Lipschitz if there exists a constant $C_L > 0$ such that for any two finite multivariate frequencies $\bblambda_1$ and $\bblambda_2$, it holds that 
	\begin{align}\label{eq:LipschitzGeneralizedResponse}
		\|\nabla_L f(\bblambda_1, \bblambda_2)\|_2 \le C_L~\text{and}~\|\nabla_L f(\bblambda_1, \bblambda_2) \odot \bblambda_2 \|_2 \le C_L,
	\end{align}
	where $\nabla_L f(\bblambda_1, \bblambda_2)$ is the Lipschitz gradient over $\bblambda_1$ and $\bblambda_2$ [Def \ref{def:LipschitzGradient}], and $\odot$ is the pointwise product of vectors.
\end{definition} 

\noindent An integral Lipschitz AirGF restricts the variability of its frequency response, i.e., $f(\bblambda)$ does not change faster than linear in the $K$-dimensional space and tends to change more slowly at larger frequency values. This allows to control spectral deviations caused by channel impairments in the graph spectrum. Definition \ref{def:LipschitzAirGF} generalizes the integral Lipschitz graph filter in \cite{gama2020stability} to the scenario over a sequence of random graphs $\{\bbS_{\rm air}^{(k)}\}_{k=1}^K$ with channel fading effects. 

With these preliminaries in place, we formally characterize the stability of AirGFs to channel impairments.
\begin{theorem}\label{thm:filterStability}
	Consider the AirGF $\bbH_{\rm air}(\bbS)$ [cf. \eqref{eq:AirGF}] and the graph filter $\bbH(\bbS)$ [cf. \eqref{eq:frequencyResponse}] with the same filter parameters $\{\alpha_k\}_{k=0}^K$ and the underlying graph $\bbS$. Let the channel coefficient $h$ follow a random distribution with expectation $\mu$ and standard deviation $\delta$, the noise $n$ follow a zero-mean Gaussian distribution with standard deviation $\eps$, and the filter frequency response over the air $f(\bblambda)$ be integral Lipschitz w.r.t. $C_L$ [cf. \eqref{eq:LipschitzGeneralizedResponse}]. Then, for any graph signal $\bbx$, it holds that 
	\begin{align}\label{eq:filterStability}
		&\mathbb{E}\!\big[\!\| \bbH_{\rm air}(\bbS)\bbx \!-\! \bbH(\bbS)\bbx \|^2_2\big]\!\le\! \big(C_1 (1\!-\!\mu)^2 \!+\! C_2 \delta^2\big) \|\bbx\|_2^2 \!\!+\! C_3 \eps^2 \nonumber \\
		&+\!\! \Big(\!\ccalO((1\!-\!\mu)^3) \!\!+\!\! \ccalO(\delta^2 (1\!-\!\mu)) \!\!+\! \!\ccalO\big((1\!-\!\mu) \eps^2\big)\!\!+\!\!\ccalO\big(\delta^2 \eps^2\big)\!\!\Big)\!, 
	\end{align}
	where $C_1 = (n+K-1) C_L^2$, $C_2 = ndC_L^2$ and $C_3 = n C_L^2$ are stability constants and $d$ is the maximal degree of the graph.
\end{theorem}
\begin{proof}
	See Appendix \ref{Proof:Theorem1}.
\end{proof}
Theorem \ref{thm:filterStability} states that the expected difference of the filter output caused by channel impairments is bounded proportionally by the deviation of the expectation $(1-\mu)^2$, the variance of the channel coefficient $\delta^2$, and the variance of the noise $\eps^2$. The stability bound decreases when (i) the expectation $\mu$ is close to one, i.e., the channel has less fading in expectation and does not affect the signal value significantly, 
and (ii) the variance $\delta^2$ and $\eps^2$ approach zero, i.e., the channel coefficients and the additive noise vary within a small range around their expectations. 

The stability constants $C_1$, $C_2$ and $C_3$ embed the role of the filter and the graph, i.e., the filter order $K$, the Lipschtiz constant $C_L$, the graph degree $d$, and the graph size $n$. First, all constants increase with the graph size $n$ and $C_2$ increases with the graph degree $d$. This indicates that a larger or denser graph introduces more channel impairments, results in an increased effect, and degrades the stability. Second, all constants increase with the Lipschitz constant $C_L$ and $C_1$ increases with the filter order $K$. This implies that an AirGF with a lower Lipschitz constant $C_L$ allows less variability of the filter frequency response $f(\bblambda)$ and is more stable to frequency deviations caused by channel impairments. However, this may decrease the discriminatory power of the AirGF because nearby frequencies $\bblambda_1$ and $\bblambda_2$ may yield similar responses $f(\bblambda_1)$ and $f(\bblambda_2)$, and the resulting AirGF may not identify the difference between $\bblambda_1$ and $\bblambda_2$, which degrades the performance.

\subsection{Stability of AirGNNs}

The stability of AirGNNs inherits from that of AirGFs by considering the additional effects of the nonlinearity and the multi-layer architecture. Before proceeding, we need a standard assumption for the nonlinearity.
\begin{assumption}\label{ass:LipschitzNonlinearity}
	The nonlinearity $\sigma(\cdot)$ satisfying $\sigma(0)=0$ is Lipschitz, i.e., there exists a constant $C_\sigma$ such that 
	\begin{align}\label{eq:LipschitzNonlinearity}
		|\sigma(a) - \sigma(b)|\le C_\sigma |a-b|~\forall~a, b\in \mathbb{R}.
	\end{align}
\end{assumption}
The Lipschitz nonlinearity is commonly used and examples include ReLU, absolute value and hyperbolic tangent. We then show the stability of AirGNNs in the following theorem.
\begin{theorem}\label{thm:AirGNNStability}
	Consider the AirGNN $\bbPhi_{\rm air}(\bbx, \bbS, \ccalA)$ of $L$ layers and $F$ features per layer [cf. \eqref{eq:AirGNN}]. For the same setting as in Theorem \ref{thm:filterStability}, let the nonlinearity satisfy Assumption \ref{ass:LipschitzNonlinearity} w.r.t. $C_\sigma$. Then, for any graph signal $\bbx$, it holds that 
	\begin{align}
		&\mathbb{E}\big[\| \bbPhi_{\rm air}(\bbx, \bbS, \ccalA) - \bbPhi(\bbx, \bbS, \ccalA) \|^2_2\big]\\
		&\le C \big(C_1 (1-\mu)^2 + C_2 \delta^2\big) \|\bbx\|_2^2 + C C_3 \eps^2 \nonumber \\
		&+\! \Big(\!\ccalO((1-\mu)^3) \!+\! \ccalO(\delta^2 (1\!-\!\mu)) \!+\! \ccalO\big((1-\mu) \eps^2\big)\!+\!\ccalO\big(\delta^2 \eps^2\big)\!\Big), \nonumber 
	\end{align}
	where $\bbPhi(\bbx, \bbS, \ccalA)$ is the nominal GNN output without channel impairments, $C = C_{\sigma}^{2L}L^2F^{2L-2}$ is the stability constant resulting from the nonlinearity and the multi-layer architecture, and $C_1$, $C_2$, $C_3$ are stability constants inheriting from the filter [cf. \eqref{eq:filterStability}].
\end{theorem} 
\begin{proof}
	See Appendix \ref{Proof:Theorem2}.
\end{proof}
Theorem \ref{thm:AirGNNStability} states that the AirGNN is stable to channel impairments as it follows from Theorem \ref{thm:filterStability}. The stability bound shares a similar form as that of the AirGF [cf. \eqref{eq:filterStability}] and thus, the conclusions of Theorem \ref{thm:filterStability} apply as well. Moreover, there is an additional stability constant $C$ that is composed of two terms representing the effects of the nonlinearity and the multi-layered architecture, respectively. The first term $C_{\sigma}^{2L}$ captures the amplification of the nonlinearity on the impact of channel impairments, while $C_\sigma$ is typically one in common nonlinearities, e.g., ReLU and absolute value. The second term $L^2F^{2L-2}$ is the consequence of the impaired signal propagating through multiple filters $F$ and layers $L$, which indicates that a wider and deeper AirGNN results in a looser stability bound. This is because AirGNNs with more filters and layers incorporate more channel impairments into the architecture and the latter increases the output deviation caused by these impairments.

\section{AirGNNs with CSI}\label{sec:airGNNsCSI}

From stability analysis, AirGNNs are able to maintain performance when channel impairments are mild, i.e., $\mu \to 1$ and $\eps, \sigma \to 0$. However, the performance gets degraded inevitably under substantial fading and noise. In this section, we propose a signal transmission strategy to alleviate channel effects on the AirGNN performance, in communication scenarios where CSI $\{h_{ij}\}_{ij}$ is available at transmitting nodes. 

\subsection{Truncated Channel-Inversion Power Control}\label{subsec:TCPC}

We propose to transmit signal values among neighboring nodes with \textit{truncated channel-inversion power control}. In particular, at each communication round, we adapt the transmission power of each node $p$ to the corresponding CSI $h$ subject to a long-term power constraint, i.e.,
\begin{align}\label{eq:powerConstraint}
	\mathbb{E}[|p(h)|^2] \le P,
\end{align}
where the expectation $\mathbb{E}[\cdot]$ is with respect to the distribution of random channel coefficients $h$. With CSI available at transmitting nodes, the goal is to invert channel coefficients through power control such that signal values are received without any channel fading effect. However, a brute-force approach that directly inverts channel coefficients is not possible under the power constraint \eqref{eq:powerConstraint}, because some channels may suffer from deep fading and require large powers for inversion. 

This necessitates to develop a practical channel-inversion strategy. Specifically, we allow each node to invert channel fading only if the corresponding channel coefficient exceeds a power-cutoff threshold $\gamma$, and allocate zero power otherwise, i.e.,
\begin{align}
	p (h) = 
	\begin{cases}
		\frac{\sqrt{P_0}}{h} & \text{if}~|h|^2 \ge \gamma, \\
		0 & \text{if}~|h|^2 < \gamma,
	\end{cases}
\end{align}
where $P_0$ is the scaling factor to satisfy the long-term power constraint \eqref{eq:powerConstraint}. The value of $P_0$ is determined by the constraint bound $P$ and the distribution of random channel coefficients $h$. We consider $h$ subject to a Rayleigh distribution with the scaling parameter $\beta$, i.e., the probability density function of $h$ is 
\begin{align}
	m(h) = \frac{h}{\beta^2} e^{-\frac{h^2}{2 \beta^2}},
\end{align}
and the channel gain $|h|^2$ is subject to an exponential distribution with the parameter $1/(2\beta^2)$, i.e., the probability density function of $z = |h|^2$ is
\begin{align}
	m(z) = \frac{1}{2 \beta^2} e^{-\frac{z}{2 \beta^2}}.
\end{align}
In this context, the power constraint \eqref{eq:powerConstraint} is equivalent to
\begin{align}
	P_0 \int_{\gamma}^\infty \frac{1}{2 \beta^2 z} e^{-\frac{z}{2 \beta^2}} dz = P.
\end{align} 
This allows to compute $P_0$ as
\begin{align}\label{eq:powerscalar}
	P_0 = \frac{P}{\int_{\gamma}^\infty \frac{1}{2 \beta^2 z} e^{-\frac{z}{2 \beta^2}} dz}.
\end{align}

{\linespread{1}
	\begin{algorithm}[t] \begin{algorithmic}[1]
			\STATE \textbf{Input:} Channel coefficient $h$, power-cutoff threshold $\gamma$ and scaling factor $P_0$
			\IF {$|h|^2 \ge \gamma$}
			\STATE Perform channel-inversion as $p(h) = \frac{\sqrt{P_0}}{h}$
			\ELSE 
			\STATE Assign zero power as $p(h) = 0$
			\ENDIF
		\end{algorithmic}
		\caption{Truncated Channel-Inversion Power Control}\label{alg:TCPC}
\end{algorithm}}

This strategy eliminates fading effects in admissible channels, i.e., $|h|^2 \ge \gamma$, but truncates signal transmissions in severe channels, i.e., $|h|^2 < \gamma$, referred to as truncated channel-inversion power control -- see Algorithm \ref{alg:TCPC}. While mitigating communication impact, it results in an inevitable information loss on received signals because of threshold truncation. The latter can be characterized by a non-truncation probability $\rho$ as
\begin{align}\label{eq:outageprobability}
	\rho = \text{Pr}(|h|^2 \ge \gamma) = \exp\big(-\frac{\gamma}{2 \beta^2}\big),
\end{align} 
which immediately results from the exponential distribution of the channel gain $|h|^2$. This indicates that the proposed truncated channel-inversion power control translates random channel fading effects to an unfaded binary communication link, where the signal value is transmitted without fading at a probability $\rho$ and not transmitted at a probability $1-\rho$.

By comparing \eqref{eq:powerscalar} and \eqref{eq:outageprobability}, we see that a higher threshold $\gamma$ yields a larger scaling factor $P_0$, which increases the transmission power and improves the signal-to-noise ratio (SNR). On the other hand, it leads to a lower non-truncation probability $\rho$ and results in an increased information loss. This identifies an explicit trade-off between the quality of transmitted signals and the outage probability of communication links, which is controlled by the selection of the power-cutoff threshold $\gamma$. The optimal value of $\gamma$ depends on specific problems of interest. For example, if the problem is more sensitive to the graph connectivity, a smaller $\gamma$ yields a lower outage probability and achieves a better performance; if the problem is more sensitive to the additive noise, a larger $\gamma$ leads to a higher SNR and allows for a better performance.

\subsection{Noisy Stochastic Graph Neural Networks}\label{subsec:noisySGNN}

The communication link in \eqref{eq:communicationChannel} with truncated channel-inversion power control can be re-written as
\begin{align}\label{eq:communicationChannelTruncated}
	[\bbx^{(1)}]_i \!=\! \sum_{j\in \ccalN_i}\! \sqrt{P_0} \delta_{ij} [\bbx]_j + n_{i}^{(1)},~\text{for}~i=1,...,n,
\end{align}
where $\delta_{ij}$ is a binary variable taking $1$ at a probability $\rho$ and $0$ otherwise. 
Multiplying the received signal with $1/\sqrt{P_0}$ yields 
\begin{align}\label{eq:communicationChannelTruncated1}
	[\bbx^{(1)}]_i \!=\! \sum_{j:(i,j)\in \ccalE}\! \delta_{ij} [\bbx]_j + \frac{n_{i}^{(1)}}{\sqrt{P_0}},~\text{for}~i=1,...,n.
\end{align}
This motivates to model the communication graph with the following random edge sampling (RES) model. 
\begin{definition}[Random edge sampling model]\label{def:RESModel}
	For a given graph $\ccalG = (\ccalV, \ccalE)$ and a non-truncation probability $\rho$, we define RES($\ccalG$, $\rho$) as a random graph with realizations $\ccalG^{(k)} = (\ccalV, \ccalE^{(k)})$ such that edge $(i, j) \in \ccalE$ is in $\ccalE^{(k)}$ at a probability $\rho$, i.e.,
	\begin{align}\label{eq:probability}
		\text{\rm Pr}\Big( (i,j) \in \ccalE^{(k)} \Big) = \rho,~\forall~(i,j)\in \ccalE.
	\end{align}
\end{definition}
\noindent Definition \ref{def:RESModel} allows to re-write the AirGSO in \eqref{eq:siganlShiftingChannel} as
\begin{align}\label{eq:siganlShiftingChannelTruncated}
	\bbx^{(1)} = \bbS^{(1)} \bbx + \bbn^{(1)},
\end{align}
where $\bbS^{(1)}$ is a random graph shift operator sampled from the RES($\ccalG$, $\rho$) model. The latter re-formulates the AirGF as a noisy stochastic graph filter
\begin{align}\label{eq:StochasticGF}
	&\bbH_{\rm sto} (\bbS)\bbx \!=\!\! \sum_{k=0}^K\! \alpha_k\! \prod_{\kappa=1}^k\! \bbS^{(\kappa)}\bbx \!+\!\! \sum_{k=1}^{K} \sum_{\kappa=k}^{K} \!\alpha_{\kappa}\!\! \prod_{\tau=k+1}^{\kappa}\! \bbS^{(\tau)} \bbn^{(k)}, 
\end{align}
where $\{\bbS^{(k)}\}_{k=1}^K$ are sampled from the RES($\ccalG$, $\rho$) model and $\{\bbn^{(k)}\}_{k=1}^K$ are zero-mean Gaussian noises. In this context, the AirGNN with truncated channel-inversion power control can be considered as a noisy stochastic graph neural network $\bbPhi_{\rm sto}(\bbx, \bbS, \ccalA)$, i.e., 
\begin{align}\label{eq:StoGNN}
	\bbx_{\ell}^f \!=\! \sigma\Big(\! \sum_{g=1}^{F} \bbH_{{\rm sto}, \ell}^{fg}(\bbS) &\bbx_{\ell-1}^{g} \!\Big)\!,\for~\!f\!=\!1,...,F,\ell\!=\!1,...,L,\\
	&\bbPhi_{\rm sto}(\bbx, \bbS, \ccalA) = \bbx_L^1,
\end{align}
where the threshold $\gamma$ determines the sampling probability $\rho$ of the RES model [cf. \eqref{eq:probability}] and the magnitude of the additive noise [cf. \eqref{eq:StochasticGF}] throughout the architecture.

With the aforementioned analysis, the truncated channel-inversion power control is equivalent to sampling the channel coefficient $h$ from a Bernoulli distribution with the expectation $\rho$ and variance $\rho (1-\rho)$. This allows to follow Theorem \ref{thm:AirGNNStability} to characterize the stability of AirGNNs with truncated channel-inversion power control under channel impairments. 
\begin{corollary}\label{coro:stabilityTruncated}
	Consider the AirGNN $\bbPhi_{\rm air}(\bbx, \bbS, \ccalA)$ with truncated channel-inversion power control in the same setting as Theorem \ref{thm:AirGNNStability}. Let $1-\rho$ be the link outage probability determined by the power-cutoff threshold $\gamma$ and the noise $n$ follow a zero-mean Gaussian distribution with standard deviation $\eps$ determined by the power scaling factor $P_0$ [cf. \eqref{eq:powerscalar} and \eqref{eq:communicationChannelTruncated1}]. Then, for any graph signal $\bbx$, it holds that
	\begin{align}\label{eq:stabilityCSI}
		&\mathbb{E}\big[\| \bbPhi_{\rm air}(\bbx, \bbS, \ccalA) - \bbPhi(\bbx, \bbS, \ccalA) \|^2_2\big]\\
		&\le\! \tilde{C}_1 \|\bbx\|_2^2 (1-\rho) \!+\! \tilde{C_2} \eps^2 \nonumber \!+\! \ccalO((1-\rho)^2) \!+\! \ccalO\big((1-\rho) \eps^2\big), \nonumber 
	\end{align}
	where $\tilde{C}_1 = C_{\sigma}^{2L}L^2F^{2L-2} n d C_L^2 \rho$ and $\tilde{C}_2 = C_{\sigma}^{2L}L^2F^{2L-2} n C_L^2$ are stability constants. 
\end{corollary}
Corollary \ref{coro:stabilityTruncated} states that AirGNNs with truncated channel-inversion power control are stable to channel impairments as well. The stability bound demonstrates an explicit trade-off between the link outage probability $\rho$ and the noise magnitude $\eps$ through the power-cutoff threshold $\gamma$. Specifically, a smaller $\gamma$ yields a lower outage probability $1-\rho$ and decreases the first term $\tilde{C}_1 \|\bbx\|_2^2 (1-\rho)$ in \eqref{eq:stabilityCSI}. However, it reduces the power scaling factor $P_0$ and results in a larger noise magnitude $\eps$, which increases the second term $\tilde{C_2} \eps^2$.

\subsection{Training Procedure}\label{subsec:trainingwCSI}

The truncated channel-inversion power control does not recover the ideal communication 
because it truncates the transmission when the channel coefficient $h$ is smaller than the power-cutoff threshold $\gamma$. The resulting AirGNN output is a random variable w.r.t. the non-truncation probability $\rho$ and Gaussian noise $\bbn$ [cf. \eqref{eq:StochasticGF}], which may still suffer from performance degradation. The latter makes it unclear how to train the AirGNN to further mitigate this performance degradation.

We propose to incorporate the randomness of the AirGNN output into the training procedure, and develop a stochastic gradient descent (SGD) based method to train the AirGNN. Specifically, given a task with the training set $\ccalR = \{ (\bbx_r, \bby_r) \}_{r=1}^R$ and loss $\ell(\cdot)$, we define the objective as the expected loss over the training set
\begin{equation} \label{eq:objectiveFunctionSto}
	\ccalL_{\rm sto}(\mathcal{R},\bbS,\ccalA) = \frac{1}{R} \sum_{r=1}^R \ell(\bby_r, \bbPhi_{\rm sto}(\bbx_r, \bbS, \ccalA)).
\end{equation}
Since $\bbPhi_{\rm sto}(\bbx_r, \bbS, \ccalA)$ is random [cf. \eqref{eq:StoGNN}], the objective function $\ccalL_{\rm sto}(\mathcal{R},\bbS,\ccalA)$ is random w.r.t. the non-truncation probability and Gaussian noise. The goal is to find the optimal architecture parameters $\ccalA^*$ that minimize $\ccalL_{\rm sto}(\mathcal{R},\bbS,\ccalA)$ in expectation. 

Define $\bbPhi_{\rm sto}(\bbx, \bbS, \ccalA | \{\bbS^{(k)}\}_k)$ as a deterministic realization of the AirGNN output, where $\{\bbS^{(k)}\}_k$ is a sequence of graph shift operators sampled from the RES model with probability $\rho$ [Def. \ref{def:RESModel}] for all AirGFs throughout the architecture. The training procedure contains successive iterations, where each iteration $t$ consists of a forward and a backward phase. In the forward phase, it samples graph shift operators $\{\bbS^{(k)}_t\}_k$ and generates a deterministic AirGNN architecture $\bbPhi_{\rm sto}(\cdot, \bbS, \ccalA_t | \{\bbS^{(k)}_t\}_k)$ with parameters $\ccalA_t$. The latter processes a random set of the training data $\ccalR_t \subseteq \ccalR$ to approximate the objective function as
\begin{equation}\label{eq:stochasticObjectivewithCSI}
	\ccalL\big(\mathcal{R}_t,\!\bbS,\!\ccalA_t | \{\bbS^{(k)}_t\}_k\!\big) \!=\! \frac{1}{|\ccalR_t|}\! \sum_{r=1}^{|\ccalR_t|} \! \ell(\bby_r,\! \bbPhi_{\rm sto}(\!\bbx_r,\! \bbS,\! \ccalA_t | \{\bbS^{(k)}_t\!\}_k\!)\!\big)\!,
\end{equation}
where $(\bbx_r,\bby_r) \in \ccalR_t$ and $|\ccalR_t|$ is the number of data samples in $\ccalR_t$. In the backward phase, the parameters $\ccalA_t$ are updated with SGD 
\begin{align}\label{eq:gradientUpdatewithCSI}
	\ccalA_{t+1} = \ccalA_t - \eta_t \nabla_\ccalA \ccalL\big(\mathcal{R}_t,\!\bbS,\!\ccalA_t | \{\bbS^{(k)}_t\}_k\big),
\end{align}
where $\eta_t$ is the step-size. It accounts for the effect of the transmission truncation by updating the parameters $\ccalA_t$ with a stochastic gradient $\nabla_\ccalA \ccalL(\mathcal{R}_t,\!\bbS,\!\ccalA_t | \{\bbS^{(k)}_t\}_k)$ at each iteration $t$. The stochasticity results from the randomness of the RES model $\text{RES}(\ccalG, \rho)$, the Gaussian noise $\bbn_t$, as well as the randomly chosen data samples $\ccalR_t$. 

Algorithm \ref{alg:trainingProcedurewithCSI} summarizes the proposed training procedure. It incorporates the graph stochasticity (caused by transmission truncation) and the noise randomness during training, to match those encountered during testing. Each node will rely on the truncated noisy information from its neighbors with some uncertainty, and the trained parameters will be more robust to these random perturbations during inference, yielding an improved performance and a robust transference for decentralized tasks.

\section{AirGNNs without CSI} \label{sec:AirGNNNoCSI}

{\linespread{1}
	\begin{algorithm}[t] \begin{algorithmic}[1]
			\STATE \textbf{Input:} training set $\mathcal{R}$, loss function $\ell$, initial parameters $\ccalA_0$, power-cutoff threshold $\gamma$, and power scaling factor $P_0$
			\FOR {$t=1,\dots,T$}
			\STATE Compute the non-truncation probability $\rho$ [cf. \eqref{eq:outageprobability}]
			\STATE Sample graph shift operators $\{\bbS_t^{(k)}\}_k$ from the RES model with the probability $\rho$ [Def. \ref{def:RESModel}] 
			\STATE Determine the corresponding AirGNN architecture $\bbPhi_{\rm sto}(\cdot, \bbS, \ccalA_t | \{\bbS^{(k)}_t\}_k)$
			\STATE Sample data $\ccalR_t \subseteq \ccalR$ and compute AirGNN outputs $\{\bbPhi_{\rm sto}(\bbx_r, \bbS, \ccalA_t | \{\bbS^{(k)}_t\}_k)\}_{r=1}^{|\ccalR_t|}$ 
			\STATE Compute the objective $\ccalL(\mathcal{R}_t,\bbS,\ccalA | \{\bbS^{(k)}_t\}_k)$ as in \eqref{eq:stochasticObjectivewithCSI}
			\STATE Compute the stochastic gradient $\nabla_\ccalA \ccalL(\mathcal{R}_t,\bbS,\ccalA | \{\bbS^{(k)}_t\}_k)$ 
			\STATE Update parameters $\ccalA_t$ with step-size $\eta_t$ as in \eqref{eq:gradientUpdatewithCSI}
			\ENDFOR
		\end{algorithmic}
		\caption{Training Procedure of AirGNN with CSI}\label{alg:trainingProcedurewithCSI}
\end{algorithm}}

In this section, we consider communication scenarios where CSI $\{h_{ij}\}_{ij}$ is unknown at transmitting nodes, such that the signal transmission strategy developed in Section \ref{sec:airGNNsCSI} cannot be applied to handle channel impairments. To overcome this issue, we propose to consider the AirGNN output directly as a random variable and conduct training by accounting for the distribution over the entire CSI space. 

We start by defining a realization of the AirGNN output over the graph signal $\bbx$ as $\bbPhi_{\rm air}(\bbx, \bbS, \ccalA | \bbh, \bbn)$, where $\bbh$ and $\bbn$ are samples of the CSI and the Gaussian noise, respectively, throughout graph shift operations of all AirGFs in the architecture. We then define a realization of the objective function over the AirGNN output $\bbPhi_{\rm air}(\bbx, \bbS, \ccalA | \bbh, \bbn)$ as
\begin{align}\label{eq:AirObjective}
	&\ccalL\big(\tilde{\ccalR},\!\bbS,\!\ccalA_t | \bbh, \bbn \big) = \frac{1}{|\tilde{\ccalR}|} \sum_{r=1}^{|\tilde{\ccalR}|} \! \ell(\bby_r, \bbPhi_{\rm air}(\bbx_r,\! \bbS, \ccalA | \bbh, \bbn)\big),
\end{align}
where $\tilde{\ccalR}$ is a subset of training data sampled from $\ccalR$ and $|\tilde{\ccalR}|$ is the number of samples in $\tilde{\ccalR}$. With these preliminaries in place, we develop the training procedure for AirGNNs without CSI. 

We similarly perform successive iterations, where each iteration $t$ contains a forward and a backward phase. The forward phase samples the CSI $\bbh_t$ and the Gaussian noise $\bbn_t$ for a deterministic AirGNN realization $\bbPhi_{\rm air}(\cdot, \bbS, \ccalA_t | \bbh_t, \bbn_t)$ with parameters $\ccalA_t$, which processes a random set of the training data $\ccalR_t \subseteq \ccalR$ to approximate the objective function $\ccalL(\mathcal{R}_t,\!\bbS,\!\ccalA_t | \bbh_t,\! \bbn_t)$ as in \eqref{eq:AirObjective}. The backward phase updates the parameters $\ccalA_t$ with the gradient of the approximated objective function as 
\begin{align}\label{eq:AirgradientUpdate}
	\ccalA_{t+1} = \ccalA_t - \eta_t \nabla_\ccalA \ccalL(\mathcal{R}_t,\bbS,\ccalA_t | \bbh_t, \bbn_t),
\end{align}
where $\eta_t$ is the step-size. The update in \eqref{eq:AirgradientUpdate} takes the effects of channel impairments into account, i.e., it updates the parameters $\ccalA_t$ with a stochastic gradient $\nabla_\ccalA \ccalL(\mathcal{R}_t,\bbS,\ccalA_t | \bbh_t, \bbn_t)$ w.r.t. the randomness of the CSI $\bbh_t$, the Gaussian noise $\bbn_t$, as well as the data subset $\ccalR_t$. Algorithm \ref{alg:trainingProcedurewithoutCSI} summarizes the above training procedure. 

We see that \eqref{eq:AirObjective}-\eqref{eq:AirgradientUpdate} incorporate random channel impairments directly into the training procedure, where each node relies on faded and noisy information from its neighbors [cf. \eqref{eq:communicationChannel1}]. This matches the distributions of channel fading and Gaussian noise encountered in the decentralized implementation at the inference phase, and renders the trained parameters more robust to communication perturbations. The latter improves the performance of AirGNNs without requiring any channel state information for signal transmission. 

{\linespread{1}
	\begin{algorithm}[t] \begin{algorithmic}[1]
			\STATE \textbf{Input:} training set $\mathcal{R}$, loss function $\ell$, initial parameters $\ccalA_0$
			\FOR {$t=1,\dots,T$}
			\STATE Sample CSI $\bbh_t$ and Gaussian noise $\bbn_t$
			\STATE Determine the corresponding AirGNN architecture $\bbPhi_{\rm air}(\cdot, \bbS, \ccalA_t | \bbh_t, \bbn_t)$
			\STATE Sample data $\ccalR_t \subseteq \ccalR$ and compute AirGNN outputs $\{\bbPhi_{\rm air}(\bbx_r, \bbS, \ccalA_t | \bbh_t, \bbn_t)\}_{r=1}^{|\ccalR_t|}$ 
			\STATE Compute the objective $\ccalL(\mathcal{R}_t,\bbS,\ccalA_t | \bbh_t, \bbn_t)$ as in \eqref{eq:AirObjective}
			\STATE Compute the stochastic gradient $\nabla_\ccalA \ccalL(\mathcal{R}_t,\bbS,\ccalA_t | \bbh_t, \bbn_t)$
			\STATE Update parameters $\ccalA_t$ with step-size $\eta_t$ as in \eqref{eq:AirgradientUpdate}
			\ENDFOR
		\end{algorithmic}
		\caption{Training Procedure of AirGNN without CSI}\label{alg:trainingProcedurewithoutCSI}
\end{algorithm}}

\smallskip
\noindent \textbf{Discussion.} The training procedure without CSI [Algorithm \ref{alg:trainingProcedurewithoutCSI}] resembles that with CSI [Algorithm \ref{alg:trainingProcedurewithCSI}], while the difference lies in two aspects: (i) the implementation condition and (ii) the gradient stochasticity. First, the training procedure without CSI does not require specific channel coefficients $h_{ij}$ but only the general distribution of $h_{ij}$ for implementation, while the training procedure with CSI requires specific channel coefficients $h_{ij}$ for signal transmission. Second, the gradient stochasticity of the training procedure without CSI is w.r.t. the random CSI and Gaussian noise [cf. \eqref{eq:communicationChannel1}], where the magnitude of the stochasticity depends on the variance of the CSI and noise. On the other hand, the gradient stochasticity of the training procedure with CSI is w.r.t. the random communication truncation (link loss) and Gaussian noise [Def. \ref{def:RESModel}]. The magnitude of the stochasticity depends on the truncation probability (link outage probability) $1-\rho$ and the noise variance, which are, in turn, determined by the cutoff threshold $\gamma$ of the truncated channel-inversion power control [cf. \eqref{eq:powerscalar}-\eqref{eq:outageprobability}].

It is important noting that while the AirGNN with CSI requires more information, i.e., the specific CSI $h_{ij}$ at transmitting nodes, it does not necessarily outperform the AirGNN without CSI in all scenarios. In particular, the specific performance comparison depends on the distribution of the CSI, the magnitude of the noise, and the selection of the threshold $\gamma$. For example in the scenario where the performance is sensitive to the Gaussian noise $\bbn$ but not sensitive to the graph topology or channel fading $\bbh$, assume that the AirGNN with CSI selects a small threshold $\gamma$ for channel inversion. While capable of inverting almost all channel fading, it will result in a small transmission power $P_0$ [cf. \eqref{eq:powerscalar}] and a large noise (a low SNR), and may perform worse than the AirGNN without CSI -- see detailed comparisons in experiments of Section \ref{sec:experiments}.

\section{Convergence and Variance}\label{sec:convergence}

The training procedures developed in Section \ref{subsec:trainingwCSI} and Section \ref{sec:AirGNNNoCSI} incorporate the randomness of the AirGNN output, which makes it unclear if these training procedures converge, what solution they search for, and how the trained model behaves. In this section, we analyze the convergence of the training procedures and characterize the variance of the trained model, to provide interpretations for the obtained solution.

\subsection{Convergence Analysis}

We begin by considering the following stochastic optimization problems 
\begin{align}\label{eq:stochasticOptimization}
	\min_{\ccalA}\bar{\ccalL}(\ccalA) &= \min_{\ccalA} \mathbb{E}_{\rho, \bbn} [\ccalL(\mathcal{R},\bbS,\ccalA | \{\bbS^{(k)}\}_{k})],\\\label{eq:stochasticOptimization1}
	\min_{\ccalA}\bar{\ccalL}(\ccalA) &= \min_{\ccalA} \mathbb{E}_{\bbh, \bbn} [\ccalL(\mathcal{R},\bbS,\ccalA | \bbh, \bbn)],
\end{align}
where the expectation $\mathbb{E}_{\rho, \bbn}[\cdot]$ in \eqref{eq:stochasticOptimization} is w.r.t. the randomness of the transmission truncation probability $1-\rho$ and the noise $\bbn$ [cf. \eqref{eq:powerscalar}-\eqref{eq:outageprobability}], and $\mathbb{E}_{\bbh, \bbn}[\cdot]$ in \eqref{eq:stochasticOptimization1} is w.r.t. the randomness of the CSI $\bbh$ and the noise $\bbn$ [cf. \eqref{eq:communicationChannel1}]. Problems \eqref{eq:stochasticOptimization} and \eqref{eq:stochasticOptimization1} differ from traditional stochastic optimization problems in that they consider the communication randomness as the problem stochasticity not only the data distribution. The standard method to solve these problems is the SGD, which approximates the true gradient with the gradient of some random sample and leverages the latter to update model parameters -- see Algorithm \ref{alg:SGD}. The following theorem relates the proposed training procedures in Section \ref{subsec:trainingwCSI} and Section \ref{sec:AirGNNNoCSI} to the stochastic optimization problems \eqref{eq:stochasticOptimization} and \eqref{eq:stochasticOptimization1}.
\begin{theorem}\label{thm:equivalence}
	Performing the proposed training procedure on the AirGNN model [Algorithm \ref{alg:trainingProcedurewithCSI} or Algorithm \ref{alg:trainingProcedurewithoutCSI}] is equivalent to running the corresponding SGD on the stochastic optimization problem \eqref{eq:stochasticOptimization} or \eqref{eq:stochasticOptimization1} [Algorithm \ref{alg:SGD}]. 
\end{theorem}
\begin{proof}
	See Appendix \ref{Proof:Theorem3}.
\end{proof}
Theorem \ref{thm:equivalence} provides interpretations for the proposed training procedure, i.e., its goal is to search the solution of the associated stochastic optimization problem \eqref{eq:stochasticOptimization} or \eqref{eq:stochasticOptimization1} depending on the incorporated randomness during training. Moreover, the result indicates that we can show the convergence of the proposed training procedures by proving the convergence of the SGD on problems \eqref{eq:stochasticOptimization} and \eqref{eq:stochasticOptimization1}. Since AirGNN is a nonlinear parameterization, these problems are typically non-convex. This motivates us to consider the convergence criterion as the gradient norm $\| \nabla_\ccalA \bar{\ccalL}(\ccalA) \|^2_2$, which is commonly used to quantify the first-order stationarity in the non-convex setting. To proceed, we need the following assumptions. 

\begin{assumption} \label{As:1}
	The gradient of the expected objective function $\bar{\ccalL}(\ccalA)$ in problem \eqref{eq:stochasticOptimization} (and \eqref{eq:stochasticOptimization1}) is Lipschitz continuous, i.e., there exists a constant $C_\ell$ such that
	\begin{equation} \label{eq:As1}
		\begin{split}
			\| \nabla_\ccalA \bar{\ccalL}(\ccalA_1) \!-\! \nabla_\ccalA \bar{\ccalL}(\ccalA_2) \|_2 \!\le\! C_\ell \| \ccalA_1\!-\!\ccalA_2 \|_2,~\forall~\ccalA_1, \ccalA_2.
		\end{split}
	\end{equation}
\end{assumption}

\begin{assumption}\label{As:2}
	The gradient of the objective function $\ccalL(\mathcal{R},\bbS,\ccalA | \{\bbS^{(k)}\}_{k})$ in problem \eqref{eq:stochasticOptimization} (and $\ccalL(\mathcal{R},\bbS,\ccalA | \bbh, \bbn)$ in problem \eqref{eq:stochasticOptimization1}) is bounded, i.e., there exists a constant $C_g$ such that
	\begin{align}
		&\| \nabla_\ccalA \ccalL(\mathcal{R},\bbS,\ccalA | \{\bbS^{(k)}\}_{k}) \|_2 \le C_g, \\
		&\| \nabla_\ccalA \ccalL(\mathcal{R},\bbS,\ccalA | \bbh, \bbn) \|_2 \le C_g.
	\end{align}
\end{assumption}

Assumptions \ref{As:1}-\ref{As:2} are standard in optimization theory \cite{pardalos1994nonconvex,jongen2007optimization,gao2022balancing}, which provide a handle to deal with the gradient stochasticity during training. Theorem \ref{thm:convergence} characterizes the convergence of the proposed training procedures.

{\linespread{1}
	\begin{algorithm}[t]  \begin{algorithmic}[1]
			\STATE \textbf{Input:} training set $\mathcal{R}$, loss function $\ell$, initial parameters $\ccalA_0$
			\STATE Set batch-size of training data $\ccalR$ as $|\ccalR_t|$ and batch-size of graph shift operators $\{\bbS^{(k)}\}_{k}$ for \eqref{eq:stochasticOptimization} or communication channels $(\bbh,\bbn)$ for \eqref{eq:stochasticOptimization1} as $1$
			\FOR {$t = 1,\dots,T$}
			\STATE Sample a random objective function $\ccalL(\mathcal{R}_t,\bbS,\ccalA_t | \{\bbS_t^{(k)}\}_{k})$ 
			for \eqref{eq:stochasticOptimization} or $\ccalL(\mathcal{R}_t,\bbS,\ccalA_t | \bbh_t, \bbn_t)$ 
			for \eqref{eq:stochasticOptimization1}
			\STATE Compute the stochastic gradient $\nabla_\ccalA \ccalL(\mathcal{R}_t,\bbS,\ccalA_t | \{\bbS_t^{(k)}\}_{k})$ for \eqref{eq:stochasticOptimization} or $\nabla_\ccalA \ccalL(\mathcal{R}_t,\bbS,\ccalA_t | \bbh_t, \bbn_t)$ for \eqref{eq:stochasticOptimization1}
			\STATE Update the model parameters with step-size $\eta_t$ as $\ccalA_{t+1} = \ccalA_t - \eta_t \nabla_\ccalA \ccalL(\mathcal{R}_t,\bbS,\ccalA_t | \{\bbS^{(k)}\}_{k})$ for \eqref{eq:stochasticOptimization} or 
			$\ccalA_{t+1} = \ccalA_t - \eta_t \nabla_\ccalA \ccalL(\mathcal{R}_t,\bbS,\ccalA_t | \bbh_t, \bbn_t)$ for \eqref{eq:stochasticOptimization1}
			\ENDFOR
		\end{algorithmic}
		\caption{SGD for Problems \eqref{eq:stochasticOptimization} and \eqref{eq:stochasticOptimization1}}\label{alg:SGD}
\end{algorithm}}

\begin{theorem} \label{thm:convergence}
	Consider the AirGNN in \eqref{eq:AirGNN} with the training procedure in Algorithm \ref{alg:trainingProcedurewithCSI} or Algorithm \ref{alg:trainingProcedurewithoutCSI} and the stochastic optimization problem \eqref{eq:stochasticOptimization} or \eqref{eq:stochasticOptimization1} with the objective function satisfying Assumptions \ref{As:1}-\ref{As:2} w.r.t. $C_\ell$ and $C_g$. Let $T$ be the total number of training iterations and $\ccalA^*$ be the global optimal solution of problem \eqref{eq:stochasticOptimization} or \eqref{eq:stochasticOptimization1}. Then, for any initial parameters $\ccalA_0$ and the step-size
	\begin{equation} \label{eq:step-size}
		\begin{split}
			\eta_t = \eta = \sqrt{\frac{2\left( \bar{\ccalL}(\ccalA_0) -\bar{\ccalL}(\ccalA^*) \right)}{T C_\ell C_g^2 }},
		\end{split}
	\end{equation}
	it holds that
	\begin{equation}
		\begin{split}
			\min_{0\le t \le T-1} \mathbb{E} \left[ \| \nabla_\ccalA \bar{\ccalL}(\ccalA_t) \|^2_2 \right] \le \frac{C}{\sqrt{T}},
		\end{split}
	\end{equation}
	where $C\!=\! \sqrt{2\left( \bar{\ccalL}(\ccalA_0) \!-\!\bar{\ccalL}(\ccalA^*) \right)C_\ell}C_g$ is a constant. 
\end{theorem}
\begin{proof}
	See Appendix \ref{Proof:Theorem4}.
\end{proof}
Theorem \ref{thm:convergence} states that the proposed AirGNN training procedure minimizes the stochastic optimization problem \eqref{eq:stochasticOptimization} or \eqref{eq:stochasticOptimization1} and the architecture parameters converge to a stationary solution with a rate on the order of $\ccalO(1/\sqrt{T})$. The result characterizes the converging behavior and interprets the convergent solution, which validates the effectiveness of the proposed training procedure. The step-size $\eta_t$ in \eqref{eq:step-size} depends on the total number of iterations $T$, while it is commonly selected as $\eta_t \propto 1/t$ or $1/\sqrt{t}$ in practice. It is worth mentioning that Theorem \ref{thm:convergence} guarantees local convergence because of the problem non-convexity, i.e., the objective function converges to a local stationary minimum rather than a global one. The latter can be improved by training the AirGNN multiple times and selecting the best solution.

\subsection{Variance Analysis}

As shown in Theorem \ref{thm:convergence}, the training procedures optimize the performance of AirGNNs in expectation [cf. \eqref{eq:stochasticOptimization}-\eqref{eq:stochasticOptimization1}], while it says little about the random behavior of the trained model around the optimized mean. We characterize the latter by analyzing the variance of the AirGNN output. For our analysis, we consider the variance over all nodes
\begin{align}
	{\rm var}[\bbPhi_{\rm air}(\bbx, \bbS, \ccalA^*)] = \sum_{i=1}^n {\rm var}\big[[\bbPhi_{\rm air}(\bbx, \bbS, \ccalA^*)]_i\big],
\end{align}
where $[\bbPhi_{\rm air}(\bbx, \bbS, \ccalA^*)]_i$ is the $i$th entry of $\bbPhi_{\rm air}(\bbx, \bbS, \ccalA^*)$ and $\ccalA^*$ is the parameters trained by Algorithm \ref{alg:trainingProcedurewithCSI} or \ref{alg:trainingProcedurewithoutCSI}. This metric characterizes how individual entries $\{[\bbPhi_{\rm air}(\bbx, \bbS, \ccalA^*)]_i\}_{i=1}^n$ deviate from their expectations, which is typical to measure the variance in a multivariate stochastic system \cite{joshi2008sensor}. 

We follow Section \ref{sec:stabilityAnalysis} to conduct the variance analysis in the graph spectral domain for establishing a graph-universal result. Specifically, we consider the integral Lipschitz AirGF [Def. \ref{def:LipschitzAirGF}] and require an additional assumption for the nonlinearity. 

\begin{assumption} \label{assumptionNonlinearVariance}
	The nonlinearity $\sigma(\cdot)$ satisfying $\sigma(0)\!=\!0$ is variance non-increasing, i.e., for any $x \in \mathbb{R}$ it holds that 
	\begin{equation}\label{eq:assumptionNonlinearVariance}
		\begin{split}
			{\rm var}[\sigma(x)] \le {\rm var}[x].
		\end{split}
	\end{equation}
\end{assumption}

\noindent Assumption \ref{assumptionNonlinearVariance} is mild because the nonlinearity is typically non-expansive, and has been proved for ReLU and absolute value. We now show that the variance of the AirGNN output is upper bounded by factors that are proportional to the variance of the channel coefficient $h$ and the Gaussian noise $n$. 

\begin{figure*}%
	\centering
	\begin{subfigure}{0.6\columnwidth}
		\includegraphics[width=1.05\linewidth, height = 0.76\linewidth, trim = {0cm 0cm 0cm 1cm}, clip]{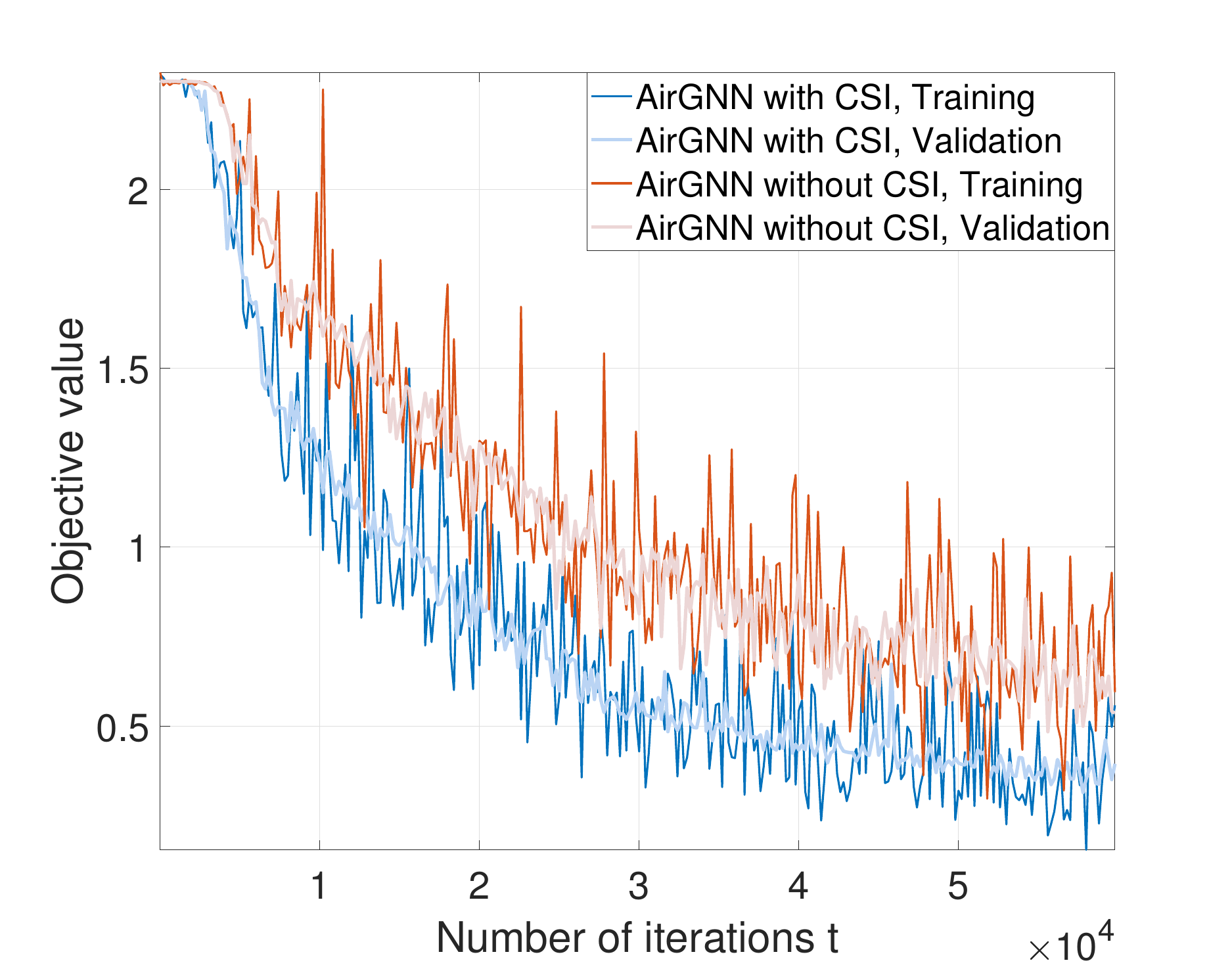}%
		\caption{}%
		\label{subfig:convergence}%
	\end{subfigure}\hfill\hfill%
	\begin{subfigure}{0.6\columnwidth}
		\includegraphics[width=1.05\linewidth,height = 0.76\linewidth, trim = {0cm 0cm 0cm 1cm}, clip]{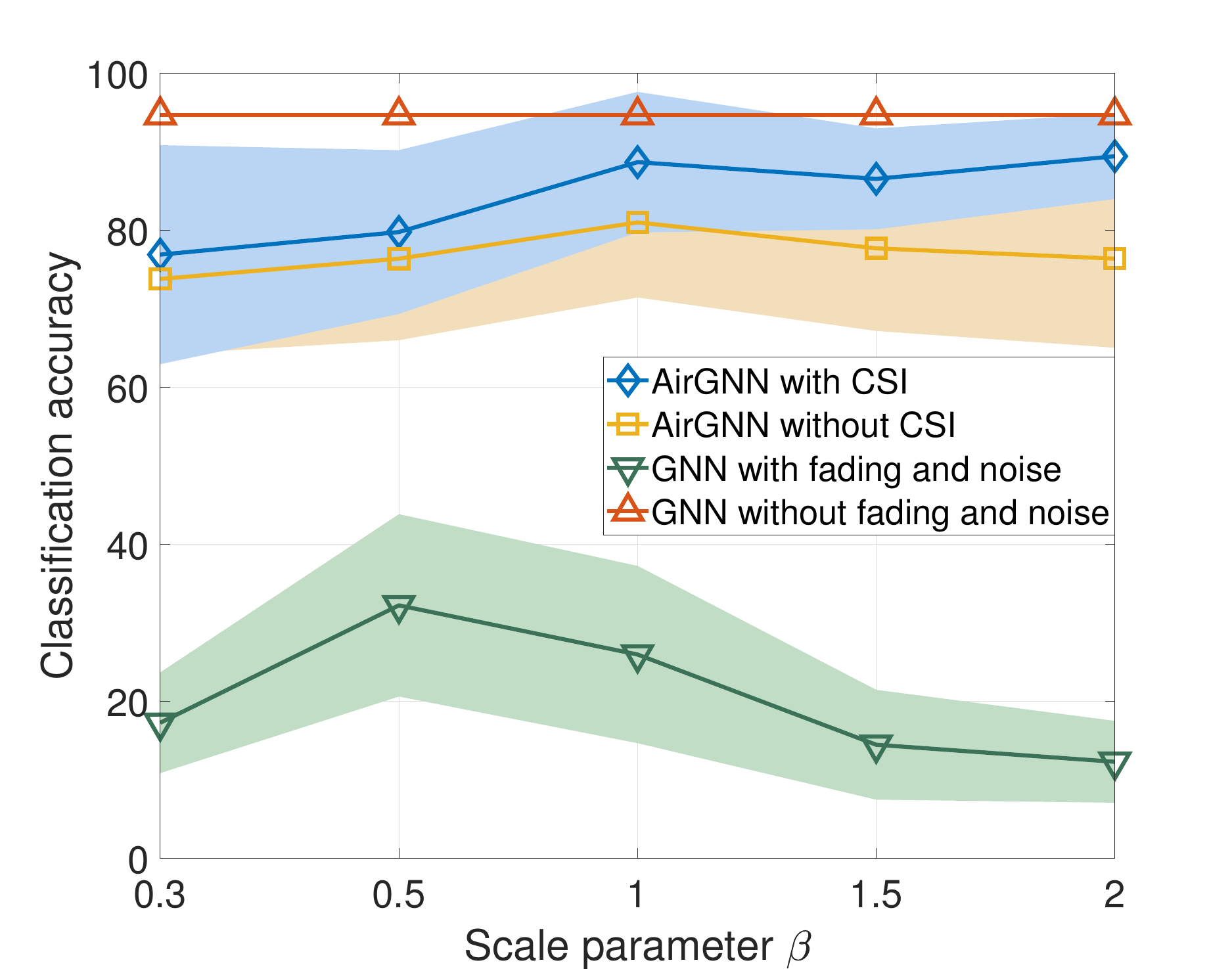}%
		\caption{}%
		\label{subfig:sourceSigma}%
	\end{subfigure}\hfill\hfill%
	\begin{subfigure}{0.6\columnwidth}
		\includegraphics[width=1.05\linewidth,height = 0.76\linewidth, trim = {0cm 0cm 0cm 1cm}, clip]{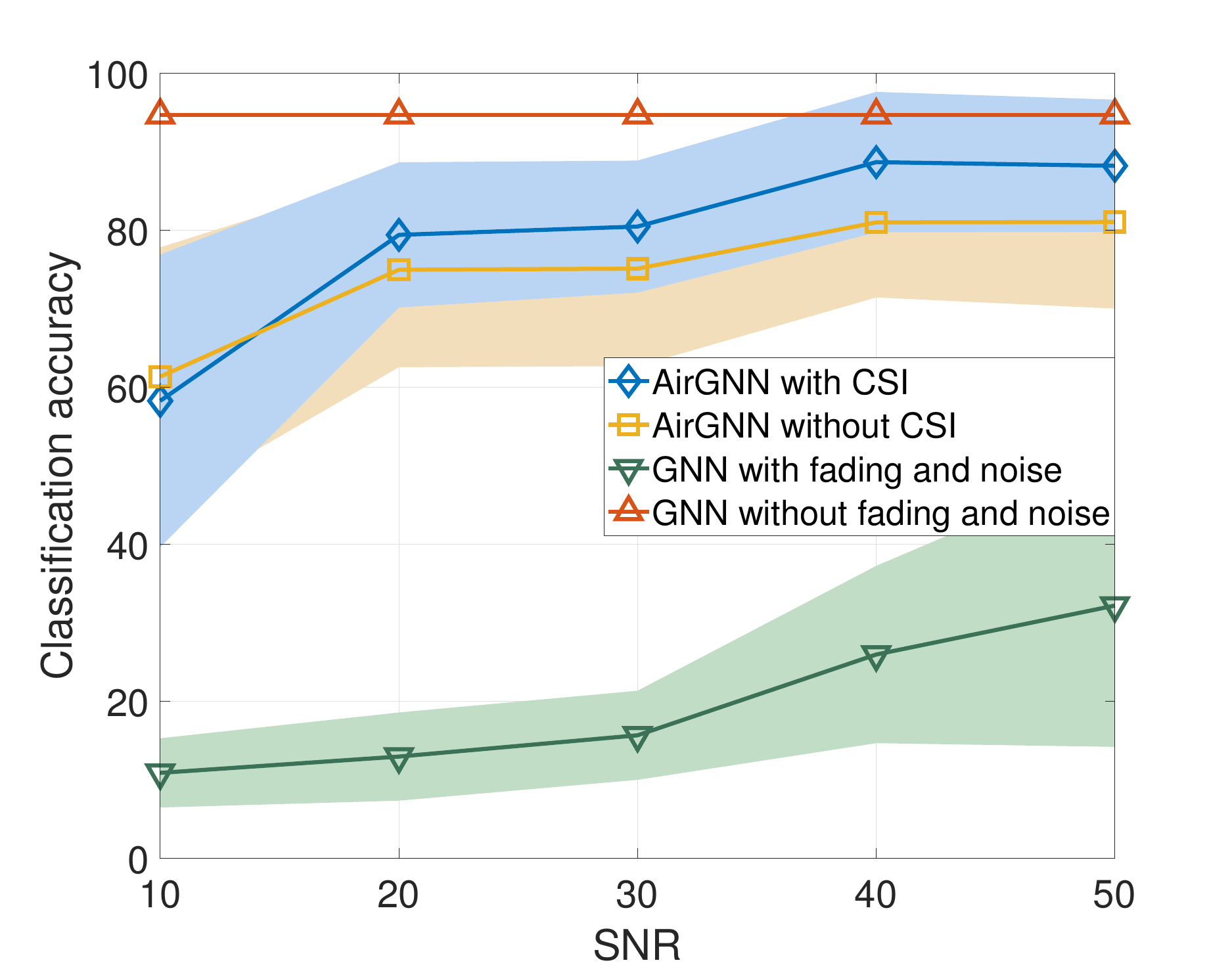}%
		\caption{}%
		\label{subfig:sourceSNR}%
	\end{subfigure}%
	\caption{(a) Training procedures of the AirGNNs with and without CSI for $\beta=1$. (b) Performance of the AirGNN and baselines in communication scenarios with different scale parameters $\beta$ for source localization. (c) Performance of the AirGNN and baselines in communication scenarios with different SNRs for source localization.}\label{fig:sourcePerformance}\vspace{-4mm}
\end{figure*}

\begin{theorem}\label{thm:variance}
	Consider the AirGNN $\bbPhi_{\rm air}(\bbx, \bbS, \ccalA^*)$ of $L$ layers and $F$ features per layer with the parameters $\ccalA^*$ and the underlying graph $\bbS$ [cf. \eqref{eq:AirGNN}]. Let filters be integral Lipschitz w.r.t. $C_L$ [cf. \eqref{eq:LipschitzGeneralizedResponse}], and the nonlinearity satisfy Assumption \ref{ass:LipschitzNonlinearity} w.r.t. $C_\sigma$ and Assumption \ref{assumptionNonlinearVariance}. For communication scenarios with CSI available at transmitting nodes [Section \ref{sec:airGNNsCSI}], it holds for any graph signal $\bbx$ that 
	\begin{align} \label{eq:varianceCSI}
		{\rm var} [ \bbPhi_{\rm air}(\bbx, \bbS, \ccalA^*) ]&\le\! C_1 \| \bbx \|_2^2 \rho(1\!-\!\rho) \!+\! C_2 \eps^2
		\\&+\! \ccalO(\rho^2(1\!-\!\rho)^2) + \ccalO(\rho(1-\rho) \eps^2), \nonumber
	\end{align}
	where $1-\rho$ is the link outage (truncation) probability determined by the power cut-off threshold $\gamma$ [cf. \eqref{eq:outageprobability}], $\eps$ is the standard deviation of the Guassian noise $n$ determined by the power scaling factor $P_0$ [cf. \eqref{eq:powerscalar}], and $C_1 = n d C_L^2 \!\sum_{\ell\!=\!1}^L\!\! F^{2L\!-\!3} C_\sigma^{2\ell\!-\!2}$ and $C_2 = n C_L^2$ are variance constants with $d$ the maximal degree of the graph. For communication scenarios with CSI unknown at transmitting nodes [Section \ref{sec:AirGNNNoCSI}], it holds that
	\begin{align} \label{eq:varianceNoCSI}
		&{\rm var} [ \bbPhi_{\rm air}(\bbx,\! \bbS,\! \ccalA^*) ] \!\le\! C_1 \| \bbx \|_2^2 \delta^2 \!+\! C_2 \eps^2 \!+\! \ccalO(\delta^4) \!+\! \ccalO(\delta^2 \eps^2), 
	\end{align}
	where $\sigma$ and $\eps$ are standard deviations of the channel coefficient $h$ and the Gaussian noise $n$, respectively, and $C_1 = n d C_L^2 \sum_{\ell\!=\!1}^L\!\! F^{2L\!-\!3} C_\sigma^{2\ell\!-\!2}$ and $C_2 = n C_L^2$ are variance constants. 
\end{theorem}
\begin{proof}
	See Appendix \ref{Proof:Theorem5}.
\end{proof}

Theorem \ref{thm:variance} states that the output of the trained AirGNN, either with or without CSI, does not diverge but changes within a finite range. The variance bound depends on the variance of the channel coefficient $\rho(1-\rho)$ (with CSI) or $\sigma^2$ (without CSI) and the variance of the additive noise $\eps^2$. When communication channels are stable, i.e., $\rho \to 0 \setminus 1$ or $\sigma \to 0$, and the noise is minor, i.e., $\eps \to 0$, the variance is small and the AirGNN output varies close to the optimized expectation. For ideal communications $\rho = 1$ or $\sigma = 0$ and $\eps = 0$, the bound reduces to zero because there is no channel impairment and all graph shift operators are deterministic.

The variance constants $C_1$ and $C_2$ consist of four terms that represent the role of the graph, filter and architecture hyperparameters, respectively. The terms $d$ and $n$ capture the graph impact, which shows that a denser and larger graph results in a worse variance bound. The term $C_L^2$ captures the filter impact, which depends on the Lipschitz property of the filter frequency response over the air $f(\bblambda)$. While a smaller $C_L$ yields a lower variance with more stable performance, it allows less variability between nearby frequencies and may decrease the discriminatory power of the AirGNN. The term $\sum_{\ell=1}^L \!F^{2L-3} C_\sigma^{2\ell-2}$ captures the architecture impact, which is the consequence of the signal propagating through the nonlinearity and multiple layers. As in the stability analysis, a wider architecture with more features and a deeper architecture with more layers lead to a larger variance. This is because more AirGFs come with more channel effects, which increase the AirGNN output randomness. The aforementioned factors also identify our potential handle to design AirGNN architectures with lower variance.

\section{Experiments}\label{sec:experiments}

We evaluate AirGNNs on the decentralized tasks of source localization (Sec. \ref{subsec:source}) and multi-robot flocking (Sec. \ref{subsec:flocking}), and corroborate the stability analysis in Section \ref{sec:stabilityAnalysis} numerically (Sec. \ref{subsec:stability}). In all experiments, the ADAM optimizer is used with decaying factors $\beta_1 = 0.9$, $\beta_2 = 0.999$ \cite{Ba2010} and the results are averaged over $10$ random simulations.

\subsection{Source Localization}\label{subsec:source}

\begin{figure*}%
	\centering
	\begin{subfigure}{0.6\columnwidth}
		\includegraphics[width=1.05\linewidth, height = 0.76\linewidth, trim = {0cm 0cm 0cm 1cm}, clip]{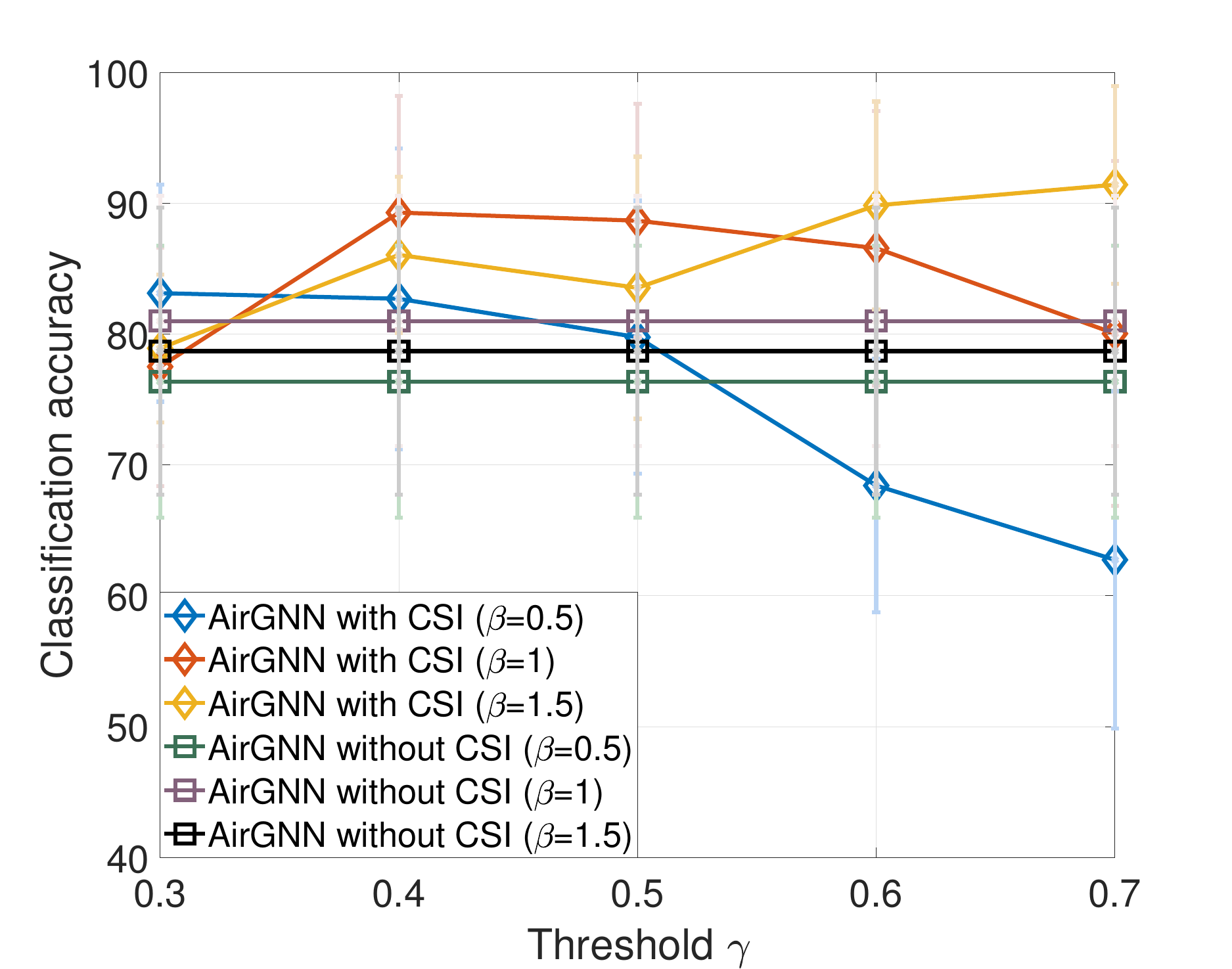}%
		\caption{}%
		\label{subfig:sourceGamma}%
	\end{subfigure}\hfill\hfill%
	\begin{subfigure}{0.6\columnwidth}
		\includegraphics[width=1.05\linewidth,height = 0.76\linewidth, trim = {0cm 0cm 0cm 1cm}, clip]{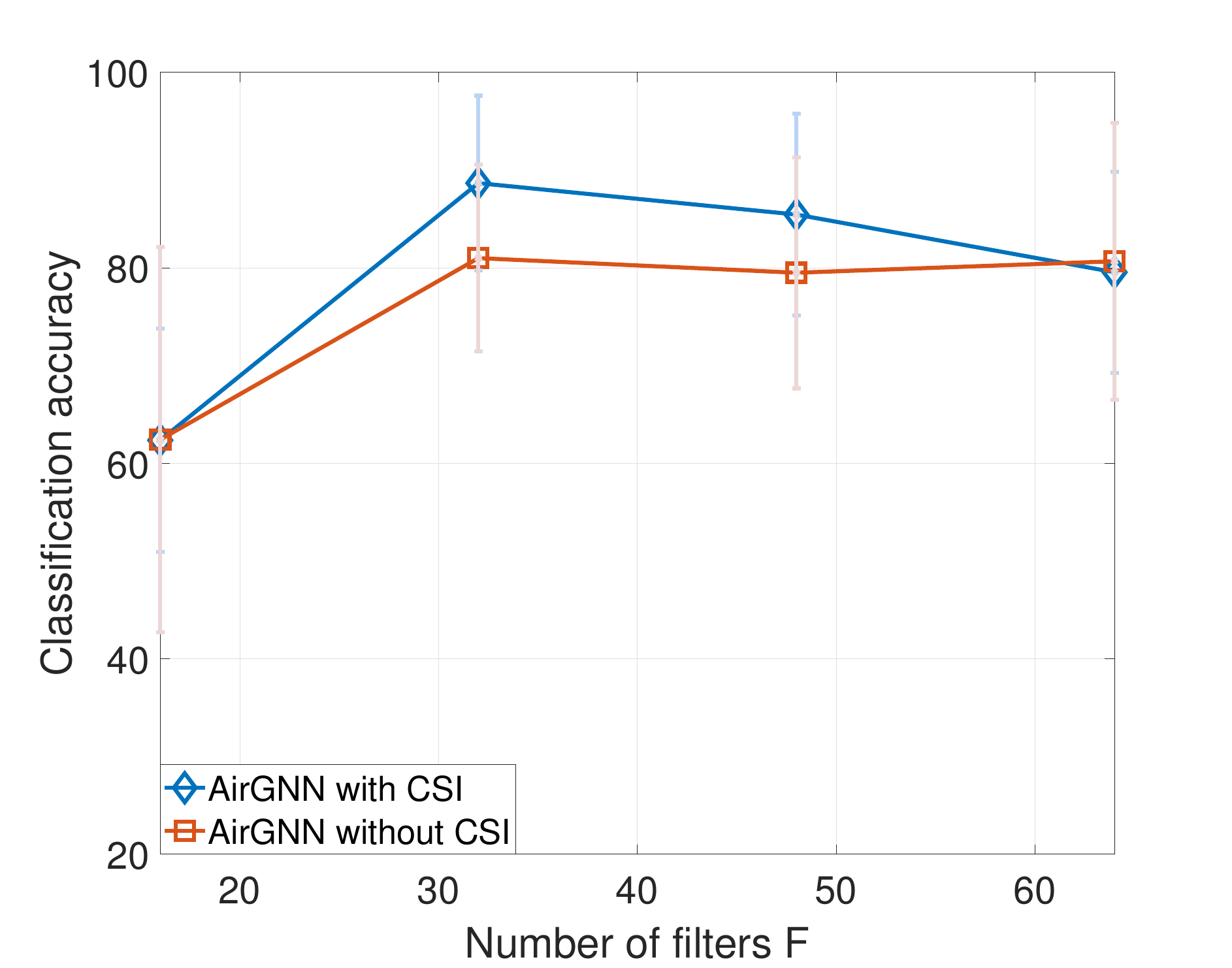}%
		\caption{}%
		\label{subfig:sourceF}%
	\end{subfigure}\hfill\hfill%
	\begin{subfigure}{0.6\columnwidth}
		\includegraphics[width=1.05\linewidth,height = 0.76\linewidth, trim = {0cm 0cm 0cm 1cm}, clip]{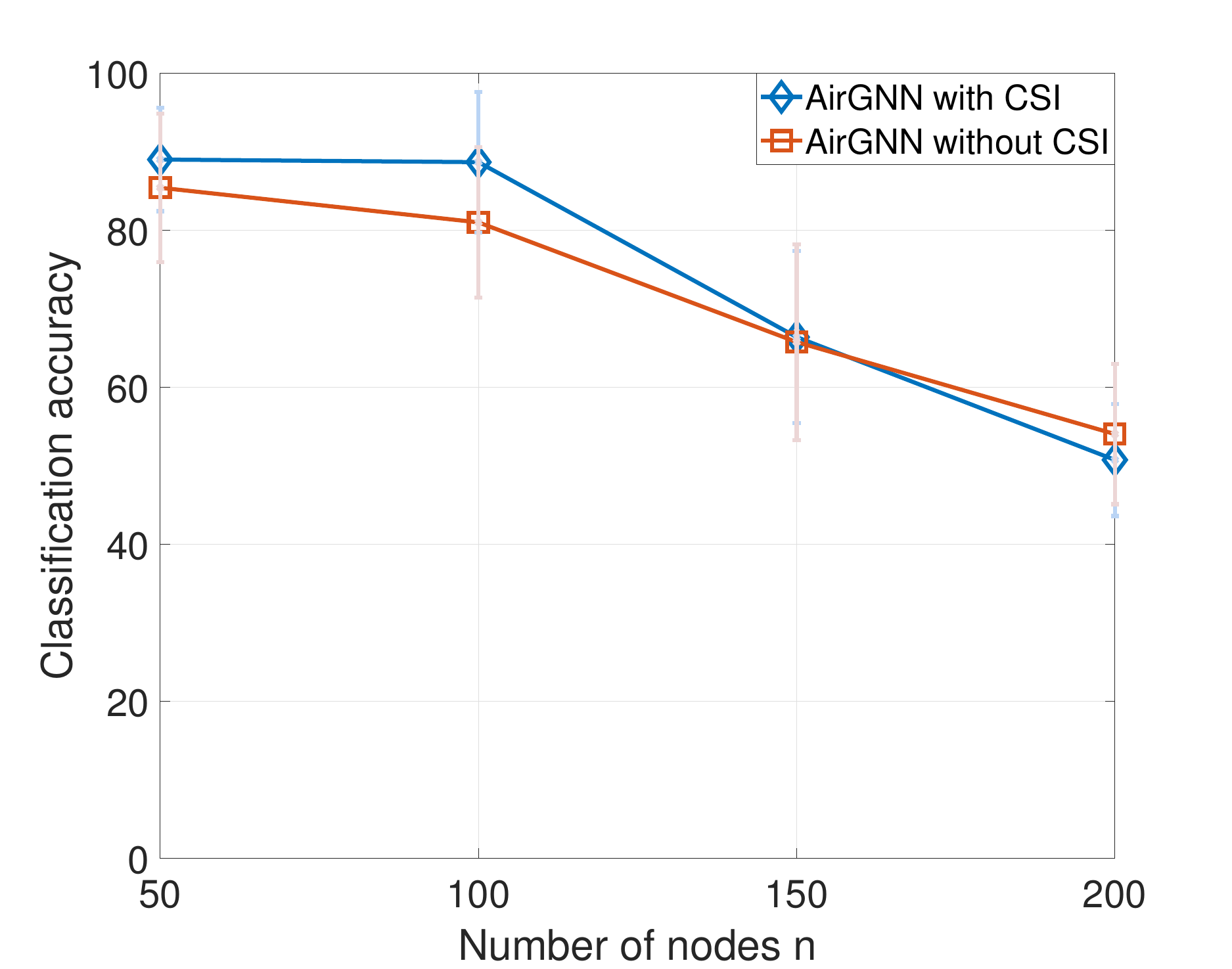}%
		\caption{}%
		\label{subfig:sourceN}%
	\end{subfigure}%
	\caption{(a) Performance of the AirGNNs with different power cut-off thresholds $\gamma$ in different channel conditions. (b) Performance of the AirGNNs with different numbers of features $F$. (c) Performance of the AirGNNs with different graph sizes, i.e., numbers of nodes $n$.}\label{fig:parameterSource}\vspace{-4mm}
\end{figure*}

This experiment considers a signal diffusion process over a stochastic block model (SBM) graph, where the intra- and the inter-community edge probabilities are $0.8$ and $0.2$. The graph has $n = 100$ nodes uniformly divided into $10$ communities, and each community has a source node $s_c$ for $c=1,\ldots,10$. The goal is to find the source community of a diffused signal in a decentralized manner. The initial signal is a Kronecker delta $\bbdelta_{c} \in \reals^n$ originated at one of the source nodes $\{s_c\}_{c=1}^{10}$ and is diffused over the graph at time $\tau$ as $\bbx_{c, \tau} =\bbS^\tau \bbdelta_{c} + \bbn$ with $\bbS = \bbA / \lambda_{max}(\bbA)$ and $\bbn$ a zero-mean Gaussian noise. The dataset consists of $1.5 \times 10^4$ samples with the source node $s_c$ and the diffusion time $\tau$ selected randomly from $\{s_1,\ldots,s_{10}\}$ and $\{1,\ldots,100\}$, which are split into $10^4$ samples for training, $2.5\times 10^3$ samples for validation, and $2.5 \times 10^3$ samples for testing. The AirGNN has two layers, each comprising $32$ filters of order $5$ and ReLU nonlinearity. We consider independent Rayleigh fading channels between the nodes with distribution scale parameter $\beta$ and the Gaussian noise with signal-to-noise ratio (SNR) $40$dB. The graph nodes are randomly distributed in the region $[-5km, 5km]^2$ and the path loss is modeled as $\zeta^{-\xi / 2}$, where $\zeta$ is the distance and $\xi=2$ is the path-loss exponent. The learning rate is set to $10^{-3}$, the batch-size to $50$, and the performance is measured in terms of the classification accuracy.

\smallskip
\noindent \textbf{Convergence.} Fig. \ref{subfig:convergence} shows the training procedures of the AirGNNs with (Section \ref{sec:airGNNsCSI}) and without CSI (Section \ref{sec:AirGNNNoCSI}) in the communication scenario with scale parameter $\beta = 1$. We see that the training / validation cost decreases with the number of iterations $t$ and approaches a stationary solution in both cases, which corroborates the convergence analysis in Theorem \ref{thm:convergence}. The convergence curve fluctuates during training, because channel fading effects and mini-batch sampling yield the randomness of the AirGNN output. The convergent cost increases from the AirGNN with CSI to the one without CSI. We attribute this behavior to the fact that the truncated power control strategy with an appropriate cut-off threshold inverses the channel fading and recovers the signal information; hence, improving the architecture performance. 

\smallskip
\noindent \textbf{Channel condition.} Figs. \ref{subfig:sourceSigma}-\ref{subfig:sourceSNR} depict the performance of the AirGNNs with different channel conditions, i.e., different scale parameters $\beta$ in Fig. \ref{subfig:sourceSigma} and different SNRs in Fig. \ref{subfig:sourceSNR}. The threshold of the AirGNN with CSI is set as $\gamma = 0.5$. We compare with two baselines: (i) GNN with ideal communications without channel fading and noise, and (ii) GNN with channel fading and noise. The first considers a classical GNN with no fading and noise, i.e., $h_{ij}^{(k)}=1$ and $n_i^{(k)}=0$ for $k=1,\ldots,K$ [cf. \eqref{eq:communicationChannel1}], during both training and inference. This is an ideal communication scenario, which may not hold in real-world wireless applications and serves as an upper bound only for reference. The second experiences fading and noise during inference, which is the practical scenario considered in this paper, but ignores them during training. 

The GNN without fading and noise exhibits the best performance as expected because it does not consider any channel effects. The AirGNN either with or without CSI consistently outperforms the GNN with fading and noise. This corroborates our theoretical finding, i.e., the AirGNN accounts for communication channels during training and is more robust to channel impairments during inference. Moreover, the AirGNN maintains a good performance in different channel conditions, while the GNN suffers more degradation in severe channels when $\beta$ is small / large or SNR is small and yields a performance close to a random classifier. This highlights the importance of architecture robustness to communication effects. 

The AirGNN \textit{with} CSI exhibits a better performance than the AirGNN \textit{without} CSI, when the cut-off threshold of truncated power control strategy $\gamma$ is appropriately selected based on the channel condition. This is because the former assumes the CSI available at transmitting nodes and leverages this information to reduce channel effects for signal transmission, while the latter does not use / require any prior knowledge about the CSI. 

\begin{figure*}%
	\centering
	\begin{subfigure}{0.495\columnwidth}
		\includegraphics[width=1.05\linewidth, height = 0.76\linewidth, trim = {0cm 0cm 0cm 1cm}, clip]{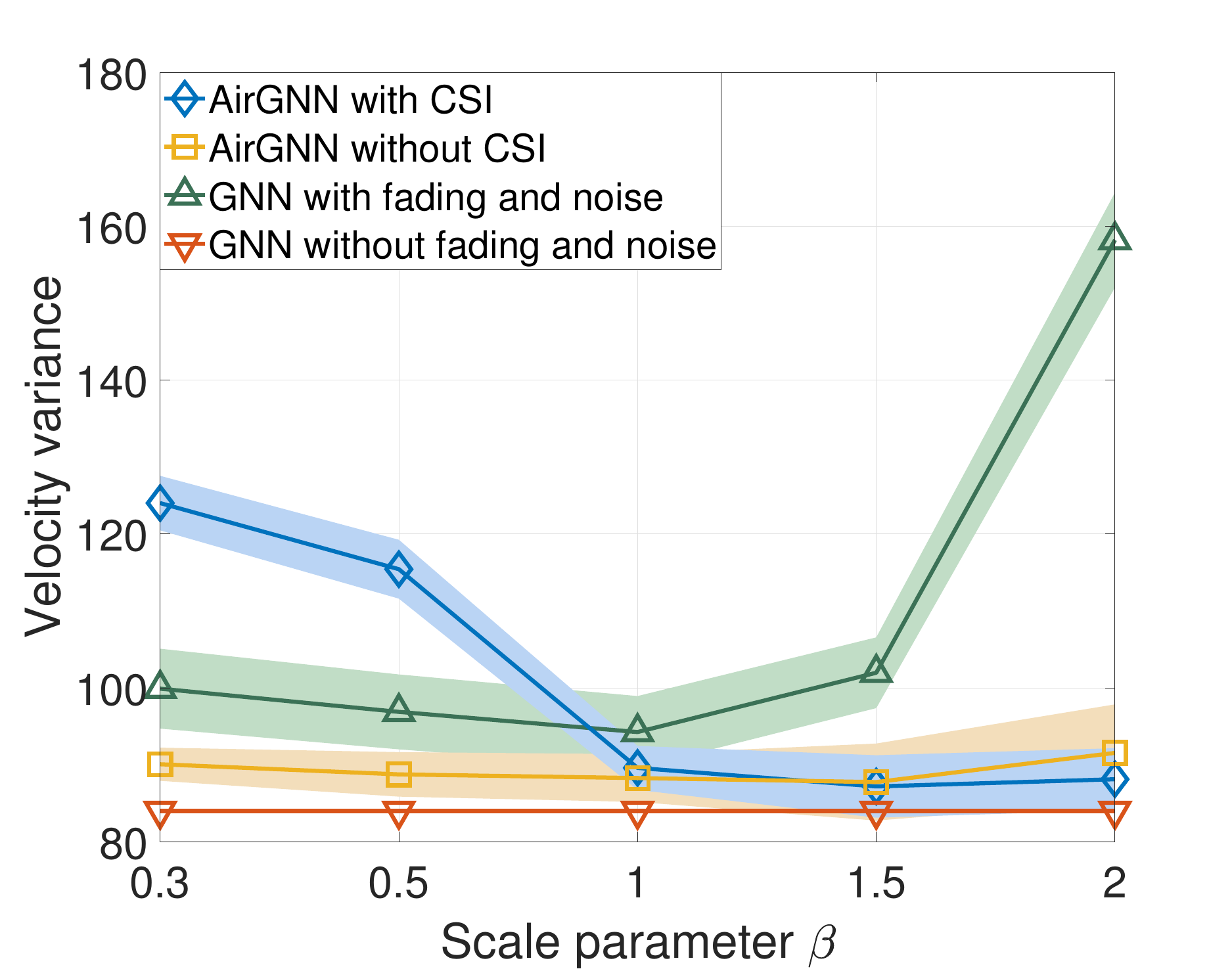}%
		\caption{}%
		\label{subfig:flockingSigma}%
	\end{subfigure}\hfill\hfill%
	\begin{subfigure}{0.495\columnwidth}
		\includegraphics[width=1.05\linewidth,height = 0.76\linewidth, trim = {0cm 0cm 0cm 1cm}, clip]{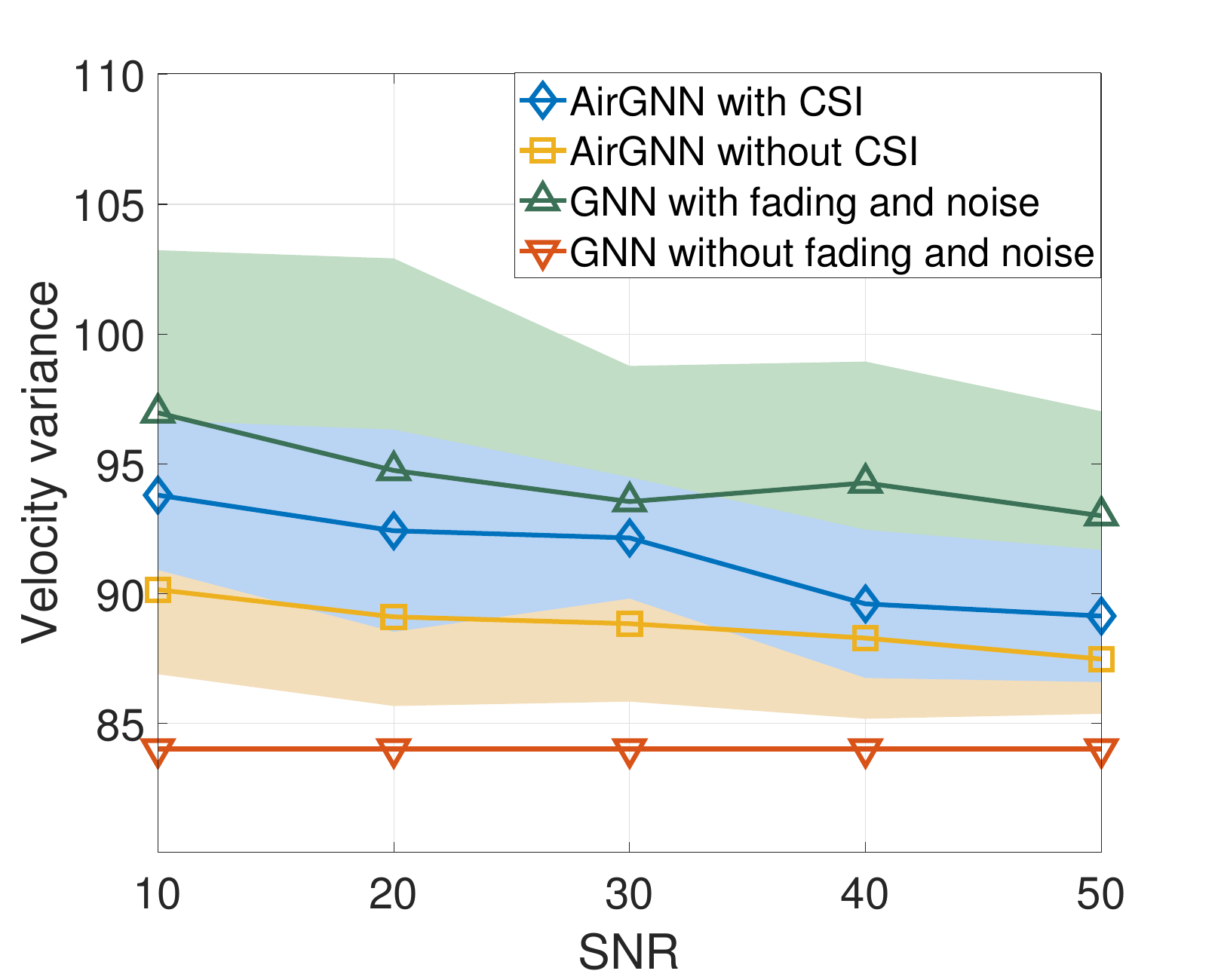}%
		\caption{}%
		\label{subfig:flockingSNR}%
	\end{subfigure}\hfill\hfill%
	\begin{subfigure}{0.495\columnwidth}
		\includegraphics[width=1.05\linewidth,height = 0.76\linewidth, trim = {0cm 0cm 0cm 1cm}, clip]{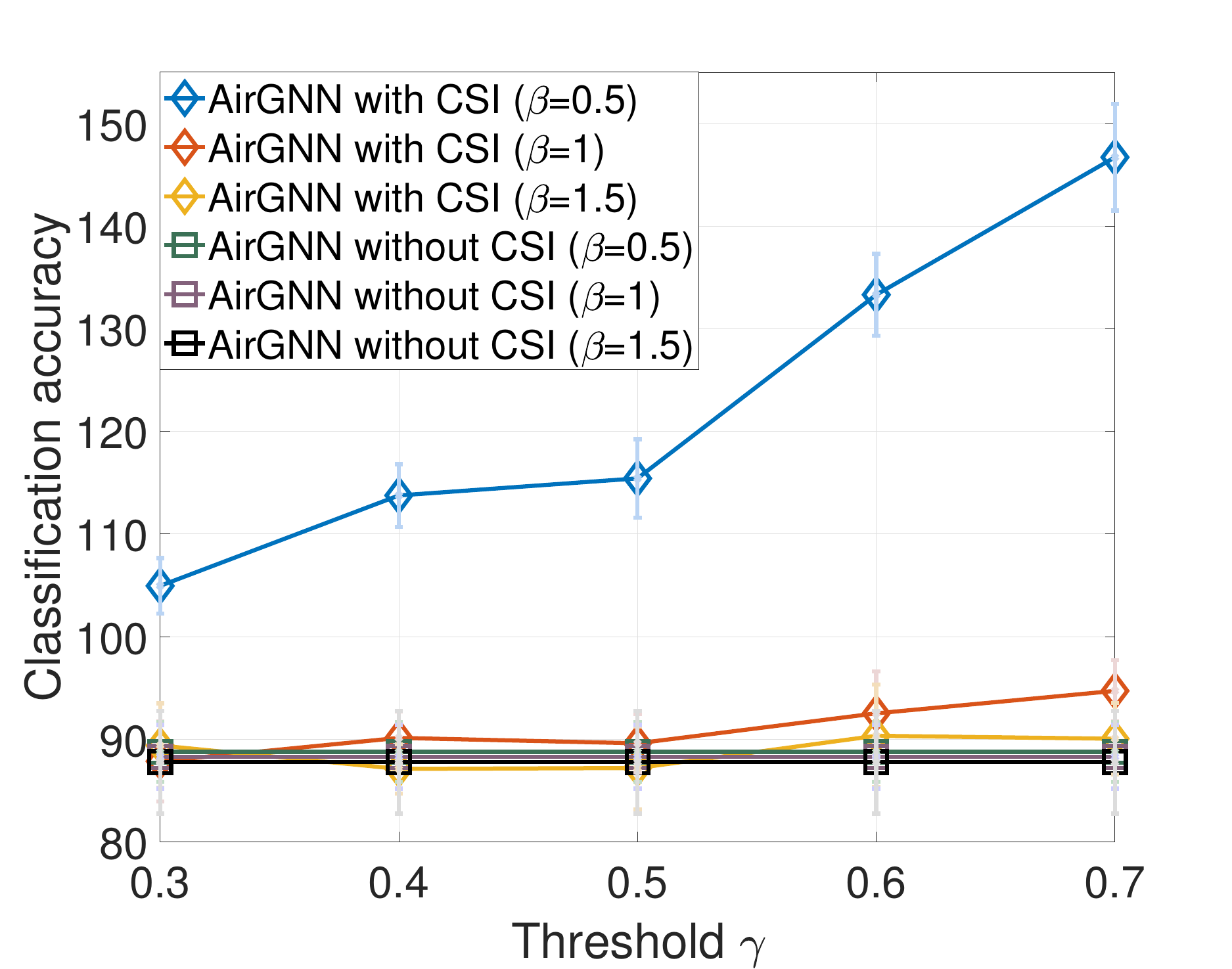}%
		\caption{}%
		\label{subfig:flockingGamma}%
	\end{subfigure}\hfill\hfill%
	\begin{subfigure}{0.495\columnwidth}
		\includegraphics[width=1.05\linewidth,height = 0.76\linewidth, trim = {0cm 0cm 0cm 1cm}, clip]{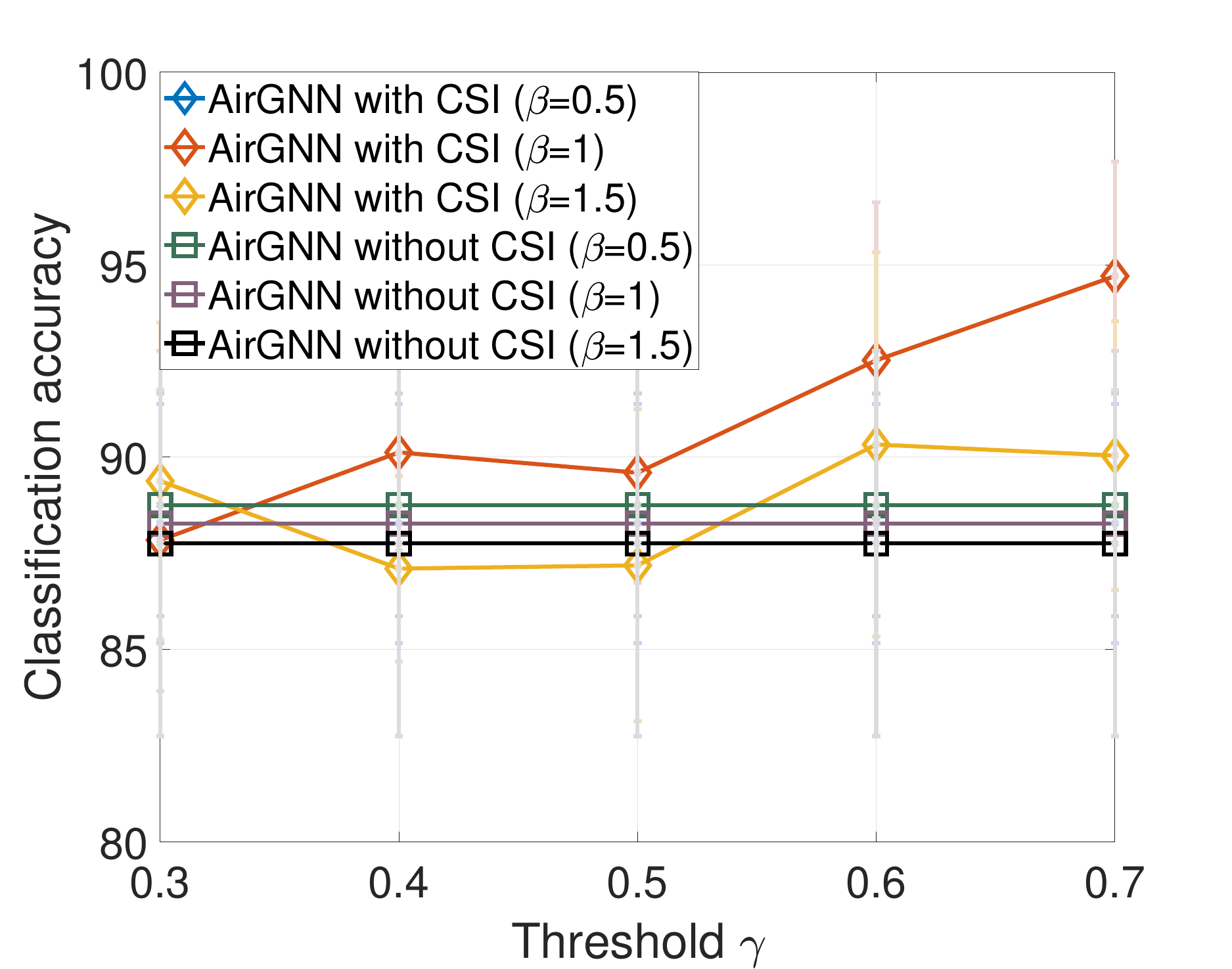}%
		\caption{}%
		\label{subfig:flockingGammaZoomed}%
	\end{subfigure}%
	\caption{(a) Performance of the AirGNNs and baselines in communication scenarios with different scale parameters $\beta$ for multi-robot flocking. (b) Performance of the AirGNNs and baselines in communication scenarios with different SNRs for multi-robot flocking. (c) Performance of the AirGNNs with different power cut-off thresholds $\gamma$. (d) Zoomed version of Fig. \ref{subfig:flockingGamma}.}\label{fig:performanceFlocking}\vspace{-4mm}
\end{figure*}

\smallskip
\noindent \textbf{Threshold.} Fig. \ref{subfig:sourceGamma} compares the AirGNNs with different power-cutoff 
thresholds $\gamma \in [0.3, 0.7]$ in communication scenarios with different scale parameters $\beta = 0.5, 1$ and $1.5$. We see that the performance of the AirGNN with CSI decreases as the threshold $\gamma$ increases for $\beta = 0.5$ because a larger $\gamma$ results in a higher link outage probability and more loss of signal information, while its performance increases with the threshold $\gamma$ for $\beta=1.5$ because a larger $\gamma$ improves the signal-to-noise ratio and reduces the impact of Gaussian noise. The performance for $\beta = 1$ first increases with $\gamma$ and then decreases at large $\gamma$, which can be explained by the trade-off between the link outage probability and the signal-to-noise ratio controlled by the selection of the threshold $\gamma$ -- see Section \ref{subsec:TCPC}. It is also worth noting that the AirGNN with CSI does not always outperform the AirGNN without CSI, even though it requires additional information, i.e., the CSI, for implementation. The latter depends on the selection of the threshold and the specific channel condition, as discussed in Section \ref{sec:AirGNNNoCSI}. 

\smallskip
\noindent \textbf{Hyperparameters.} We further compare the AirGNNs with different architecture hyperparameters, i.e., number of filters $F$ in Fig. \ref{subfig:sourceF} and graph size $n$ in Fig. \ref{subfig:sourceN}. Fig. \ref{subfig:sourceF} shows that the classification accuracy first increases and then decreases with the number of filters $F$. The former is because the expressive power (representational capacity) of the AirGNN increases with $F$, which improves the performance, while the latter is because more filters involve more channel impairments into the architecture, which amplifies the impact of fading and noise. Fig. \ref{subfig:sourceF} shows that the classification accuracy decreases with the number of nodes $n$. This is because given the same AirGNN architecture (i.e., the same expressive power), a larger graph (i) makes the task more challenging and (ii) introduces more channel impairments into the architecture.

\subsection{Multi-Robot Flocking}\label{subsec:flocking}

This experiment considers the problem of robot swarm control over communication graph. The goal is to learn a decentralized controller that coordinates the robots to move together and avoid collision. The network has $n = 50$ robots with initial velocities sampled randomly in $[-3m/s, 3m/s]^2$, and robot $i$ can communicate with robot $j$ if they are within a communication radius $r = 1.5m$. The communication graph is $\ccalG = (\ccalV, \ccalE)$ with the node set $\ccalV$ as the robots and the edge set $\ccalE$ as available communication links, and the graph signal $\bbx$ is the relevant feature of robot position and velocity. We use imitation learning to train the AirGNN by mimicing the optimal centralized controller \cite{tolstaya2020learning}. The dataset consists of $450$ trajectories of $100$ time steps, which are split into $400$ for training, $25$ for validation, and $25$ for testing. We consider a single-layer AirGNN with $32$ filters of order $5$ and the hyperbolic tangent nonlinearity. The learning rate is $5 \times 10^{-4}$ and the batch-size is $20$. We measure the performance with the variance of robot velocities for a trajectory, which quantifies how far from the consensus the robots are, and a lower value represents a better performance. 

\smallskip
\noindent \textbf{Channel condition.} Figs. \ref{subfig:flockingSigma}-\ref{subfig:flockingSNR} compare the performance of the AirGNNs with baselines in different channel conditions, i.e., different scale parameters $\beta$ in Fig. \ref{subfig:flockingSigma} and different SNRs in Fig. \ref{subfig:flockingSNR}. The GNN without fading and noise performs best following our intuition since it assumes ideal communication links. The AirGNNs exhibit a comparable performance with slight degradation. This is because the information loss induced by channel impairments leads to inevitable errors, which cannot be resolved completely by training. The AirGNN without CSI consistently outperforms the GNN with fading and noise, and this improved performance is more visible for worse channel conditions (e.g., large scale parameters and small SNRs). We again attribute this behavior to the robustness of the AirGNN to channel effects because it accounts for communication channels during training. The AirGNN with CSI achieves a lower velocity variance than the GNN with fading and noise in most cases, but performs worse in communication scenarios with small $\beta$. This can be explained by the inappropriate selection of the power cut-off threshold $\gamma$. That is, $\gamma = 0.5$ is too large for communication channels with small $\beta$, which leads to a high link outage probability and performance degradation. 

\smallskip
\noindent \textbf{Threshold.} Lastly, we evaluate the AirGNNs with different power cut-off thresholds $\gamma$ in communication scenarios with different scale parameters $\beta$. In Fig. \ref{subfig:flockingGamma}, we see that the performance of the AirGNN with CSI improves rapidly with the decreasing of $\gamma$ for communication scenario with $\beta = 0.5$. This is because channel coefficients $h$ are small in this condition, which requires a small $\gamma$ to reduce the link outage probability and improve the architecture performance. Similar phenomenon is observed for communication scenario with $\beta=1$. In Fig. \ref{subfig:flockingGammaZoomed}, we present a zoomed view of Fig. \ref{subfig:flockingGamma} and see that the performance for communication scenario with $\beta=1.5$ first increases and then decreases as the threshold $\gamma$ decreases. It achieves the lowest velocity variance at $\gamma = 0.4$, which corroborates the trade-off between the link outage probability and the signal-to-noise ratio controlled by $\gamma$. The comparison between the AirGNN with CSI and the one without CSI depends on the selection of the power cut-off threshold and the condition of communication channels -- see Section \ref{sec:AirGNNNoCSI} for discussions.

\subsection{Stability Corroboration}\label{subsec:stability}

The goal of this experiment is to corroborate the stability analysis in Section \ref{sec:stabilityAnalysis}. We consider the decentralized task of source localization in the SBM graph of $n=50$ nodes. The GNN contains $2$ layers, each of which contains $F=32$ filters of order $K=5$ followed by a ReLU nonlinearity. The channel coefficient is drawn from a Gaussian distribution with mean $\mu = 0.9$ and standard deviation $\delta = 0.1$,\footnote{We consider a Gaussian distribution for stability analysis because it allows to study mild channel effects by considering its mean $\mu$ close to one and standard deviation $\delta$ close to zero.} and the SNR is $50$dB by default. For stability analysis, we train the GNN assuming ideal communication links, i.e., $h_{ij}^{(k)}=1$ and $n_i^{(k)}=0$ for $k=1,\ldots,K$ [cf. \eqref{eq:communicationChannel}], but evaluate the trained model under channel impairments in different communication scenarios, i.e., different means $\mu$ in Fig. \ref{subfig:stabilityMean}, different standard deviations $\delta$ in Fig. \ref{subfig:stabilityVariance} and different SNRs in Fig. \ref{subfig:stabilitySNR}. 

\begin{figure*}%
	\centering
	\begin{subfigure}{0.6\columnwidth}
		\includegraphics[width=1.05\linewidth, height = 0.74\linewidth, trim = {0cm 0cm 0cm 1cm}, clip]{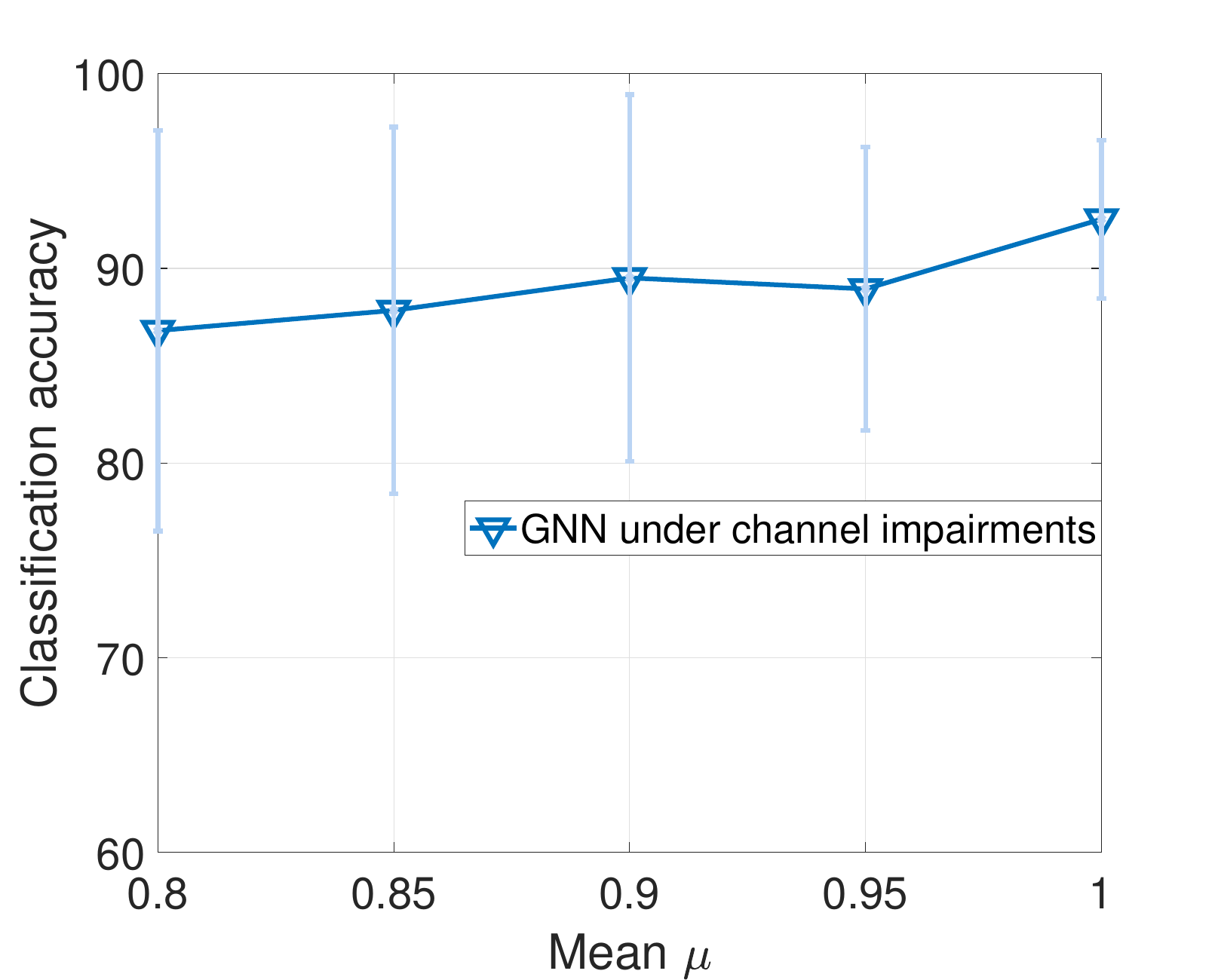}%
		\caption{}%
		\label{subfig:stabilityMean}%
	\end{subfigure}\hfill\hfill%
	\begin{subfigure}{0.6\columnwidth}
		\includegraphics[width=1.05\linewidth,height = 0.74\linewidth, trim = {0cm 0cm 0cm 1cm}, clip]{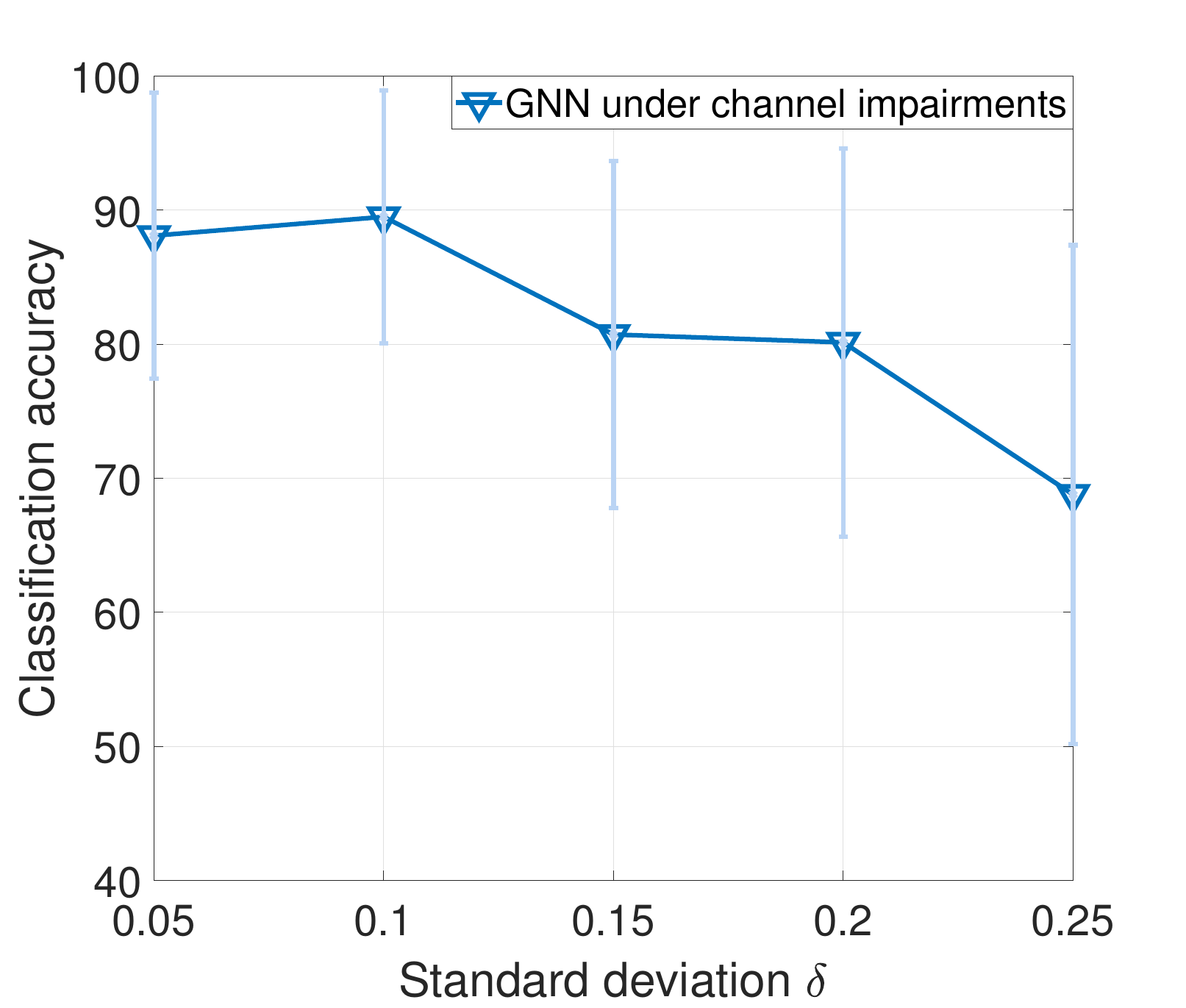}%
		\caption{}%
		\label{subfig:stabilityVariance}%
	\end{subfigure}\hfill\hfill%
	\begin{subfigure}{0.6\columnwidth}
		\includegraphics[width=1.05\linewidth,height = 0.74\linewidth, trim = {0cm 0cm 0cm 1cm}, clip]{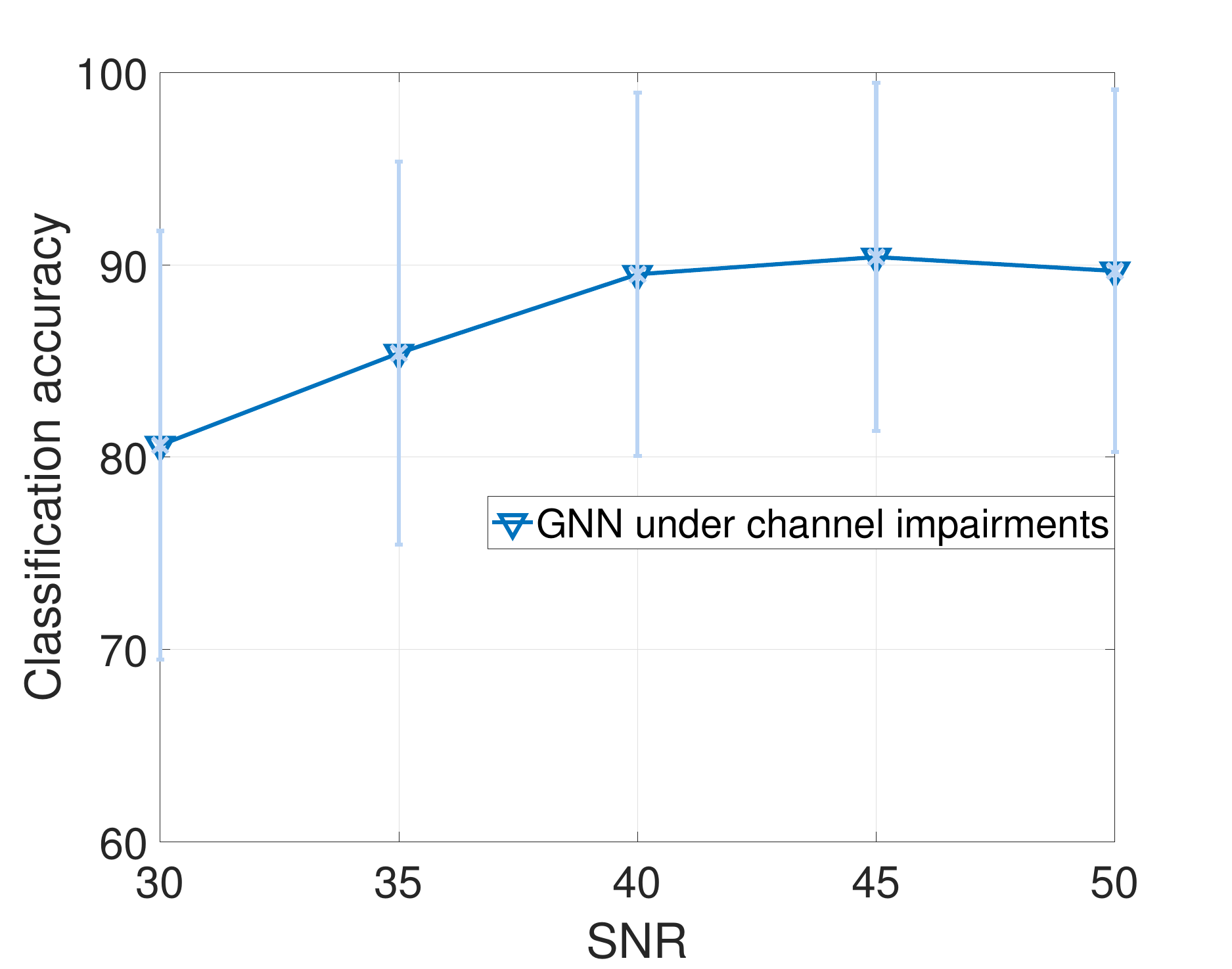}%
		\caption{}%
		\label{subfig:stabilitySNR}%
	\end{subfigure}%
	\caption{Performance of the GNN under channel impairments in different communication scenarios. (a) Different means $\mu$. (b) Different standard deviations $\delta$. (c) Different SNRs.}\label{fig:stability}\vspace{-4mm}
\end{figure*}

We see that when channel conditions are stable, i.e., mean $\mu$ approaches one, standard deviation $\delta$ approaches zero, and SNR becomes large, the GNN under channel impairments maintains a comparable classification accuracy to that assuming ideal communications, which indicates that it is stable to small channel effects. However, the performance degrades when channel conditions get worse, i.e., $\mu$ decreases, $\delta$ increases, and SNR decreases. This follows our analysis in Theorem \ref{thm:AirGNNStability} because large values of $(1-\mu)$, $\delta$ and $\eps$\footnote{A larger standard deviation of Gaussian noise $\eps$ corresponds to a lower SNR.} increase the stability bound of the output deviation and result in more performance degradation, highlighting the importance of incorporating communication channels into the GNN architecture.

\section{Conclusion}
\label{sec:conclusion}

This paper focused on channel impairments of wireless communications in decentralized implementation of GNNs. We performed stability analysis to study the performance degradation of GNNs under channel impairments, and developed graph neural networks over the air (AirGNNs) to alleviate this performance degradation. AirGNNs consist of multiple layers with each layer comprising graph filters implemented over the air and a pointwise nonlinearity. They incorporate communication channels into graph signal processing and rely on noisy neighborhood information to extract features, which improve the architecture robustness to channel fading and noise. When channel state information (CSI) is available at the transmitting nodes, we proposed a signal transmission strategy, referred to as truncated channel-inversion power control, to further improve the performance within the AirGNN architecture. When CSI is unknown, we account for random channel states in the cost function and developed a SGD based method to train AirGNNs. We showed the proposed training procedure is equivalent to running SGD on an associated stochastic optimization problem, proved its convergence to a stationary solution, and characterized statistical behavior of the trained model through variance analysis. Numerical experiments are performed on decentralized tasks of source localization and multi-robot flocking, corroborating the effectiveness of AirGNNs to improve the task performance in wireless networks over communication channels. For future work, we plan to evaluate AirGNNs with other channel models and in other decentralized tasks, such as power outage prediction in smart grids and decentralized navigation of multi-robot systems.

\appendices

\section{Proof of Theorem 1}\label{Proof:Theorem1}

Let $\tilde{\bbu} = \bbH_{\rm air}(\bbS)\bbx = \tilde{\bbu}_1 + \tilde{\bbu}_2$ be the AirGF output [cf. (6)], where $\tilde{\bbu}_1 = \bbP_{\rm air}(\bbS, \bbx)$ is the signal component and $\tilde{\bbu}_2 = \bbN_{\rm air}(\bbS, \bbx)$ is the noise component, and $\bbu = \bbH(\bbS)\bbx$ be the nominal graph filter output. The expected output difference between $\tilde{\bbu}$ and $\bbu$ is 
\begin{align} \label{proof:thm11}
	&\mathbb{E}\left[\| \tilde{\bbu} - \bbu \|^2_2\right]\!=\! \mathbb{E}\!\left[\tr \big(\tilde{\bbu} \tilde{\bbu}^\top \!+\! \bbu \bbu^\top \!-\! 2\tilde{\bbu} \bbu^\top \big)\right] \\
	& = \mathbb{E}\!\left[\tr \big(\tilde{\bbu}_1 \tilde{\bbu}_1^\top \!+\! \bbu \bbu^\top \!-\! 2\tilde{\bbu}_1 \bbu^\top \!+\! 2 \tilde{\bbu}_1 \tilde{\bbu}_2^\top \!+\! \tilde{\bbu}_2 \tilde{\bbu}_2^\top \!-\! 2 \tilde{\bbu}_2 \bbu^\top \big)\right], \nonumber
\end{align}
where $\tr(\cdot)$ is the trace operator. Since the Gaussian noise is w.r.t. zero mean, we have $\mathbb{E}[\bbn^{(k)}] = \bb0$ for $k=1,\ldots,K$ and thus $\mathbb{E}[\tilde{\bbu}_2] = \bb0$ by definition [cf. (6)]. By using this result in \eqref{proof:thm11}, we have
\begin{align} \label{proof:thm12}
	&\mathbb{E}\!\!\left[\!\| \tilde{\bbu} \!-\! \bbu \|^2_2\right]\!\!\!=\! \mathbb{E}\!\!\left[\tr \big(\!\tilde{\bbu}_1 \tilde{\bbu}_1^\top \!\!+\! \bbu \bbu^\top\!\!\!-\! 2\tilde{\bbu}_1 \bbu^\top \!\big)\!\right] \!\!+\!  \mathbb{E}[\tr \big(\tilde{\bbu}_2 \tilde{\bbu}_2^\top \!\big)\!].
\end{align}
We consider two terms in \eqref{proof:thm12} separately.

\smallskip
\noindent \textbf{First term in \eqref{proof:thm12}.} By adding and subtracting $\bbu^\top \bbu$ in the trace, we can rewrite it as
\begin{align} \label{proof:them1255}
	\mathbb{E}\!\left[\tr \big(\tilde{\bbu}_1 \tilde{\bbu}_1^\top \!-\! \bbu \bbu^\top \big)\right]+ 2 \mathbb{E}\!\left[\tr\big( \bbu \bbu^\top \!-\! \tilde{\bbu}_1 \bbu^\top \big)\right],
\end{align}
where the linearity of the expectation and the trace is used. We now consider two terms in \eqref{proof:them1255}. For the first term $\mathbb{E}\!\left[\tr \Big(\tilde{\bbu}_1 \tilde{\bbu}_1^\top \!-\! \bbu \bbu^\top \Big)\!\right]$, we substitute the expressions of signal component $\tilde{\bbu}_1$ and filter output $\bbu$, and use the symmetry of the shift operators $\bbS^{(k)}$\footnote{For convenience of expression, we drop the subscript of $\bbS_{\rm air}^{(k)}$ in all proofs.} and $\bbS$ to write\footnote{Throughout the proofs, we use the shorthand notation $\sum_{a,b,c = \alpha, \beta, \gamma}^{A, B, C} (\cdot)$ to denote $\sum_{a = \alpha}^A\sum_{b = \beta}^B\sum_{c = \gamma}^C (\cdot)$ to avoid overcrowded expressions. When the extremes of the sum ($\alpha, \beta, \gamma$ or $A, B, C$) are the same, we will write directly the respective value.}
\begin{gather} \label{proof:them13}
	\begin{split}
		&\mathbb{E}\big[\tr \big(\tilde{\bbu}_1 \tilde{\bbu}_1^\top \!-\! \bbu \bbu^\top \big)\big] \\
		&= \sum_{k, \ell=0}^K \alpha_k \alpha_\ell \big( \mathbb{E} \big[ \tr \big( \tilde{W}(k,\ell) \big) \big] -  \tr \big( W(k,\ell) \big) \big)
	\end{split}
\end{gather}
with $\tilde{W}(k,\ell)= \bbS^{(k:0)} \bbx \bbx^{\top} \bbS^{(0:\ell)}$ and $W(k, \ell)= \bbS^k  \bbx \bbx^{\top} \bbS^\ell$, where $\bbS^{(k:0)} = \bbS^{(k)} \cdots \bbS^{(0)}$, $\bbS^{(0:\ell)} = \bbS^{(0)}\cdots\bbS^{(\ell)}$ are concise notations and $\bbS^{(0)} = \bbI$ by default. By representing the $k$th AirGSO $\bbS^{(k)}$ as $\bbS^{(k)} = \bbS + \bbE^{(k)}$ with $\bbE^{(k)}$ the deviation of $\bbS^{(k)}$ from $\bbS$ and substituting this representation into $\tilde{W}(k,\ell)$, we can expand the terms in $\tilde{W}(k,\ell)$ as 
\begin{align}
	\label{proof:them135} &\mathbb{E}\left[ \tilde{W}(k,\ell)\right] = \mathbb{E}\left[ (\bbS + \bbE^{(k)})\cdots \bbx \bbx^\top \cdots (\bbS + \bbE^{(\ell)})\right]\\
	&\!=\! {\bbS}^{k}\bbx \bbx^\top\! {\bbS}^{\ell}  \!\!+\! \mathbb{E}\big[ (\bbS \!+\! \bbE^{(k)})\!\cdots\! \bbx \bbx^\top \bbS^\ell \!\!-\! {\bbS}^{k}\bbx \bbx^\top\! {\bbS}^{\ell} \big] \nonumber\\
	&+\! \mathbb{E}\big[ \bbS^{k}\bbx \bbx^\top \cdots (\bbS \!+\! \bbE^{(\ell)}) \!-\! {\bbS}^{k}\bbx \bbx^\top\! {\bbS}^{\ell}\big] \nonumber\\
	& +\!\! \mathbb{E}\Big[\! \sum_{r_1\!=\!1}^{k}\! \sum_{r_2\!=\!1}^\ell\!\! \bbS^{k\!-\!r_1}\bbE^{(r_1)} \bbS^{r_1\!-\!1}\bbx \bbx^\top\! \bbS^{r_2\!-\!1}\bbE^{(r_2)} \bbS^{\ell\!-\!r_2}\!\Big] \!\!+\! \mathbb{E}\left[ \bbC_{k\ell}\right] \nonumber 
\end{align}
for $k,\ell \ge 1$. The first term in \eqref{proof:them135} contains the maximal power of $\bbS$; the second term includes error matrices only expanded from left-side $\bbS^{(k:0)}=(\bbS \!+\! \bbE^{(k)})\cdots (\bbS+\bbE^{(1)})$; the third term includes error matrices only expanded from right-side $\bbS^{(0:\ell)}=(\bbS \!+\! \bbE^{(1)})\cdots (\bbS+\bbE^{(\ell)})$; the fourth term collects the cross-products that include two error matrices $\bbE^{(r_1)}$ and $\bbE^{(r_2)}$, and the last term $\bbC_{k\ell}$ aggregates the sum of the remaining terms. We can rewrite \eqref{proof:them135} as
\begin{align}
	\label{proof:them14} 
	&\mathbb{E}\left[ \tilde{W}(k,\ell)\right]\!\!=\! - {\bbS}^{k}\bbx \bbx^\top\! {\bbS}^{\ell}  \!\!+\! \mathbb{E}\!\left[\! \bbS^{(k:0)}\bbx \bbx^\top\! \bbS^\ell \!\right]\!+\! \mathbb{E}\!\left[\! \bbS^{k}\bbx \bbx^\top \bbS^{(0:\ell)}\!\right] \nonumber\\
	&+\! \mathbb{E}\Big[\! \sum_{r_1=1}^{k}\! \sum_{r_2=1}^\ell\!\! \bbS^{k-r_1}\bbE^{(r_1)} \bbS^{r_1-1}\bbx \bbx^\top\! \bbS^{r_2-1}\bbE^{(r_2)} \bbS^{\ell-r_2}\!\Big] \nonumber\\
	&+\! \mathbb{E}\left[ \bbC_{k\ell}\right]. 
\end{align}
By substituting $W(k,\ell) =  {\bbS}^{k}\bbx \bbx^\top {\bbS}^{\ell}$ and \eqref{proof:them14} into \eqref{proof:them13}, we have
\begin{align} \label{proof:them15}
	&\mathbb{E}\!\left[\tr \big(\tilde{\bbu}_1^\top\! \tilde{\bbu}_1 \!-\! \bbu^\top\! \bbu \big)\right] = - 2 \mathbb{E}\!\left[\tr\big( \bbu^\top\! \bbu \!-\! \tilde{\bbu}_1^\top \bbu \big)\right] \\
	& \!+\!\!\! \sum_{k\!=\!1}^K\! \sum_{\ell\!=\!1}^K\!\!\! \alpha_k \alpha_\ell \tr \Big(\! \mathbb{E}\!\Big[\! \sum_{r_1\!=\!1}^{k}\! \sum_{r_2\!=\!1}^\ell\!\! \bbS^{k\!-\!r_1}\!\bbE^{(\!r_1\!)} \bbS^{r_1\!-\!1}\bbx \bbx^\top\! \bbS^{r_2\!-\!1}\bbE^{(\!r_2\!)} \bbS^{\ell\!-\!r_2}\!\Big] \!\Big) \nonumber \\
	&+\! \sum_{k=0}^K \sum_{\ell=0}^K \alpha_k \alpha_\ell \tr \left( \mathbb{E}\left[ \bbC_{k\ell}\right] \right), \nonumber 
\end{align}
where the second and third terms in \eqref{proof:them14} are the same and become $\tilde{\bbu}_1^\top \bbu$ when summed up over indexes $k$ and $\ell$. We then consider the three terms in \eqref{proof:them15}, respectively. For this analysis, we will use the inequality
\begin{gather} \label{proof:them16}
	\begin{split}
		{\rm tr} (\bbA \bbB) \le \frac{\| \bbA + \bbA^\top \|_2}{2}{\rm tr}(\bbB) \le \| \bbA \|_2 {\rm tr}(\bbB)
	\end{split}
\end{gather}
that holds for any square matrix $\bbA$ and positive semi-definite matrix $\bbB$ \cite{Wang1986}. The first term in \eqref{proof:them15} is the opposite of the second term in \eqref{proof:them1255} such that it cancels out when substituted into \eqref{proof:them1255}.

For the second term in \eqref{proof:them15}, we bring the trace inside the expectation and leveraging the trace cyclic property $\tr(\bbA\bbB\bbC) = \tr(\bbC\bbA\bbB) = \tr(\bbB\bbC\bbA)$ to write
\begin{align}\label{proof:them165}
	\mathbb{E}\!\Big[\! \tr \Big(\! \sum_{k\!=\!1}^K\! \sum_{\ell\!=\!1}^K\!\! \alpha_k \alpha_\ell\!\! \sum_{r_1\!=\!1}^{k}\! \sum_{r_2\!=\!1}^\ell\!\! \bbE^{(\!r_2\!)}\bbS^{k\!+\!\ell\!-\!r_1\!-\!r_2}\bbE^{(\!r_1\!)} \bbS^{r_1\!-\!1}\bbx \bbx^\top\! \bbS^{r_2\!-\!1}\! \Big)\!\Big].
\end{align}
By using Lemma \ref{lemma0:traceOperation}, we can upper bound \eqref{proof:them165} as
\begin{align}
	\label{proof:them17}
	&\mathbb{E}\!\Big[\! \tr \Big(\! \sum_{k\!=\!1}^K\! \sum_{\ell\!=\!1}^K\!\! \alpha_k \alpha_\ell\!\! \sum_{r_1\!=\!1}^{k}\! \sum_{r_2\!=\!1}^\ell\!\! \bbE^{(\!r_2\!)}\bbS^{k\!+\!\ell\!-\!r_1\!-\!r_2}\bbE^{(\!r_1\!)} \bbS^{r_1\!-\!1}\bbx \bbx^\top\! \bbS^{r_2\!-\!1}\! \Big)\!\Big] \nonumber\\
	& \le (n+K-1) C_L^2 \|\bbx\|_2^2 (1-\mu)^2 + ndC_L^2\|\bbx\|_2^2 \delta^2
\end{align}
with $d$ the maximal degree of the graph $\ccalG$. 

For the third term in \eqref{proof:them1255}, matrix $\bbC_{k\ell}$ comprises the sum of the remaining expansion terms, where each term contains at least three error matrices and at least two error matrices are different. In this context, these terms can be bounded by a factor containing at least two terms of $\{\mathbb{E}[\bbE^{(r_1)}], \mathbb{E}[\bbE^{(r_2)}], \mathbb{E}[\bbE^{(r_3)}], \mathbb{E}[{\bbE^{(r_1)}}^2], \mathbb{E}[{\bbE^{(r_2)}}^2], \mathbb{E}[{\bbE^{(r_3)}}^2]\}$ inside the trace, where $r_1 \ne r_2 \ne r_3$ are different indices. Since the frequency response $f(\bblambda)$ is bounded, the coefficients $\{ \alpha_k \}_{k=0}^K$ are also bounded. From the facts that $\| \bbS \|_2$ is bounded and $\mathbb{E}[\bbE^{(r)}] = (1-\mu)\bbS$ and $\mathbb{E}\left[ {\bbE^{(r)}}^2 \right] = (1-\mu)^2\bbS^2+\delta^2 \bbD$ from Lemma \ref{LemmmaErrorMatrix}, we have
\begin{gather} \label{proof:them18}
	\begin{split}
		\mathbb{E} \Big[{\rm tr}\Big(\!\sum_{k, \ell=0}^K\! \alpha_k \alpha_\ell \bbC_{k\ell} \!\Big)\!\Big] \!\!=\!\! \Big(\!\ccalO((1\!-\!\mu)^3) \!+\! \ccalO(\delta^2 (1\!-\!\mu))\!\Big)\|\bbx\|_2^2.
	\end{split}
\end{gather}

\smallskip
\noindent \textbf{Second term in \eqref{proof:thm12}.} By substituting the expression of noise component $\tilde{\bbu}_2$ into the second term in \eqref{proof:thm12}, we have
\begin{align} \label{proof:them19}
	& \mathbb{E}\!\left[\tr \big(\tilde{\bbu}_2^\top \tilde{\bbu}_2 \!\big)\!\right] \!=\! \mathbb{E} \Big[ \tr \Big(\!\Big(\sum_{k=1}^{K} \sum_{\kappa=k}^{K}\alpha_{\kappa} \prod_{\tau=k+1}^{\kappa} \bbS^{(\tau)} \bbn^{(k)}\Big)^\top \nonumber \\
	& \cdot \Big(\sum_{k=1}^{K} \sum_{\kappa=k}^{K}\alpha_{\kappa}\! \prod_{\tau=k+1}^{\kappa} \bbS^{(\tau)} \bbn^{(k)}\Big)\Big) \Big].
\end{align}
Since $\{\bbn^{(k)}\}_{k=1}^K$ are independent and $\mathbb{E}[\bbn^{(k)}] = \bb0$ from the zero-mean Gaussian distribution, we can re-write \eqref{proof:them19} as
\begin{align} \label{proof:them110}
	\mathbb{E}\!\left[\tr \big(\tilde{\bbu}_2^\top \tilde{\bbu}_2 \!\big)\!\right] \!= & \mathbb{E} \Big[ \tr \Big(\!\sum_{k=1}^{K} \Big(\sum_{\kappa=k}^{K} \alpha_{\kappa}\!\! \prod_{\tau=k+1}^{\kappa} \bbS^{(\tau)} \bbn^{(k)}\Big)^\top \\
	&\cdot \Big(\sum_{\kappa=k}^{K} \alpha_{\kappa}\!\! \prod_{\tau=k+1}^{\kappa} \bbS^{(\tau)} \bbn^{(k)}\Big) \Big)\Big]. \nonumber
\end{align}
We similarly represent the $k$th AirGSO $\bbS^{(k)}$ as $\bbS^{(k)} = \bbS + \bbE^{(k)}$ with $\bbE^{(k)}$ the deviation of $\bbS^{(k)}$ from $\bbS$. Substituting this representation into \eqref{proof:them110} and expanding the terms yields
\begin{align}
	\label{proof:them111} &\mathbb{E}\!\left[\tr \big(\tilde{\bbu}_2^\top \tilde{\bbu}_2 \!\big)\!\right] = \mathbb{E}\!\left[\tr \big(\bbZ\big)\!\right]\\
	&+ \mathbb{E}\Big[\!\tr \Big(\! \sum_{k=1}^{K} \sum_{\kappa,\ell=k}^K\!\alpha_{\kappa} \alpha_{\ell} \bbS^{\kappa-1}\bbn^{(k)} {\bbn^{(k)}}^\top\! \bbS^{\ell-1}\!\Big)\!\Big], \nonumber 
\end{align}
where $\bbZ$ collects the expanding terms that contain the noise matrix $\bbn^{(k_1)} {\bbn^{(k_1)}}^\top$ and at least one error matrix $\bbE^{(k_2)}$. For the second term in \eqref{proof:them111}, we can re-write it as
\begin{align}\label{proof:them112}
	\tr \Big(\! \sum_{k=1}^{K} \sum_{\kappa,\ell=k}^K\!\alpha_{\kappa} \alpha_{\ell} \bbS^{\kappa-1}\mathbb{E}\Big[\!\bbn^{(k)} {\bbn^{(k)}}^\top\Big] \bbS^{\ell-1}\!\Big).
\end{align}
Since each entry of $\bbn^{(k)}$ is independent with zero mean and standard deviation $\eps$, we can compute $\mathbb{E}\Big[\bbn^{(k)} {\bbn^{(k)}}^\top\Big]$ as
\begin{align}\label{proof:them113}
	\mathbb{E}\Big[\bbn^{(k)} {\bbn^{(k)}}^\top\Big] = \eps^2 \bbI. 
\end{align}
By substituting \eqref{proof:them113} into \eqref{proof:them112}, we can represent \eqref{proof:them112} as
\begin{align}\label{proof:them114}
	\eps^2 \tr \Big(\! \sum_{k=1}^{K} \sum_{\kappa,\ell=k+1}^K\!\alpha_{\kappa} \alpha_{\ell} \bbS^{k+\ell-2} \bbI\!\Big).
\end{align}
By further using the inequality \eqref{proof:them16}, we can upper bound it by
\begin{align}\label{proof:them115}
	\eps^2 \Big\|\!\sum_{k=1}^{K} \sum_{\kappa,\ell=k}^K\!\alpha_{\kappa} \alpha_{\ell} \bbS^{\kappa+\ell-2}\Big\|_2 {\rm tr} (\bbI).
\end{align}
We consider the matrix norm in \eqref{proof:them115}. For any matrix $\bbA$, a standard way to bound $\|\bbA\|_2$ is to obtain the inequality $\|\bbA\bba\|_2 \le A \|\bba\|_2$ that holds for any vector $\bba$ \cite{Meyer2000}. In this context, $A$ is the upper bound satisfying $\|\bbA\|_2 \le A$. Following this rationale, we consider the GFT of $\bba$ over $\bbS$ as $\bba = \sum_{i=1}^n \hat{a}_i \bbv_i$ and have
\begin{align}
	\label{proof:them116} & \Big\| \!\sum_{k\!=\!1}^K\! \sum_{\kappa, \ell\!=\!r}^K\!\!\! \alpha_\kappa \alpha_\ell \bbS^{\kappa\!+\!\ell\!-\!2} \bba \Big\|_2^2\!\!=\!\! \sum_{i\!=\!1}^n\! \hat{a}_i^2 \Big|\! \sum_{k\!=\!1}^K\! \sum_{\kappa, \ell\!=\!r}^K\!\!\! \alpha_\kappa \alpha_\ell \lambda_i^{\kappa\!+\!\ell\!-\!2}\Big|^2\!\!.
\end{align}
The expression inside the absolute value in \eqref{proof:them116} can be linked to the Lipschitz gradient of the frequency response over the air $f(\bblambda)$ as analyzed in the proof of Lemma \ref{lemma0:traceOperation}. Specifically, let $\bblambda_1 = \bblambda_2 = [\lambda_i,\ldots,\lambda_i]^\top$ be the multivariate frequency and the partial derivative of $f(\bblambda)$ w.r.t. the $r$th entry $\lambda^{(r)}$ at $\bblambda_{1:2,r}$ is
\begin{align}
	\label{proof:them117}
	\frac{\partial f(\bblambda_{1:2,r})}{\partial \lambda^{(r)}} \!=\! \sum_{k=r}^K \alpha_k \lambda_i^{k-1},\! ~\forall~r\!=\!1,\!\ldots,\!K.
\end{align}
The expression inside the absolute value in \eqref{proof:them116} can be written in the compact form and upper bounded by
\begin{align}
	\label{proof:them118}
	&\Big|\! \sum_{k=1}^K\! \sum_{\kappa, \ell=r}^K\!\!\! \alpha_\kappa \alpha_\ell \lambda_i^{\kappa\!+\!\ell\!-\!2}\Big| = \sum_{r=1}^K \Big(\frac{\partial f(\bblambda_{1:2,r})}{\partial \lambda^{(r)}}\Big)^2 \le C_L^2,
\end{align}
where the integral Lipschitz condition is used. By substituting \eqref{proof:them118} into \eqref{proof:them116} and the latter into \eqref{proof:them115}, we can upper bound the second term in \eqref{proof:them111} as
\begin{align}
	\label{proof:them119} \mathbb{E}\Big[\!\tr \Big(\! \sum_{k=1}^{K} \sum_{\kappa,\ell=k}^K\!\alpha_{\kappa} \alpha_{\ell} \bbS^{\kappa-1}\bbn^{(k)} {\bbn^{(k)}}^\top\! \bbS^{\ell-1}\!\Big)\!\Big] \le n C_L^2 \eps^2. 
\end{align}

For the first term in \eqref{proof:them111}, matrix $\bbZ$ comprises the sum of the remaining expansion terms, each of which contains the noise matrix $\bbn^{(k_1)} {\bbn^{(k_1)}}^\top$ and at least one error matrix $\bbE^{(k_2)}$. Since the frequency response $f(\bblambda)$ is bounded, the coefficients $\{ \alpha_k \}_{k=0}^K$ are also bounded. From the facts that $\| \bbS \|_2$ is bounded and $\mathbb{E}[\bbE^{(k)}] = (1-\mu)\bbS$, $\mathbb{E}\left[ {\bbE^{(k)}}^2 \right] = (1-\mu)^2\bbS^2+\delta^2 \bbD$ and $\mathbb{E}\Big[\!\bbn^{(k)} {\bbn^{(k)}}^\top\Big] = \eps^2 \bbI$, we have
\begin{gather} \label{proof:them120}
	\begin{split}
		\mathbb{E} \Big[{\rm tr}(\bbZ) \!\Big] \!\le\! \ccalO\big((1-\mu) \eps^2\big)+\ccalO\big(\delta^2 \eps^2\big). 
	\end{split}
\end{gather}

Finally, substituting the bounds for the first and second terms into \eqref{proof:thm12} and altogether into \eqref{proof:thm11}, we obtain
\begin{align} \label{proof:them125}
	&\mathbb{E}\left[\| \tilde{\bbu} - \bbu \|^2_2\right]\\
	&\le (n+K-1) C_L^2 \|\bbx\|_2^2 (1-\mu)^2 + ndC_L^2\|\bbx\|_2^2 \delta^2 + n C_L^2 \eps^2 \nonumber \\
	&+\!\ccalO\big((1\!-\!\mu)^3\big) \!+\! \ccalO\big(\delta^2 (1\!-\!\mu)\!\big) \!+ \!\ccalO\big((1\!-\!\mu) \eps^2\big)\!+\!\ccalO\big(\delta^2 \eps^2\big) \nonumber
\end{align}
completing the proof.


\section{Proof of Theorem 2}\label{Proof:Theorem2}

From the AirGNN architecture in (7) and the Lipschitz condition of the nonlinearity [cf. (18)], the output difference of the AirGNN can be upper bounded by
\begin{align} \label{eq:thm21}
	\big\|\tilde{\bbPhi} - \bbPhi\big\|_2 &= \big\| \sigma\Big(\sum_{f=1}^{F} \tilde{\bbu}^f_{L-1}\Big)\! -\!\sigma\Big(\sum_{f=1}^{F} \bbu^f_{L-1}\Big)\big\|_2\\
	&\le C_\sigma \big\| \sum_{f=1}^{F} \tilde{\bbu}^f_{L-1} - \sum_{f=1}^{F} \bbu^f_{L-1}\big\|_2, \nonumber
\end{align}
where $\tilde{\bbPhi}$ and $\bbPhi$ are concise notations of $\bbPhi_{\rm air}(\bbx;\!\bbS,\!\ccalH)$ and $\bbPhi(\bbx;\!\bbS,\!\ccalH)$, and $\tilde{\bbu}^f_{L-1}$ and $\bbu^f_{L-1}$ are the $f$th outputs of the AirGNN and nominal GNN at layer $L-1$. By applying the triangular inequality in \eqref{eq:thm21}, we have
\begin{equation} \label{eq:thm22}
	\begin{split}
		&\big\|\tilde{\bbPhi} - \bbPhi\big\|_2 \!\le\! C_\sigma\sum_{f=1}^{F} \| \tilde{\bbu}^f_{L-1}\! - \! \bbu^f_{L-1}\|_2 .
	\end{split}
\end{equation}
We consider each term $\| \tilde{\bbu}^f_{L-1}\! - \! \bbu^f_{L-1}\|_2$ separately. Denote by $\tilde{\bbu}_{L-1}^f = \tilde{\bbH}_{L}^{f}\tilde{\bbx}^f_{L-1}$ and $\bbu_{L-1}^f=\bbH_L^f\bbx_{L-1}^f$ concise notations of $\bbH_{{\rm air}, L}^{f}(\bbS)\tilde{\bbx}^f_{L-1}$ and $\bbH(\bbS)_{L}^{f}\bbx^f_{L-1}$. By adding and subtracting $\tilde{\bbH}_{L}^{f}\bbx^f_{L-1}$ inside the norm, we get
\begin{align} \label{eq:thm23}
	&\| \tilde{\bbu}^f_{L\!-\!1}\!\! - \!\! \bbu^f_{L\!-\!1}\|_2\!=\!\| \tilde{\bbH}_{L}^{f}\tilde{\bbx}^f_{L\!-\!1}\! - \!\tilde{\bbH}_{L}^{f}\bbx^f_{L\!-\!1} \!+\! \tilde{\bbH}_{L}^{f}\bbx^f_{L\!-\!1}\!- \! \bbH_L^f\bbx_{L\!-\!1}^f\|_2 \nonumber \\
	& \le \| \tilde{\bbH}_{L}^{f}\big(\tilde{\bbx}^f_{L-1}\! - \!\bbx^f_{L-1}\big)\|_2 \!+\!  \| \tilde{\bbH}_{L}^{f}\bbx^f_{L-1}\!- \! \bbH_L^f\bbx_{L-1}^f\|_2.
\end{align}
From the facts that $|f(\bblambda)|\le 1$ 
and that the noise component $\bbN_{{\rm air}, L}^f$ of $\bbH_{{\rm air}, L}^{f}(\bbS)$ is same for either $\tilde{\bbx}^f_{L-1}$ or $\bbx^f_{L-1}$, we can upper bound the first term in \eqref{eq:thm23} as
\begin{align} \label{eq:thm24}
	\| \tilde{\bbH}_{L}^{f}\big(\tilde{\bbx}^f_{L-1}\! - \!\bbx^f_{L-1}\big)\|_2 \le \| \tilde{\bbx}^f_{L-1}\! - \!\bbx^f_{L-1} \|_2.
\end{align}
By substituting \eqref{eq:thm24} into \eqref{eq:thm23} and altogether into \eqref{eq:thm22}, we have
\begin{align} \label{eq:thm25}
	&\big\|\tilde{\bbPhi} - \bbPhi \big\|_2 \\
	&\le C_\sigma\!\!\sum_{f=1}^{F}\! \| \tilde{\bbH}_{L}^{f}\bbx^f_{L-1}\!- \! \bbH_L^f\bbx_{L-1}^f\|_2 \!+\!C_\sigma\!\! \sum_{f\!=\!1}^{F}\! \| \tilde{\bbx}^f_{L\!-\!1}\! - \!\bbx^f_{L\!-\!1} \|_2.\nonumber
\end{align}

We now observe a recursion where the output difference of $\ell$th layer output is bounded by the output difference of $(\ell-1)$th layer output with an extra term (the first term in \eqref{eq:thm25}). Following the same process of \eqref{eq:thm21}-\eqref{eq:thm25}, we get
\begin{align} \label{eq:thm26}
	&\| \tilde{\bbx}^f_{L-1} - \bbx^f_{L-1} \|_2 \\
	&\le\!\! C_\sigma\!\!\!\sum_{g\!=\!1}^{F}\! \| \tilde{\bbH}_{L\!-\!1}^{fg}\bbx^g_{L\!-\!2}\!\!- \! \bbH_{L-1}^{fg}\bbx_{L-2}^g\|_2 \!+\! C_\sigma\!\!\sum_{g\!=\!1}^{F}\! \| \tilde{\bbx}^g_{L-2}\! - \!\bbx^g_{L-2} \|_2.\nonumber
\end{align}
By substituting \eqref{eq:thm26} into \eqref{eq:thm25}, we get
\begin{align} \label{eq:thm27}
	&\big\|\tilde{\bbPhi} - \bbPhi\big\|_2 \le C_\sigma\!\sum_{f=1}^{F} \| \tilde{\bbH}_{L}^{f}\bbx^f_{L-1}\!- \! \bbH_L^f\bbx_{L-1}^f\|_2\nonumber\\
	&+C_\sigma^2\sum_{f,g=1}^{F} \| \tilde{\bbH}_{L-1}^{fg}\bbx^g_{L-2}\!- \! \bbH_{L-1}^{fg}\bbx_{L-2}^g\|_2 \\
	&+ C_\sigma^2 F \sum_{g=1}^{F} \| \tilde{\bbx}^g_{L-2}- \bbx^g_{L-2} \|_2. \nonumber
\end{align}
Unrolling this recursion until the input layer yields
\begin{align} \label{eq:thm28}
	&\big\|\tilde{\bbPhi}-\bbPhi\big\|_2 \!\le\! \frac{C_\sigma}{F}\sum_{f,g=1}^{F}\! \| \tilde{\bbH}_{L}^{f}\bbx^f_{L-1}\!- \! \bbH_L^f\bbx_{L-1}^f\|_2 \nonumber \\
	&+ \sum_{\ell = 2}^{L-1} C_\sigma^{L+1-\ell}F^{L - 1 - \ell}\! \sum_{f,g=1}^{F}\! \| \tilde{\bbH}_{\ell}^{fg}\bbx^g_{\ell\!-\!1}\!- \! \bbH_{\ell}^{fg}\bbx_{\ell\!-\!1}^g\|_2 \\
	&+ C_\sigma^L F^{L-3}\!\!\! \sum_{f,g=1}^{F} \| \tilde{\bbH}_{1}^{f}\bbx_{0}\!-\! \bbH_1^f\bbx_{0} \|_2, \nonumber
\end{align}
where $\bbx_0 = \bbx$ is the input signal. From \eqref{eq:thm28} and the inequality of arithmetic and geometric means, the square of $\big\|\tilde{\bbPhi}- \bbPhi\big\|_2$ is bounded as
\begin{align} \label{eq:thm29}
	&\big\|\tilde{\bbPhi} \!-\! \bbPhi\big\|^2_2 \le L C_\sigma^2\!\sum_{f,g=1}^{F} \| \tilde{\bbH}_{L}^{f}\bbx^f_{L-1}\!- \! \bbH_L^f\bbx_{L-1}^f\|_2^2 \\
	&+ L \sum_{\ell = 2}^{L-1}\!C_\sigma^{2L+2-2\ell}\!F^{2L - 2\ell} \!\sum_{f,g=1}^{F}\!\| \tilde{\bbH}_{\ell}^{fg}\bbx^g_{\ell\!-\!1}\!- \! \bbH_{\ell}^{fg}\bbx_{\ell\!-\!1}^g\|_2^2 \nonumber\\
	&+ L C_\sigma^{2L} \!F^{2L-4}\!\sum_{f,g=1}^{F}\! \| \tilde{\bbH}_{1}^{f}\bbx_{0}\!-\! \bbH_1^f\bbx_{0} \|_2^2.\nonumber
\end{align}

We now consider $\| \tilde{\bbH}_{L}^{f}\bbx^f_{L-1}\!- \! \bbH_{L}^{f}\bbx_{L-1}^f\|_2^2$ in the first term of \eqref{eq:thm29}. Using the result of Theorem 1, we have
\begin{align} \label{eq:thm210}
	&\mathbb{E}\big[ \| \tilde{\bbH}_{L}^{f}\bbx^f_{L\!-\!1}\!\!- \! \bbH_{L}^{f}\bbx_{L\!-\!1}^f\|_2^2 \big] \!\!\le\!\! \big(\!C_1 (\!1\!-\!\mu\!)^2 \!\!+\! C_2 \delta^2\big)\! \| \bbx^f_{L\!-\!1} \|_2^2 \!\!+\! C_3 \eps^2 \nonumber \\
	&+\! \Big(\!\ccalO((1\!-\!\mu)^3) \!+\! \ccalO(\delta^2 (1\!-\!\mu)) \!+ \!\ccalO\big((1\!-\!\mu) \eps^2\big)\!+\!\ccalO\big(\delta^2 \eps^2\big)\!\Big). 
\end{align}
where $C_1, C_2$ and $C_3$ are stability constants. For the square norm $\| \bbx^f_{L-1} \|_2^2$ in the bound of \eqref{eq:thm210}, we observe
\begin{align}\label{eq:thm211}
	&\| \bbx^f_{L-1} \|_2^2\le \| \sigma\big(\sum_{g=1}^F\bbu_{L-2}^{fg}\big) \|_2^2 \\
	&\le C_\sigma^2 F \sum_{g=1}^F \| \bbu_{L-2}^{fg} \|_2^2 \le C_\sigma^2 F \sum_{g=1}^F \left\| \bbx_{L-2}^{g} \right\|_2^2, \nonumber
\end{align}          
where the Lipschitz condition of the nonlinearity, the triangular inequality and the fact $|f(\bblambda)|\le 1$ are used. Following this process yields
\begin{equation}\label{eq:thm212}
	\begin{split}
		\| \bbx^f_{L-1} \|_2^2\le C_\sigma^{2L-2} F^{2L-4} \left\| \bbx \right\|_2^2 .
	\end{split}
\end{equation}        
By substituting \eqref{eq:thm212} into \eqref{eq:thm210}, we get
\begin{align} \label{eq:thm213}
	&\mathbb{E}\!\big[ \| \tilde{\bbH}_{L}^{f}\bbx^f_{L\!-\!1}\!\!-\!\! \bbH_{L}^{f}\bbx_{L\!-\!1}^f\!\|_2^2 \big]\! \!\!\le\!\! C \!\big(\!C_1\! (\!1\!\!-\!\!\mu)^2 \!\!+\!\! C_2 \delta^2\!\big)\! \| \bbx^f_{L\!-\!1}\! \|_2^2 \!+\! C C_3 \eps^2 \nonumber \\
	&+\!\! \Big(\!\!\ccalO((1\!-\!\mu)^3\!) \!\!+\!\! \ccalO(\delta^2 (1\!-\!\mu)\!) \!\!+\! \!\ccalO\big(\!(1\!-\!\mu) \eps^2\big)\!\!+\!\!\ccalO\big(\delta^2 \eps^2\big)\!\!\Big)\!, 
\end{align}
where $C=C_\sigma^{2L-2} F^{2L-4}$ is the stability constant resulting from \eqref{eq:thm212}. We can bound $\| \tilde{\bbH}_{\ell}^{fg}\bbx^g_{\ell-1}\!- \! \bbH_{\ell}^{fg}\bbx_{\ell-1}^g\|_2^2$ and $\| \tilde{\bbH}_{1}^{f}\bbx^f_{0}- \bbH_1^f\bbx_{0}^f \|_2^2$ in the second and third terms of \eqref{eq:thm29} by following \eqref{eq:thm210}-\eqref{eq:thm213} since they have a similar form. By substituting these bounds into \eqref{eq:thm29} and using the linearity of the expectation in \eqref{eq:thm29}, we complete the proof
\begin{align}
	&\mathbb{E}\big[ \big\|\tilde{\bbPhi} \!-\! \bbPhi\big\|^2_2\big] \!\le\! C \big(C_1 (1\!-\!\mu)^2 \!+\! C_2 \delta^2\big) \| \bbx \|_2^2 \!+\! C C_3 \eps^2 \\
	&+\! \Big(\ccalO((1\!-\!\mu)^3) \!\!+\!\! \ccalO(\delta^2 (1\!-\!\mu)) \!+ \!\ccalO\big((1\!-\!\mu) \eps^2\big)\!+\!\ccalO\big(\delta^2 \eps^2\big)\!\Big), \nonumber
\end{align}
where $C=C_\sigma^{2L} L^2 F^{2L-2}$ is the stability constant resulting from the nonlinearity and multi-layered architecture.


\section{Proof of Theorem 3}\label{Proof:Theorem3}

The randomness of the objective function $\ccalL(\mathcal{R},\bbS,\ccalA | \{\bbS^{(k)}\}_{k})$ for AirGNNs with CSI is determined by the sequence of AirGSOs $\{\bbS^{(k)}\}_{k}$ sampled from the RES model with probability $\rho$ [Def. 4]. This indicates that sampling an objective function $\ccalL(\mathcal{R}_t,\bbS,\ccalA_t | \{\bbS_t^{(k)}\}_{k})$ is equivalent to sampling an AirGNN architecture $\bbPhi_{\rm sto}(\cdot, \bbS, \ccalA | \{\bbS_t^{(k)}\}_{k})$ over the dataset $\ccalR_t$. With this observation in mind, we compare Algorithm 2 with Algorithm 4 to find that steps $3-7$ in the proposed training procedure [Alg. 2] is equivalent to step $4$ in the SGD on problem (40) [Alg. 4]. Since the other steps between Algorithms 2 and 4 correspond to each other, the proposed training procedure of AirGNNs with CSI is equivalent to the SGD on problem (40) with the training data batch-size $|R_t|$ and the AirGSO batch-size $1$.

For the objective function $\ccalL(\mathcal{R},\bbS,\ccalA | \bbh, \bbn)$ of AirGNNs without CSI, we follow a similar process for proof. Specifically, $\ccalL(\mathcal{R},\bbS,\ccalA | \bbh, \bbn)$ is determined by the AirGNN architecture, which is, in turn, determined by the channel coefficients $\bbh$ and the Gaussian noise $\bbn$. This indicates that sampling an objective function $\ccalL(\mathcal{R}_t,\bbS,\ccalA_t | \bbh_t, \bbn_t)$ is equivalent to sampling an AirGNN architecture $\bbPhi_{\rm air}(\cdot, \bbS, \ccalA | \bbh_t, \bbn_t)$, i.e., sampling $\bbh_t$ and $\bbn_t$, over the dataset $\ccalR_t$. By leveraging this observation and comparing Algorithm 3 with Algorithm 4, steps $3$-$6$ in the proposed training procedure [Alg. 3] is equivalent to step $4$ in the SGD on problem (41) [Alg. 4]. Since the other steps are the same, we conclude that the proposed training procedure of AirGNNs without CSI is equivalent to the SGD on problem (41) with the training data batch-size $|R_t|$ and the communication channel batch-size $1$.


\section{Proof of Theorem 4}\label{Proof:Theorem4}

By using the Taylor expansion of $\mathbb{E} \left[ \bar{\ccalL}(\ccalA_{t+1}) \right]$ at $\ccalA_t$, we can represent $\mathbb{E} \left[ \bar{\ccalL}(\ccalA_{t+1}) \right]$ as
\begin{align} \label{eq:prthm11}
	&\mathbb{E}\!\left[ \bar{\ccalL}(\ccalA_{t\!+\!1})\right] \!=\! \mathbb{E}\!\left[ \bar{\ccalL}(\ccalA_{t}) \right. \!+\! \nabla_\ccalA \bar{\ccalL}(\ccalA_{t})^\top \!\left(\ccalA_{t+1}-\ccalA_t\right) \\
	&~~~~~~~~~~~~~~+\frac{1}{2}\left( \ccalA_{t+1}\!-\!\ccalA_t\right)^\top \nabla^2_\ccalA \bar{\ccalL} (\tilde{\ccalA}_{t}) \left( \ccalA_{t+1}-\ccalA_t\right) \big],\! \nonumber
\end{align}
where $\tilde{\ccalA}_{t}$ is in the line segment joining $\ccalA_{t+1}$ and $\ccalA_{t}$, which truncates the expansion series at the Hessian term. From the fact that $\bbv^\top \bbM \bbv \le \lambda_{\max}(\bbM) \| \bbv \|_2^2$ for any vector $\bbv$ and matrix $\bbM$ with $\lambda_{\max}(\bbM)$ the maximal eigenvalue of $\bbM$ and the Lipschitz continuity $\| \nabla^2_\ccalA \bar{\ccalL} (\tilde{\ccalA}_{t}) \|_2 \le C_\ell$ of Assumption 2, we have
\begin{align} \label{eq:prthm12}
	&\mathbb{E}\!\left[ \bar{\ccalL}(\ccalA_{t+1})\right] \\
	&\le \mathbb{E}\big[ \bar{\ccalL}(\ccalA_{t})\! +\! \nabla_\ccalA \bar{\ccalL}(\ccalA_{t})^\top ( \ccalA_{t+1}\!-\!\ccalA_t) \!+\! \frac{C_\ell}{2}\! \| \ccalA_{t+1}\!-\!\ccalA_t \|_2^2\big]. \nonumber
\end{align}
Substituting the update rule of the AirGNN $\ccalA_{t+1} = \ccalA_t - \eta_t \nabla_\ccalA \ccalL(\ccalR_t, \bbS, \ccalA_t | \bbh_t, \bbn_t)$ into \eqref{eq:prthm12} yields\footnote{This convergence proof considers the training procedure of AirGNNs without CSI $\bbPhi_{\rm air}(\cdot, \bbS, \ccalA | \bbh_t, \bbn_t)$. The convergence of AirGNNs with CSI $\bbPhi_{\rm sto}(\cdot, \bbS, \ccalA | \{\bbS_t^{(k)}\}_{k})$ can be proved by following the same process.}
\begin{align} \label{eq:prthm125}
	\mathbb{E}\!\left[ \bar{\ccalL}(\ccalA_{t+1})\right]\! &\le \mathbb{E}\big[ \bar{\ccalL}(\ccalA_{t}) \!+\!\! \frac{\eta_t^2 C_\ell}{2} \| \nabla_\ccalA \ccalL(\ccalR_t, \bbS, \ccalA_t | \bbh_t, \bbn_t) \|_2^2 \nonumber\\
	&-\! \eta_t \nabla_\ccalA \bar{\ccalL}(\ccalA_t)^\top\! \nabla_\ccalA \ccalL(\ccalR_t, \bbS, \ccalA_t | \bbh_t, \bbn_t) \big]\!.   
\end{align}
By using the linearity of expectation and the identity $\mathbb{E}[\nabla_\ccalA \ccalL(\ccalR_t, \bbS, \ccalA_t | \bbh_t, \bbn_t)]=\mathbb{E}[\nabla_\ccalA \bar{\ccalL}(\ccalA_t)]$, we can re-write \eqref{eq:prthm125} as
\begin{align} \label{eq:prthm13}
	\mathbb{E}\left[ \bar{\ccalL}( \ccalA_{t+1})\right] &\le \mathbb{E}\big[ \bar{\ccalL}( \ccalA_t) ] - \eta_t \mathbb{E}\big[\| \nabla_\ccalA \bar{\ccalL}(\ccalA_t) \|_2^2\big] \\
	&+ \frac{\eta_t^2 C_\ell}{2} \mathbb{E} \big[\| \nabla_\ccalA \ccalL(\ccalR_t, \bbS, \ccalA_t | \bbh_t, \bbn_t) \|_2^2\big].   \nonumber  
\end{align}
From the gradient bound of Assumption 3, we can upper bound the third term of \eqref{eq:prthm13} as
\begin{equation} \label{eq:prthm14}
	\begin{split}
		\frac{\eta_t^2 C_\ell}{2} \mathbb{E} \big[\| \nabla_\ccalA \ccalL(\ccalR_t, \bbS, \ccalA_t | \bbh_t, \bbn_t) \|_2^2\big] \le \frac{\eta_t^2 C_\ell C_g^2}{2}.
	\end{split}
\end{equation}
By substituting \eqref{eq:prthm14} into \eqref{eq:prthm13}, we have
\begin{equation} \label{eq:prthm15}
	\begin{split}
		\mathbb{E}[\| \nabla_\ccalA \bar{\ccalL}(\ccalA_t) \|_2^2] \le \frac{1}{\eta_t} \mathbb{E}\big[ \bar{\ccalL}( \ccalA_t) \!-\! \bar{\ccalL}(\ccalA_{t+1}) \big] \!+\! \frac{\eta_t C_\ell C_g^2}{2}.
	\end{split}
\end{equation}
Since \eqref{eq:prthm15} holds for all iterations $t = 0, \ldots, T-1$, by considering a constant step-size $\eta_t = \eta$ and summing up \eqref{eq:prthm15} over all iterations, we have
\begin{align}\label{eq:prthm155}
	&\sum_{t=0}^{T-1} \mathbb{E}[\| \nabla_\ccalA \bar{\ccalL}(\ccalA_t) \|^2] \\
	&\le \frac{1}{\eta} \mathbb{E}\big[ \bar{\ccalL}( \ccalA_0) - \bar{\ccalL}(\ccalA_{T}) \big]+ \frac{ \eta T C_\ell C_g^2}{2}.\nonumber
\end{align}
For the optimal parameters $\ccalA^*$, we have $\bar{\ccalL} (\ccalA^*)\! \le\! \bar{\ccalL} ( \ccalA_{T})$. This allows to bound \eqref{eq:prthm155} as
\begin{align} \label{eq:prthm16}
	&\min_t \mathbb{E}\big[\| \nabla_\ccalA \bar{\ccalL}(\ccalA_t) \|^2\big] \le \frac{1}{T} \sum_{t=0}^{T-1} \mathbb{E}\big[\| \nabla_\ccalA \bar{\ccalL}(\ccalA_t) \|^2\big] \nonumber \\
	& \le \frac{1}{T \eta} \left( \bar{\ccalL} (\ccalA_0) - \bar{\ccalL} (\ccalA^*) \right)+ \frac{ \eta C_\ell C_g^2}{2}.
\end{align}
By further setting the constant step-size as
\begin{equation} \label{eq:prthm165}
	\begin{split}
		\eta \!=\! \sqrt{\frac{2\left( \bar{\ccalL} (\ccalA_0) - \bar{\ccalL} (\ccalA^*)\right)}{T C_\ell C_g^2 }}
	\end{split}
\end{equation}
and substituting it into \eqref{eq:prthm16}, we complete the proof
\begin{equation} \label{eq:prthm17}
	\begin{split}
		&\min_t \mathbb{E}[\| \nabla_\ccalA \bar{\ccalL}(\ccalA_t) \|^2] \le \frac{C}{\sqrt{T}},
	\end{split}
\end{equation}
where $C\!=\! \sqrt{2\!\left( \bar{\ccalL} (\ccalA_0) - \bar{\ccalL} (\ccalA^*)\right)\!C_\ell}C_g$ is a constant.

\section{Proof of Theorem 5}\label{Proof:Theorem5}

Since the AirGNN is constructed with the AirGF, we first characterize the variance of the AirGF output and then extend the result to the AirGNN. Let $\tilde{\bbu} = \bbH_{\rm air}(\bbS)\bbx = \tilde{\bbu}_1 + \tilde{\bbu}_2$ be the AirGF output [cf. (6)], where $\tilde{\bbu}_1 = \bbP_{\rm air}(\bbS, \bbx)$ is the signal component and $\tilde{\bbu}_2 = \bbN_{\rm air}(\bbS, \bbx)$ is the noise component, and $\bar{\bbu} = \mathbb{E}[\bbH_{\rm air}(\bbS)\bbx] = \bar{\bbu}_1 + \bar{\bbu}_2$ be the expected AirGNN output, where $\bar{\bbu}_1$ is the expected signal component and $\bar{\bbu}_2$ is the expected noise component. Since the noise $n$ follows a zero-mean Gaussian distribution, we have $\bar{\bbu}_2 = \mathbb{E}[\tilde{\bbu}_2] = \bb0$. In this context, the variance of the AirGNN output is 
\begin{align} \label{eq:prop331}
	&{\rm var}[\tilde{\bbu}] = \mathbb{E} \left[ \tr \left( \tilde{\bbu} \tilde{\bbu}^{\top} - \bar{\bbu} \bar{\bbu}^{\top} \right) \right]\\
	& =\! \mathbb{E} \left[ \tr \left( \tilde{\bbu}_1 \tilde{\bbu}_1^{\top} - \bar{\bbu}_1 \bar{\bbu}_1^{\top} \right) \right] \!+\! \mathbb{E}[\tr \big(\tilde{\bbu}_2^\top \tilde{\bbu}_2 \big)] \nonumber \\
	& =\!\!\! \sum_{k, \ell=0}^K\!\!\!\! \alpha_k \alpha_\ell \!\left( \mathbb{E}\! \left[ \tr \!\left( W(k,\ell) \right) \right] \!-\!  \mathbb{E}\! \left[ \tr\! \left( \bar{W}(k,\ell) \right) \right] \right) \!+\! \mathbb{E}[\tr \big(\tilde{\bbu}_2^\top \tilde{\bbu}_2 \big)], \nonumber
\end{align}
where $W(k,\ell)= \prod_{\kappa=0}^k\bbS^{(\kappa)} \bbx \bbx^{\top} \prod_{\kappa=0}^\ell \bbS^{(\ell - \kappa)}$ and $\bar{W}(k, \ell)= \bar{\bbS}^k  \bbx \bbx^{\top} \bar{\bbS}^\ell$, $\bar{\bbS} = \mathbb{E}[\bbS^{(k)}]$ is the expected shift operator, $\bbS^{(0)}=\bbI$ is the identity matrix by default, and the linearity of the trace and expectation and the symmetry of the shift operators $\bbS^{(k)}$ and $\barbS$ are used. In the following proof, we denote by $\lceil k\ell \rceil = \max(k,\ell)$ and $\lfloor k\ell \rfloor = \min(k,\ell)$ to simplify notation. By representing $\bbS^{(k)} = \barbS + \bbE^{(k)}$ with $\bbE^{(k)}$ the deviation of $\bbS^{(k)}$ from the mean $\barbS$ and substituting the latter into $W(k,\ell)$, we have 
\begin{align}
	\label{eq:prop33Ts1} &\mathbb{E}\left[ W(k,\ell)\right] = \mathbb{E}\left[ (\barbS + \bbE^{(k)})\cdots \bbx \bbx^\top \cdots (\barbS + \bbE^{(\ell)})\right]\\
	&=\!{\barbS}^{k}\bbx \bbx^\top\! {\barbS}^{\ell} \!\!+\!\! \mathbb{E}\!\!\left[\! \sum_{r=1}^{\lfloor k\ell \rfloor} \!\barbS^{k\!-\!r}\bbE^{(\!r\!)} \barbS^{r\!-\!1}\bbx \bbx^\top\! \barbS^{r\!-\!1}\bbE^{(\!r\!)} \barbS^{\ell\!-\!r}\!\right] \!\!\!+\!\! \mathbb{E}\!\left[ \bbC_{k\ell}\right]\!, \nonumber 
\end{align}
where the fact $\mathbb{E}[\bbE^{(k)}] = \bb0$ is used. The first term in \eqref{eq:prop33Ts1} results 
from the maximal power of $\barbS$ in the products; the second term captures all cross-products with at most two deviation matrices $\bbE^{(r)}$, where we note that $\mathbb{E}\left[\bbE^{(k)} \bbE^{(\ell)}\right] = \mathbb{E}[\bbE^{(k)}]\mathbb{E}[\bbE^{(\ell)}] = \bb0$ for $k \ne \ell$ due to independence and 
the terms with $r > \lfloor k \ell\rfloor$ are null due to the presence of a single expectation $\mathbb{E}[\bbE^{(r)}] = \bb0$; the third term $\bbC_{k\ell}$ is the sum of the remaining terms. By substituting $\tr ( \mathbb{E}[ \bar{W}(k,\ell)]) = \tr ( {\barbS}^{k}\bbx \bbx^\top {\barbS}^{\ell})$ and \eqref{eq:prop33Ts1} into \eqref{eq:prop331}, we have
\begin{align} \label{eq:prop332}
	&{\rm var}[\bbu] =\sum_{k=0}^K \sum_{\ell=0}^K \alpha_k \alpha_\ell \tr \left( \mathbb{E}\left[ \bbC_{k\ell}\right] \right)\!+\! \mathbb{E}[\tr \big(\tilde{\bbu}_2^\top \tilde{\bbu}_2 \big)]\\
	& \! +\!\sum_{k=1}^K \sum_{\ell=1}^K\! \alpha_k \alpha_\ell \tr\! \left(\!\! \mathbb{E}\!\left[\! \sum_{r=1}^{\lfloor k\ell \rfloor} \barbS^{k-r}\bbE^{(r)} \barbS^{r\!-\!1}\bbx \bbx^\top \barbS^{r-1}\bbE^{(r)} \barbS^{\ell-r}\!\right] \!\right)\!\!\nonumber.
\end{align}
We now analyze the three terms in \eqref{eq:prop332} separately. For the following analysis, we similarly need the inequality
\begin{gather} \label{eq:prop3ineq0}
	\begin{split}
		{\rm tr} (\bbA \bbB) \le \frac{\| \bbA + \bbA^\top \|_2}{2}{\rm tr}(\bbB) \le \| \bbA \|_2 {\rm tr}(\bbB)
	\end{split}
\end{gather}
for any square matrix $\bbA$ and positive semi-definite matrix $\bbB$.

$\textbf{Third term in \eqref{eq:prop332}.}$ We bring the trace inside the expectation due to their linearity and leverage the trace cyclic property $\tr(\bbA\bbB\bbC) = \tr(\bbC\bbA\bbB) = \tr(\bbB\bbC\bbA)$ to have
\begin{align}
	\label{eq:prop33} &\mathbb{E}\left[ \tr\! \left( \sum_{k,\ell=1}^K \alpha_k \alpha_\ell \sum_{r=1}^{\lfloor k\ell \rfloor} \barbS^{k-r}\bbE^{(r)} \barbS^{r-1}\bbx \bbx^\top \barbS^{r-1}\bbE^{(r)} \barbS^{\ell-r} \right) \!\right]\nonumber \\
	&=\!\mathbb{E}\!\left[ \! \sum_{r=1}^K\! \tr\! \left( \!\sum_{k, \ell\!=\!r}^K \!\!\! \alpha_k \alpha_\ell \bbE^{(\!r\!)} \barbS^{k+\ell-2r}\bbE^{(\!r\!)} \barbS^{r-1}\bbx \bbx^\top\! \barbS^{r\!-\!1}\!\!\! \right)\!\!\right]\!\!\!,
\end{align}
where the terms are rearranged to change the sum limits. Since both matrices $\sum_{k, \ell=r}^K \alpha_k \alpha_\ell \bbE^{(r)} \barbS^{k+\ell-2r}\bbE^{(r)} = \big(\sum_{k=r}^K \alpha_k \barbS^{k-r}\bbE^{(r)}\big)^\top \big(\sum_{k=r}^K \alpha_k \barbS^{k-r}\bbE^{(r)}\big)$ and $\barbS^{r-1}\bbx \bbx^\top \barbS^{r-1}$ are positive semi-definite, we use the Cauchy-Schwarz inequality $\tr(\bbA \bbB) \le \tr(\bbA)\tr(\bbB)$ \cite{Zhang1999} to upper bound \eqref{eq:prop33} by
\begin{align}
	\label{eq:prop34} \!\mathbb{E}\!\left[ \!\sum_{r=1}^K\! \sum_{k, \ell=r}^K \!\!\alpha_k \alpha_\ell \tr\! \left( \bbE^{(r)} \barbS^{k+\ell-2r}\bbE^{(r)}\! \right)\! \tr\! \left( \barbS^{r-1}\bbx \bbx^\top \barbS^{r-1} \right)\!\right]\!\!.
\end{align}
We now represent the signal $\bbx$ in the frequency domain over the expected graph. Specifically, let $\barbS = \barbV\barbLambda\barbV^\top$ be the eigendecomposition of $\barbS$ with eigenvectors $\barbV = [\barbv_1,\ldots,\barbv_n]^\top$ and eigenvalues $\barbLambda = \text{diag}(\bar{\lambda}_1, \ldots, \bar{\lambda}_n)$. By substituting the graph Fourier expansion $\bbx = \sum_{i=1}^n \hat{x}_i \barbv_i$ into $\tr\! \left( \barbS^{r-1}\bbx \bbx^\top \barbS^{r-1} \right)$, we get
\begin{gather}\label{eq:prop35}
	\tr\! \left( \barbS^{r-1}\bbx \bbx^\top \barbS^{r-1} \right) \!=\! \sum_{i=1}^n \hat{x}_i^2 \bar{\lambda}_i^{2r-2} \tr\! \left(\! \barbv_i \barbv_i^\top\! \right) \!=\! \sum_{i=1}^n \hat{x}_i^2 \bar{\lambda}_i^{2r\!-\!2},
\end{gather}
where $\tr(\bbv_i\bbv_i^\top) = 1$ for $i=1,\ldots,n$ due to the orthonormality of eigenvectors. By substituting \eqref{eq:prop35} into \eqref{eq:prop34}, we get
\begin{align}
	\label{eq:prop36} \sum_{i\!=\!1}^n\! \hat{x}_i^2 \mathbb{E}\!\!\left[ \! \sum_{r\!=\!1}^K \sum_{k, \ell=r}^K\! \alpha_k \alpha_\ell \bar{\lambda}_i^{2r-2} \tr\! \left( \bbE^{(r)} \barbS^{k+\ell-2r}\bbE^{(r)} \right)\!\right]\!.
\end{align}
We again use the trace cyclic property to write $\tr(\bbE_r\barbS^{k+\ell-2}\bbE_r) = \tr(\barbS^{k+\ell-2}\bbE_r^2)$ and the linearity of the expectation in \eqref{eq:prop36} to have 
\begin{align} \label{eq:prop37} 
	& \sum_{i = 1}^n\hat{x}_i^2\! \sum_{r\!=\!1}^K\!\tr\! \left( \sum_{k, \ell=r}^K \alpha_k \alpha_\ell \bar{\lambda}_i^{2r-2} \barbS^{k+\ell-2r} \mathbb{E}\!\left[ {\bbE^{(r)}}^2\right] \right).
\end{align}
By following the same steps in Lemma \ref{LemmmaErrorMatrix}, we have $\mathbb{E}\big[ {\bbE^{(k)}}^2 \big] =  \delta^2 \bbD$ with $\bbD$ the degree matrix of the graph. By substituting this result into \eqref{eq:prop37} and using inequality \eqref{eq:prop3ineq0} with the positive semi-definite matrix $\bbD$, we have
\begin{align}
	\label{eq:prop39}
	& \delta^2\sum_{i = 1}^n\hat{x}_i^2\! \sum_{r\!=\!1}^K\!\tr\! \left( \sum_{k, \ell=r}^K \alpha_k \alpha_\ell \bar{\lambda}_i^{2r-2} \barbS^{k+\ell-2r}\bbD \right)\\
	& \le \delta^2 \sum_{i = 1}^n\hat{x}_i^2 \big\| \sum_{r\!=\!1}^K\! \sum_{k, \ell=r}^K\! \alpha_k \alpha_\ell \bar{\lambda}_i^{2r-2} \barbS^{k+\ell-2r} \big\|_2 \tr\! \left( \bbD \right), \nonumber
\end{align}
where $\tr\left( \bbD \right)=\sum_{i=1}^n d_i \le nd$ and $d$ is the maximal degree of the graph. 

We then upper bound the filter matrix norm in \eqref{eq:prop39}. We use the same procedure to bound the norm of a matrix $\bbA$, i.e., to bound the norm of $\|\bbA\bba\|_2$ as $\|\bbA\bba\|_2 \le A \|\bba\|_2$ for any vector $\bba$ \cite{Meyer2000}. Following this rationale, we consider the GFT of a vector $\bba$ on the expected graph as $\bba = \sum_{j=1}^n \hat{a}_j \barbv_j$ where $\{ \barbv_j \}_{j=1}^n$ are orthonormal. Then, we have
\begin{align}
	\label{eq:prop310} & \!\big\| \!\sum_{r\!=\!1}^K\!\! \sum_{k, \ell\!=\!r}^K\!\! \alpha_k \alpha_\ell\! \bar{\lambda}_i^{2r\!-\!2}\! \barbS^{k\!+\!\ell\!-\!2r}\! \bba \big\|_2^2 \!\!=\!\!\! \sum_{j\!=\!1}^n\!\! \hat{a}_j^2 \big|\! \sum_{r\!=\!1}^K\!\! \sum_{k, \ell\!=\!r}^K\!\!\! \alpha_k \alpha_\ell\! \bar{\lambda}_i^{2r\!-\!2}\! \bar{\lambda}_j^{k\!+\!\ell\!-\!2r}\big|^2\!\!.
\end{align}
Consider now the expression inside the absolute value in \eqref{eq:prop310}. This expression can be similarly linked to the Lipschitz gradient of the filter frequency response over the air $f(\bblambda)$ [Def. 2]. Specifically, the partial derivative of $f(\bblambda)$ w.r.t. the $r$th entry $\lambda^{(r)}$ of $\bblambda$ is 
\begin{align}
	\label{eq:prop31051} 
	\frac{\partial f(\bblambda)}{\partial \lambda^{(r)}} \!=\! \sum_{k=r}^K \alpha_k \lambda_{k:(r+1)} \lambda_{(r-1):1}
\end{align}
for $r=1,\ldots,K$, where $\lambda_{K:(r+1)} = \lambda^{(K)} \cdots \lambda^{(r+1)}$ and $\lambda_{(r-1):1}=\lambda^{(r-1)}\cdots \lambda^{(1)}$. Consider two multivariate frequencies $\bar{\bblambda}_i = [\bar{\lambda}_i,\ldots,\bar{\lambda}_i]^\top$ and $\bar{\bblambda}_j = [\bar{\lambda}_j,\ldots,\bar{\lambda}_j]^\top$. The $r$th entry of the Lipschitz gradient is 
\begin{align}\label{eq:prop310515} 
	[\nabla_L h(\bar{\bblambda}_i, \bar{\bblambda}_j)]_r = \sum_{k=r}^K \alpha_k \bar{\lambda}_{j}^{k-r} \bar{\lambda}_{i}^{r-1}.
\end{align}
By comparing \eqref{eq:prop310} and \eqref{eq:prop310515}, we observe that the expression inside the absolute value in \eqref{eq:prop310} can be represented as the square norm of $K$ partial derivatives
\begin{align}\label{eq:prop311} 
	&\sum_{r\!=\!1}^K\! \sum_{k, \ell=r}^K\! \alpha_k \alpha_\ell \bar{\lambda}_i^{2r-2} \bar{\lambda}_j^{k+\ell-2r} \\
	&=  \sum_{r\!=\!1}^K [\nabla_L f(\bar{\bblambda}_i, \bar{\bblambda}_j)]_r^2 = \| \nabla_L f(\bar{\bblambda}_i, \bar{\bblambda}_j) \|_2^2. \nonumber
\end{align}
Since the filter frequency response $f(\bblambda)$ is integral Lipschitz w.r.t. $C_L$ [Def. 3], we can upper bound \eqref{eq:prop311} as
\begin{align}
	\label{eq:prop3115} & \big| \sum_{r\!=\!1}^K\! \sum_{k,\ell=r}^K\! \alpha_k \alpha_\ell \bar{\lambda}_i^{2r-2} \bar{\lambda}_j^{k+\ell-2r}\big|^2  \le C_L^4,
\end{align}
which implies the norm of the filter matrix in \eqref{eq:prop310} is upper bounded by $C_L^2$. By substituting this norm bound into \eqref{eq:prop39} and altogether into \eqref{eq:prop33}, we have
\begin{align}
	\label{eq:prop312} &\mathbb{E}\left[ \sum_{k, \ell=1}^K \alpha_k \alpha_\ell \sum_{r=1}^{\lfloor k\ell \rfloor} \tr \left( \barbS^{k-r}\bbE^{(r)} \barbS^{r-1}\bbx \bbx^\top \barbS^{r-1}\bbE^{(r)} \barbS^{\ell-r} \right) \right]\nonumber \\
	& \le n d C_L^2 \delta^2 \sum_{i=1}^N \hat{x}_i^2 = n d C_L^2 \| \bbx \|^2_2 \delta^2.
\end{align}

$\textbf{First term in \eqref{eq:prop332}.}$ Matrix $\bbC_{k\ell}$ comprises the sum of the remaining expansion terms. Each of these terms is a quadratic form of the error matrices $\bbE^{(k)}$, $\bbE^{(\ell)}$ with $k \neq \ell$; i.e., it is of the form $f_1(\barbS)\bbE^{(k)} f_2(\barbS)\bbE^{(k)}f_3(\barbS)\bbE^{(\ell)} f_4(\barbS)\bbE^{(\ell)}$ for some functions $f_1(\cdot), ..., f_4(\cdot)$ that depend on the expected shift operator and filter parameters. Each of these double-quadratic terms can be bounded by a factor containing at least two terms $\tr\left(\mathbb{E}\big[{\bbE^{(k)}}^2\big]\right)$ and $\tr\left(\mathbb{E}\big[{\bbE^{(\ell)}}^2\big]\right)$. Since the frequency response $f(\bblambda)$, the filter parameters $\{ \alpha_k \}_{k=0}^K$ and $\| \barbS \|_2$ are bounded, we can bound the first term in \eqref{eq:prop332} by
\begin{gather} \label{eq:prop3125}
	\begin{split}
		\mathbb{E} \left[\sum_{k, \ell=0}^K \alpha_k \alpha_\ell \bbC_{k\ell} \right] = \ccalO(\delta^4)\|\bbx\|_2^2.
	\end{split}
\end{gather}

$\textbf{Second term in \eqref{eq:prop332}.}$ Following the proof of Theorem 1, we have
\begin{align}\label{eq:prop3126}
	\mathbb{E}[\tr \big(\tilde{\bbu}_2^\top \tilde{\bbu}_2 \big)] \!\le\! n C_L^2 \eps^2 \!+\!\ccalO\big(\delta^2 \eps^2\big).
\end{align}

By substituting the results for the first term \eqref{eq:prop3125}, second term \eqref{eq:prop3126} and third term \eqref{eq:prop312} into \eqref{eq:prop332}, we can bound the variance of the filter output as 
\begin{align}\label{eq:prop3127}
	{\rm var} [ \tilde{\bbu} ] \!\le\! n d C_L^2 \| \bbx \|_2^2 \delta^2 \!+\! n C_L^2 \eps^2 \!+\! \ccalO(\delta^4) + \ccalO(\delta^2 \eps^2).
\end{align}

Since the AirGNN is a multi-layered architecture consisting of AirGFs and nonlinearities, the variance of the AirGNN output can be bounded as
\begin{equation} \label{eq:thm41}
	\begin{split}
		& \!{\rm var} \!\left[ \bbPhi_{\rm air}(\bbx,\!\bbS,\!\ccalA) \right] \!\!=\! {\rm var}\!\! \left[\! \sigma\!\!\left(\!\sum_{f\!=\!1}^{F}\! \tilde{\bbu}^f_{L-1}\!\!\right)\! \!\right] \!\!\!\le\!  {\rm var}\! \!\left[ \!\sum_{f\!=\!1}^{F}\! \tilde{\bbu}^f_{L\!-\!1}\!\! \right]\!\!,
	\end{split}
\end{equation}
where the non-expansiveness of the nonlinearity $\sigma(\cdot)$ from Assumption 4 is used. By representing the variance with the trace operator, we can rewrite \eqref{eq:thm41} as
\begin{align} \label{eq:thm42}
	{\rm var}\! \left[ \!\sum_{f=1}^{F}\! \tilde{\bbu}^f_{L-1}\! \right] &= \!\tr\! \left( \mathbb{E}\! \left[\!\big(\! \sum_{f=1}^{F}\! \tilde{\bbu}^f_{L-1}\big) \big(\! \sum_{f=1}^{F}\! \tilde{\bbu}^f_{L-1}\big)^\top\! \right]\! \right. \nonumber\\
	& \left. - \mathbb{E}\!\!\left[\!\sum_{f=1}^{F}\! \tilde{\bbu}^f_{L-1}\! \right]\!\! \mathbb{E}\!\!\left[\!\sum_{f=1}^{F}\! \tilde{\bbu}^f_{L-1}\! \right]^\top \right)\!. 
\end{align}
Denote by $\tilde{\bbu}_{L-1}^f = \tilde{\bbH}_{L}^{f} \tilde{\bbx}^f_{L-1}$ and $\barbu_{L-1}^f=\barbH_L^f \tilde{\bbx}_{L-1}^f$ the concise notations of the AirGF output $\bbH_{{\rm air}, L}^f(\bbS) \tilde{\bbx}^f_{L-1}$ and its expected 
output $\mathbb{E}\left[ \bbH_{{\rm air}, L}^f(\bbS) \right]\tilde{\bbx}^f_{L-1}=\bbH_{L}^f(\barbS) \tilde{\bbx}^f_{L-1}$. By expanding \eqref{eq:thm42}, we get
\begin{align} \label{eq:thm43}
	& {\rm var}\! \left[ \!\sum_{f=1}^{F}\! \tilde{\bbu}^f_{L-1}\! \right] \!= \!\sum_{f\!=\!1}^{F}\! \sum_{g\!=\!1}^{F}\tr\! \left(\!  \mathbb{E} \!\left[ \!\tilde{\bbH}_{L}^{f} \tilde{\bbx}^f_{L-1}\! \left( \tilde{\bbH}_{L}^{g} \tilde{\bbx}^{g}_{L-1} \right)^\top \!\right.  \right] \nonumber\\
	& \quad \quad \quad \quad \quad \quad \left.  - \mathbb{E}\left[ \tilde{\bbH}_{L}^{f} \tilde{\bbx}^f_{L-1} \right] \mathbb{E} \left[ \tilde{\bbH}_{L}^{g} \tilde{\bbx}^{g}_{L-1} \right]^\top \right).
\end{align}
By adding and subtracting $\barbH_{L}^{f} \tilde{\bbx}^f_{L-1} \left( \barbH_{L}^{g} \tilde{\bbx}^{g}_{L-1}\right)^\top$ inside the first expectation, \eqref{eq:thm43} becomes
\begin{align} \label{eq:thm44}
	& \! \sum_{f\!=\!1}^{F}\! \sum_{g\!=\!1}^{F}\!\tr\! \left(\! \mathbb{E} \!\left[ \!\tilde{\bbH}_{L}^{f}\! \tilde{\bbx}^f_{L-1}\! \left( \tilde{\bbH}_{L}^{g}\! \tilde{\bbx}^{g}_{L-1} \right)^\top \!-\!\barbH_{L}^{f}\! \tilde{\bbx}^f_{L-1}\! \left( \barbH_{L}^{g}\! \tilde{\bbx}^{g}_{L-1}\right)^\top \!\right. \right] \nonumber\\
	&  \!\!+\!\mathbb{E}\!\left[ \barbH_{L}^{f} \tilde{\bbx}^f_{L\!-\!1}\! \left( \barbH_{L}^{g} \tilde{\bbx}^{g}_{L\!-\!1}\right)^\top \right]\! \left.\!- \mathbb{E}\!\left[ \tilde{\bbH}_{L}^{f} \tilde{\bbx}^f_{L\!-\!1} \!\right]\! \mathbb{E}\! \left[ \tilde{\bbH}_{L}^{g} \tilde{\bbx}^{g}_{L\!-\!1} \right]^\top \right).
\end{align}
The preceding expression \eqref{eq:thm44} is composed of two groups of terms shown in the two separate lines. 

\textbf{First group of terms in \eqref{eq:thm44}.} When $f \ne g$ such that filters $\tilde{\bbH}_{L}^{f}$ and $\tilde{\bbH}_{L}^{g}$ are independent, we have
\begin{equation} \label{eq:thm455}
	\begin{split}
		& \!\tr\! \left(\! \mathbb{E} \!\left[ \!\tilde{\bbH}_{L}^{f} \tilde{\bbx}^f_{L\!-\!1}\! \left( \tilde{\bbH}_{L}^{g} \tilde{\bbx}^{g}_{L\!-\!1} \right)^\top \!-\!\barbH_{L}^{f} \tilde{\bbx}^f_{L-1}\! \left( \barbH_{L}^{g} \tilde{\bbx}^{g}_{L\!-\!1}\right)^\top \right] \right)= 0.
	\end{split}
\end{equation}
We then use \eqref{eq:thm455} to derive the upper bound
\begin{align} \label{eq:thm45}
	& \!\sum_{f\!=\!1}^{F}\! \sum_{g\!=\!1}^{F}\!\tr\! \left(\! \mathbb{E} \!\left[ \!\tilde{\bbH}_{L}^{f}\! \tilde{\bbx}^f_{L-1}\! \left(\! \tilde{\bbH}_{L}^{g}\! \tilde{\bbx}^{g}_{L\!-\!1} \!\right)^\top \!-\!\barbH_{L}^{f}\! \tilde{\bbx}^f_{L-1}\! \left( \barbH_{L}^{g}\! \tilde{\bbx}^{g}_{L-1}\right)^\top \right]\! \right)\!\nonumber\\
	& = \!\sum_{f\!=\!1}^{F}\!\tr\! \left( \mathbb{E} \!\left[ \!\tilde{\bbH}_{L}^{f}\! \tilde{\bbx}^f_{L\!-\!1}\! \left(\! \tilde{\bbH}_{L}^{f}\! \tilde{\bbx}^{f}_{L\!-\!1} \!\right)^\top \!-\!\barbH_{L}^{f}\! \tilde{\bbx}^f_{L\!-\!1}\! \left( \!\barbH_{L}^{f}\! \tilde{\bbx}^{f}_{L\!-\!1}\!\right)^\top \!\right]\! \right)\! \nonumber\\
	& \le \delta^2 \Delta \sum_{f=1}^{F}\! \mathbb{E}[ \| \tilde{\bbx}^f_{L\!-\!1} \|_2^2 ] \!+\! n C_L^2 \eps^2 \!+\! \mathcal{O}(\delta^4) \!+\! \ccalO(\delta^2 \eps^2),
\end{align}
where $\Delta = n d C_L^2$ and the last inequality holds from \eqref{eq:prop3127}. For the norm of $\tilde{\bbx}_{L-1}^f$, we observe that
\begin{equation}\label{eq:thm456}
	\begin{split}
		\mathbb{E}\!\!\left[ \| \tilde{\bbx}^f_{L\!-\!1} \|_2^2 \right] \!=\! \mathbb{E}\!\!\left[ \Big\| \sigma \!\!\left( \sum_{g=1}^F \tilde{\bbu}_{L\!-\!2}^{fg} \right) \!\Big\|_2^2 \right] \!\le\!C_\sigma^2 F\! \sum_{g=1}^F\! \mathbb{E}\!\left[ \left\| \tilde{\bbu}_{L\!-\!2}^{fg} \right\|_2^2 \right]\!,
	\end{split}
\end{equation}
where the Lipschitz property of the nonlinearity from Assumption 1 and the triangle inequality are used in the last inequality. By further representing $\big\| \tilde{\bbu}_{L-2}^{fg} \big\|_2^2$ with the trace $\tr (\tilde{\bbu}_{L-2}^{fg} (\tilde{\bbu}_{L-2}^{fg})^\top)$ and expanding the latter as in \eqref{eq:prop33Ts1}, we have
\begin{align}
	\label{eq:thm458} &\!\!\mathbb{E}\!\left[\!\tilde{\bbu}_{L-2}^{fg} (\tilde{\bbu}_{L-2}^{fg})^\top\!\right]\! \!=\!\!\! \sum_{k,\ell=0}^K \!\!\alpha_{k,L-2}^{fg}\alpha_{\ell,L-2}^{fg}{\barbS}^{k}\mathbb{E}\!\left[\tilde{\bbX}_{L-2}^g\right] {\barbS}^{\ell}\\
	&+\!\!\!\sum_{k,\ell=0}^K\!\!\! \alpha_{k,L\!-\!2}^{fg}\alpha_{\ell,L\!-\!2}^{fg}\mathbb{E}\!\!\left[\! \sum_{r=\!1}^{\lfloor\! k\ell\! \rfloor}\!\! \barbS^{k\!-\!r}\bbE^{(r)} \barbS^{r\!-\!1}\tilde{\bbX}_{L\!-\!2}^g \barbS^{r\!-\!1}\bbE^{(r)} \barbS^{\ell\!-\!r}\!\right] \nonumber\\
	&+\!\sum_{k,\ell=0}^K\!\!\! \alpha_{k,L\!-\!2}^{fg}\alpha_{\ell,L\!-\!2}^{fg} \mathbb{E}\!\left[ \bbC_{k\ell}\right] \!+\! \mathbb{E}\!\left[\!\tilde{\bbu}_{2, L-2}^{fg} (\tilde{\bbu}_{2, L-2}^{fg})^\top\!\right], \nonumber 
\end{align}
where $\tilde{\bbX}_{L-2}^g = \tilde{\bbx}_{L-2}^g (\!\tilde{\bbx}_{L-2}^g)^\top$ is a concise notation and $\tilde{\bbu}_{2, L-2}^{fg}$ is the noise component of $\tilde{\bbu}_{L-2}^{fg}$. For the first term, by using the cyclic property of trace and the inequality \eqref{eq:prop3ineq0}, we can bound it as
\begin{align} \label{eq:thm459}
	&\tr \left( \sum_{k,\ell=0}^K\! \alpha_{k,L-2}^{fg}\alpha_{\ell,L-2}^{fg}{\barbS}^{k}\mathbb{E}\!\left[\tilde{\bbX}_{L-2}^g\right] {\barbS}^{\ell}\right)\nonumber\\
	& \le \Big\| \sum_{k,\ell=0}^K \!\alpha_{k,L-2}^{fg}\alpha_{\ell,L-2}^{fg} \bar{\bbS}^{k+\ell} \Big\|_2 \tr \left( \mathbb{E}\!\left[\tilde{\bbX}_{L-2}^g\right] \right)  \\
	& = \| \barbH_{L-2}^{fg} \barbH_{L-2}^{fg} \|_2 \tr \left( \mathbb{E}\!\left[\tilde{\bbX}_{L-2}^g \right]\right) \le 
	\mathbb{E}\left[\|\tilde{\bbx}_{L-2}^{g} \|^2_2\right],\nonumber
\end{align}
where the filter bound $|f(\bblambda)| \le 1$ and the fact $\tr \left( \mathbb{E}\left[\tilde{\bbX}_{L-2}^g \right]\right) = \mathbb{E}\left[\|\tilde{\bbx}_{L-2}^{g} \|^2_2\right]$ are used in the last inequality. For the second term and the third term, we use the result \eqref{eq:prop312} and \eqref{eq:prop3125} to write
\begin{align}
	\label{eq:thm4510} &\sum_{k,\ell\!=\!0}^K\!\!\! \alpha_{k,L\!-\!2}^{fg}\alpha_{\ell,L\!-\!2}^{fg} \tr\!\! \left(\!\! \!\mathbb{E}\!\!\left[\! \sum_{r\!=\!1}^{\lfloor\! k\ell\! \rfloor}\! \barbS^{k\!-\!r}\bbE^{(\!r\!)}\! \barbS^{r\!-\!1}\!\tilde{\bbX}_{L\!-\!2}^g \barbS^{r\!-\!1}\bbE^{(\!r\!)}\! \barbS^{\ell\!-\!r}\!\!\!+\!\! \bbC_{k\ell}\!\right]\!\!\right)\nonumber\\
	& \le \ccalO(\delta^2). 
\end{align}
For the forth term, we use the result \eqref{eq:prop3126} to write
\begin{align}
	\label{eq:thm45105} \mathbb{E}\!\left[\!\tilde{\bbu}_{2, L-2}^{fg} (\tilde{\bbu}_{2, L-2}^{fg})^\top\!\right] \le \ccalO(\eps^2). 
\end{align}
By substituting \eqref{eq:thm459}, \eqref{eq:thm4510} and \eqref{eq:thm45105} into \eqref{eq:thm458} and the latter into \eqref{eq:thm456}, we get
\begin{equation}\label{eq:thm4511}
	\begin{split}
		\mathbb{E}\!\!\left[ \| \tilde{\bbx}^f_{L\!-\!1} \|_2^2 \right] \!\!\le\!\!  C_\sigma^2 F\! \sum_{g=1}^F \!\mathbb{E}\!\!\left[ \left\| \tilde{\bbx}_{L\!-\!2}^{g} \right\|_2^2 \right] \!\!+\! \ccalO(\delta^2) \!+\! \ccalO(\eps^2).
	\end{split}
\end{equation}
Continuing the recursion \eqref{eq:thm4511} with the initial condition $\| \tilde{\bbx}_0^1 \|_2^2 = \| \bbx \|_2^2$ yields
\begin{equation}\label{eq:thm4512}
	\begin{split}
		& \mathbb{E}[ \| \tilde{\bbx}_{L-1}^f \|_2^2]\! \le\! C_\sigma^{2L-2} F^{2L\!-4} \|\bbx\|_2^2 \!+\! \ccalO(\delta^2) \!+\! \ccalO(\eps^2).
	\end{split}
\end{equation}
By substituting \eqref{eq:thm4512} into \eqref{eq:thm45}, we can bound \eqref{eq:thm45} by 
\begin{equation} \label{eq:thm46}
	\begin{split}
		\Delta F^{2L-3} C_\sigma^{2L-2} \|\bbx\|_2^2 \delta^2  \!+\! n C_L^2 \eps^2 \!+\! \ccalO(\delta^4) \!+\! \ccalO(\eps^2\delta^2).
	\end{split}
\end{equation}

\textbf{Second group of terms in \eqref{eq:thm44}.} For the second group of terms in \eqref{eq:thm44}, we have\footnote{We use the notation ${\rm cov}[\bbx, \bby] = \sum_{i=1}^N {\rm cov}[[\bbx]_i, [\bby]_i]$ for any two random vectors $\bbx$ and $\bby$.}
\begin{align} \label{eq:thm47}
	& \! \sum_{f\!=\!1}^{F}\! \sum_{g\!=\!1}^{F}\! \tr\!\left(\! \mathbb{E} \! \left[ \!\barbH_{L}^{f}\!\tilde{\bbx}^f_{L\!-\!1}\! \left( \!\barbH_{L}^{g}\!\tilde{\bbx}^{g}_{L\!-\!1}\!\right)^\top \right] \! -\! \mathbb{E}\!\left[\! \barbH_{L}^{f}\!\tilde{\bbx}^f_{L\!-\!1}\! \right]\! \mathbb{E} \!\left[\! \barbH_{L}^{g}\!\tilde{\bbx}^{g}_{L\!-\!1}\! \right]^\top \right) \nonumber\\
	& \!= \! \sum_{f=1}^{F} \! {\rm var} \! \left[ \!\barbH_{L}^{f}\!\tilde{\bbx}^f_{L-1} \!\right] \!+\! \sum_{f \ne g}^{F} \! {\rm cov} \! \left[ \! \barbH_{L}^{f}\!\tilde{\bbx}^f_{L-1}, \barbH_{L}^{g}\!\tilde{\bbx}^{g}_{L\!-\!1} \right].
\end{align}
By using the property of covariance \cite{Chandler1987}
\begin{equation} \label{eq:thm48}
	\begin{split}
		{\rm cov}[x,y] \le \sqrt{{\rm var}[x]} \sqrt{{\rm var}[y]} \le \frac{{\rm var}[x]+{\rm var}[y]}{2} 
	\end{split}
\end{equation}
for two random variables $x$ and $y$, we can bound \eqref{eq:thm47} by
\begin{align} \label{eq:thm485}
	&F \sum_{f=1}^{F} \! {\rm var} \! \left[ \!\barbH_{L}^{f}\!\tilde{\bbx}^f_{L-1} \!\right] =\! F\sum_{f=1}^{F} \tr\left( \barbH_{L}^{f} \barbH_{L}^{f}\bbSigma_{L-1}^f \right),
\end{align}
where the trace cyclic property is used, $\bbSigma_{L-1}^f = \mathbb{E} \! \left[ \left( \!\tilde{\bbx}^f_{L\!-\!1}\! -\! \mathbb{E}\!\left[\! \tilde{\bbx}^f_{L\!-\!1}\! \right] \right) \left( \!\tilde{\bbx}^f_{L\!-\!1}\! -\! \mathbb{E}\!\left[\! \tilde{\bbx}^f_{L\!-\!1}\! \right] \right)^\top \right]$ is positive semi-definite matrix for all $f=1,\ldots,F$ and $\barbH_{L}^{f}\! \barbH_{L}^{f}$ is square. We then refer to inequality \eqref{eq:prop3ineq0} to bound \eqref{eq:thm485} as
\begin{align} \label{eq:thm49}
	& F \!\sum_{f\!=\!1}^{F}\! \| \barbH_{L}^{f} \|_2^2 \tr\!\left( \bbSigma_{L-1}^f  \right)  \le 
	F \sum_{g=1}^{F} {\rm var} \left[ \tilde{\bbx}^f_{L-1} \right],
\end{align}
where the filter bound $|f(\bblambda)|\le 1$ is used. By substituting bounds \eqref{eq:thm46} and \eqref{eq:thm49} into \eqref{eq:thm44} and then altogether into \eqref{eq:thm42}, we observe a recursion where the variance of $\ell$th layer output depends on the variance of $(\ell-1)$th layer output as well as the bound in \eqref{eq:thm46}. Therefore, we have
\begin{align} \label{eq:thm410}
	&{\rm var}[\tilde{\bbx}_L]\!:=\!{\rm var}[\bbPhi_{\rm air}(\bbx,\bbS,\ccalA)] \! \le\! F \!\sum_{f=1}^{F}\! {\rm var} \left[ \tilde{\bbx}^f_{L-1} \right]  \\
	& + \Delta F^{2L-3} C_\sigma^{2L-2} \| \bbx \|_2^2 \delta^2 + nC_L^2 \eps^2  + \ccalO(\delta^4) + \ccalO(\eps^2\delta^2).\nonumber
\end{align}

\textbf{Unrolling.} By performing this recursion until the input layer, we have
\begin{align} \label{eq:thm411}
	& {\rm var}[ \bbPhi_{\rm air}(\bbx,\bbS,\ccalA)] \!\!\le\!\! \Delta\!\!\! \sum_{\ell=2}^L \!\!F^{2L-3} C_\sigma^{2\ell\!-\!2} \| \bbx \|_2^2 \delta^2 \!\!+\! nC_L^2 \eps^2 \\
	& \!+\! F^{2L-4} \left(\! \sum_{f\!=\!1}^{F} \!\!{\rm var}\! \left[ \tilde{\bbx}^f_1 \right] \!\!+ \!\!\!\sum_{f\ne g}^{F} \!\!{\rm cov}\! \left[ \!\tilde{\bbx}^f_1, \tilde{\bbx}^g_1 \!\right]\!\! \right)\!+\! \ccalO(\delta^4) \!+\! \ccalO(\eps^2\delta^2),\nonumber
\end{align}
where the first term in \eqref{eq:thm411} is the accumulated sum of all second terms in \eqref{eq:thm410} during the recursion. Since $\tilde{\bbx}_1^f = \tilde{\bbH}_{1}^f \bbx$, $\tilde{\bbx}_1^g = \tilde{\bbH}_{1}^g \bbx$ and $\tilde{\bbH}_{1}^f$, $\tilde{\bbH}_{1}^g$ are independent if $f \ne g$, we have $\sum_{f\ne g}^{F} \!{\rm cov}\! \left[ \tilde{\bbx}^f_1, \tilde{\bbx}^g_1 \right]=0$. Therefore, \eqref{eq:thm411} becomes
\begin{align} \label{eq:thm412}
	& {\rm var}[ \bbPhi_{\rm air}(\bbx,\bbS,\ccalA)] \\
	& \le\! n d C_L^2 \! \sum_{\ell\!=\!1}^L\! F^{2L\!-\!3} C_\sigma^{2\ell\!-\!2} \| \bbx \|_2^2 \delta^2\!+\! nC_L^2 \eps^2\!+\!\ccalO(\delta^4)\!+\!\ccalO(\eps^2\delta^2) \nonumber 
\end{align}
completing the proof for AirGNNs without CSI. For AirGNNs with CSI, the channel coefficient $h$ can be considered as a binary variable following the Bernoulli distribution with expectation $\rho$ and variance $\rho (1-\rho)$ [Section IV-B]. By substituting this result into \eqref{eq:thm412}, we complete the proof for AirGNNs with CSI.


\section{Lemmas and Proofs}

\begin{lemma}\label{lemma0:traceOperation}
	Consider the same setting as Theorem 1. For any graph signal $\bbx$, it holds that
	\begin{align}
		\label{lemma0:mainresults}
		&\mathbb{E}\!\Big[\! \tr \Big(\! \sum_{k\!=\!1}^K\! \sum_{\ell\!=\!1}^K\!\! \alpha_k \alpha_\ell\!\! \sum_{r_1\!=\!1}^{k}\! \sum_{r_2\!=\!1}^\ell\!\! \bbE^{(\!r_2\!)}\bbS^{k\!+\!\ell\!-\!r_1\!-\!r_2}\bbE^{(\!r_1\!)} \bbS^{r_1\!-\!1}\bbx \bbx^\top\! \bbS^{r_2\!-\!1}\! \Big)\!\Big] \nonumber\\
		& \le (n+K-1) C_L^2 \|\bbx\|_2^2 (1-\mu)^2 + ndC_L^2\|\bbx\|_2^2 \delta^2,
	\end{align}
	where $d$ is the maximal degree of the graph $\ccalG$. 
\end{lemma}
\begin{proof}
	Since $\bbE^{(r_1)}$ and $\bbE^{(r_2)}$ are independent if $r_1 \ne r_2$, we separate the terms in \eqref{lemma0:mainresults} as
	\begin{align}\label{proof:lemma01}
		&\mathbb{E}\!\Big[\! \tr \Big(\! \sum_{k\!=\!1}^K\! \sum_{\ell\!=\!1}^K\!\! \alpha_k \alpha_\ell\!\! \sum_{r_1\!=\!1}^{k}\! \sum_{r_2\!=\!1}^\ell\!\! \bbE^{(\!r_2\!)}\bbS^{k\!+\!\ell\!-\!r_1\!-\!r_2}\bbE^{(\!r_1\!)} \bbS^{r_1\!-\!1}\bbx \bbx^\top\! \bbS^{r_2\!-\!1}\! \Big)\!\Big] \nonumber\\
		& =\!\! \mathbb{E}\Big[\! \sum_{r\!=\!1}^K \!\tr \Big(\! \sum_{k, \ell\!=\!r}^K \! \alpha_k \alpha_\ell \bbE^{(r)} \bbS^{k\!+\!\ell\!-\!2r}\bbE^{(r)} \bbS^{r\!-\!1}\bbx \bbx^\top\! \bbS^{r\!-\!1} \!\Big)\!\Big] \!+\! {\rm tr}(\bbZ).
	\end{align}
	The first term in \eqref{proof:lemma01} contains the terms with error matrices of the same index and the second term contains the rest terms. We consider these two terms separately.
	
	\smallskip
	\noindent \textbf{First term in \eqref{proof:lemma01}.} Since both matrices $\sum_{k, \ell=r}^K \alpha_k \alpha_\ell \bbE^{(\!r\!)} \bbS^{k\!+\!\ell\!-\!2r}\bbE^{(\!r\!)} = \big(\sum_{k=r}^K \alpha_k \bbS^{k-r}\bbE_r\big)^\top \big(\sum_{k=r}^K \alpha_k \bbS^{k-r}\bbE_r\big)$ and $\bbS^{r-1}\bbx \bbx^\top \bbS^{r-1}$ are positive semi-definite, we use the Cauchy-Schwarz inequality $\tr(\bbA \bbB) \le \tr(\bbA)\tr(\bbB)$ and the linearity of the trace to have 
	\begin{align}
		\label{proof:lemma02} &\mathbb{E}\Big[\! \sum_{r=1}^K \tr \Big(\! \sum_{k, \ell=r}^K \! \alpha_k \alpha_\ell \bbE^{(r)} \bbS^{k+\ell-2r}\bbE^{(r)} \bbS^{r-1}\bbx \bbx^\top\! \bbS^{r\!-\!1} \!\Big)\!\Big]\\
		&\le\mathbb{E}\Big[\! \sum_{r\!=\!1}^K\! \sum_{k, \ell=r}^K \!\alpha_k \alpha_\ell \tr\! \left( \bbE^{(r)} \bbS^{k+\ell-2r}\bbE^{(r)}\! \right)\! \tr\! \left( \bbS^{r\!-\!1}\bbx \bbx^\top \bbS^{r\!-\!1} \right)\!\!\Big]\!.\nonumber
	\end{align}
	Given the eigendecomposition $\bbS = \bbV\bbLambda\bbV^\top$ with eigenvectors $\bbV = [\bbv_1,\ldots,\bbv_n]^\top$ and eigenvalues $\bbLambda = \text{diag}(\lambda_1, \ldots, \lambda_n)$, the GFT of the graph signal $\bbx$ is $\bbx = \sum_{i=1}^n \hat{x}_i \bbv_i$. By substituting the latter into $\tr\! \left( \bbS^{r-1}\bbx \bbx^\top \bbS^{r-1} \right)$, we get
	\begin{gather}\label{proof:lemma03}
		\tr\! \left(\! \bbS^{r\!-\!1}\bbx \bbx^\top\! \bbS^{r\!-\!1}\! \right) \!\!=\!\! \sum_{i=1}^n\! \hat{x}_i^2 \lambda_i^{2r-2} \tr\! \left( \bbv_i \bbv_i^\top\!\right) \!\!=\!\! \sum_{i=1}^n \hat{x}_i^2 \lambda_i^{2r-2}\!,
	\end{gather}
	where $\tr(\bbv_i\bbv_i^\top) = 1$ for $i=1,\ldots,n$ is used due to the orthonormality of eigenvectors. By substituting \eqref{proof:lemma03} into \eqref{proof:lemma02} and using the trace cyclic property $\tr(\bbA\bbB\bbC) = \tr(\bbC\bbA\bbB) = \tr(\bbB\bbC\bbA)$, we can rewrite \eqref{proof:lemma02} as
	\begin{align} \label{proof:lemma04}
		& \sum_{i = 1}^n\hat{x}_i^2\! \sum_{r=1}^K\!\tr \Big( \sum_{k, \ell=r}^K \alpha_k \alpha_\ell \lambda_i^{2r-2} \bbS^{k+\ell-2r} \mathbb{E}\!\left[ {\bbE^{(r)}}^2\right] \Big).
	\end{align}
	From the result of Lemma \ref{LemmmaErrorMatrix}, we have $\mathbb{E}\left[ \bbE_k^2 \right] = (1-\mu)^2\bbS^2+\delta^2 \bbD$ where $\bbD$ is the degree matrix of graph $\ccalG$ and $\mu$, $\delta$ are the expectation, standard deviation of channel distribution. Substituting these results into \eqref{proof:lemma04}, we can represent it as
	\begin{align} \label{proof:lemma05}
		& (1 - \mu)^2 \sum_{i = 1}^n\hat{x}_i^2\! \sum_{r\!=\!1}^K\!\tr \Big( \sum_{k, \ell=r}^K \alpha_k \alpha_\ell \lambda_i^{2r-2} \bbS^{k+\ell-2r+2} \bbI \Big)\\
		&+ \delta^2 \sum_{i = 1}^n\hat{x}_i^2\! \sum_{r\!=\!1}^K\!\tr \Big( \sum_{k, \ell=r}^K h_k h_\ell \lambda_i^{2r-2} \bbS^{k+\ell-2r} \bbD \Big) \nonumber.
	\end{align}
	We consider two terms in \eqref{proof:lemma05}, respectively. For the first term, we use the inequality \eqref{proof:them16} (since $\bbI$ is positive semi-definite) to upper bound it by
	\begin{align}
		\label{proof:lemma06}
		& (1\!-\! \mu)^2\sum_{i = 1}^n\!\hat{x}_i^2 \big\|\! \sum_{r\!=\!1}^K\! \sum_{k, \ell=r}^K\! \alpha_k \alpha_\ell \lambda_i^{2r\!-\!2} \bbS^{k\!+\!\ell\!-\!2r\!+\!2} \big\|_2 \tr\! \left( \bbI \right).
	\end{align}
	We now consider the matrix norm in \eqref{proof:lemma06}. For any matrix $\bbA$, a standard way to bound $\|\bbA\|_2$ is to obtain the inequality $\|\bbA\bba\|_2 \le A \|\bba\|_2$ that holds for any vector $\bba$ \cite{Meyer2000}. In this context, $A$ is the upper bound satisfying $\|\bbA\|_2 \le A$. Following this rationale, we consider the GFT of $\bba$ over $\bbS$ as $\bba = \sum_{j=1}^n \hat{a}_j \bbv_j$ and have
	\begin{align}
		\label{proof:lemma07} & \big\| \!\sum_{r=1}^K\! \sum_{k, \ell=r}^K\! \alpha_k \alpha_\ell \lambda_i^{2r-2} \bbS^{k+\ell-2r+2} \bba \big\|_2^2\\
		&=\! \sum_{j=1}^n\! \hat{a}_j^2 \big| \sum_{r=1}^K\! \sum_{k, \ell=r}^K\! \alpha_k \alpha_\ell \lambda_i^{2r-2}\! \lambda_j^{k+\ell-2r+2}\big|^2. \nonumber
	\end{align}
	The expression inside the absolute value in \eqref{proof:lemma07} can be linked to the Lipschitz gradient of the frequency response over the air $f(\bblambda)$ in (14). More specifically, let $\bblambda_i = [\lambda_i, \ldots,\lambda_i]^\top$ and $\bblambda_j = [\lambda_j,\ldots,\lambda_j]^\top$ be specific multivariate frequencies and $\bblambda_{i:j,r} = [\lambda_{i}, \ldots, \lambda_{i}, \lambda_{j}, \ldots, \lambda_{j}]^\top$ formed by concatenating the first $r$ entries of $\bblambda_j$ and the last $K-r$ entries of $\bblambda_i$. The partial derivative of $f(\bblambda)$ w.r.t. the $r$th entry $\lambda^{(r)}$ at $\bblambda_{i:j,r}$ is
	\begin{align}
		\label{proof:lemma08}
		\frac{\partial f(\bblambda_{i:j,r})}{\partial \lambda^{(r)}} \!=\! \sum_{k=r}^K \alpha_k \lambda_{i}^{r-1} \lambda_j^{k-r},\! ~\forall~r\!=\!1,\!\ldots,\!K.
	\end{align}
	The Lipschitz gradient of $f(\bblambda)$ over $\bblambda_i$ and $\bblambda_j$ [Def. 2] is
	\begin{equation}\label{proof:lemma09}
		\nabla_L f(\bblambda_i, \bblambda_j) = \Big[\frac{\partial f(\bblambda_{i:j,1})}{\partial \lambda^{(1)}}, \ldots, \frac{\partial f(\bblambda_{i:j,K})}{\partial \lambda^{(K)}}\Big]^\top.
	\end{equation}
	We observe the expression inside the absolute value in \eqref{proof:lemma07} can be written in the compact form and upper bounded as
	\begin{align}
		\label{proof:lemma010} &\big|\sum_{r\!=\!1}^K\! \sum_{k, \ell=r}^K\! h_k h_\ell \lambda_i^{2r-2} \lambda_j^{k+\ell-2r+2}\big| \!=\!\!  \sum_{r\!=\!1}^K \!\Big(\frac{\partial f(\bblambda_{i:j,r})}{\partial\lambda^{(r)}} \lambda_j \Big)^2\nonumber \\
		&=\! \| \nabla_L f(\bblambda_i, \bblambda_j) \!\odot\! \bblambda_j \|^2_2 \le C_L^2, 
	\end{align}
	where the integral Lipschitz condition [cf. (16)] is used and $\odot$ is the elementwise product. Using \eqref{proof:lemma010} in \eqref{proof:lemma07}, we can bound the matrix norm by $C_L^2$. Further substituting this result into \eqref{proof:lemma06} yields 
	\begin{align}
		\label{proof:lemma011} &(1 - \mu)^2\sum_{i = 1}^N\hat{x}_i^2 \big\| \sum_{r\!=\!1}^K\! \sum_{k, \ell=r}^K\! \alpha_k \alpha_\ell \lambda_i^{2r-2} \bbS^{k+\ell-2r+2} \big\|_2 \tr\! \left( \bbI \right)\nonumber \\
		& \qquad\le (1-\mu)^2n C_L^2 \sum_{i=1}^N \hat{x}_i^2 = n C_L^2 \| \bbx \|^2_2 (1-\mu)^2.
	\end{align}
	
	For the second term in \eqref{proof:lemma05}, we use the inequality \eqref{proof:them16} to bound it by
	\begin{align}
		\label{proof:lemma012} & \delta^2 \sum_{i = 1}^N\hat{x}_i^2 \big\|\! \sum_{r=1}^K\! \sum_{k, \ell=r}^K\! \alpha_k \alpha_\ell \lambda_i^{2r-2} \bbS^{k+\ell-2r} \big\|_2 \tr\! \left( \bbD \right).
	\end{align}
	We follow \eqref{proof:lemma07}-\eqref{proof:lemma010} to bound the matrix norm in \eqref{proof:lemma012} by using the Lipschitz gradient [Def. 2] and the integral Lipschitz condition [cf. (16)] as
	\begin{align}
		\label{proof:lemma013} \big\| \sum_{r=1}^K \sum_{k, \ell=r}^K\! \alpha_k \alpha_\ell \lambda_i^{2r-2} \bbS^{k+\ell-2r} \big\|_2 \!\le\! \| \nabla_L f(\bblambda_i, \bblambda_j) \|^2_2 \!\le\! C_L^2.
	\end{align}
	Using \eqref{proof:lemma013} in \eqref{proof:lemma012} and the fact $\tr(\bbD) \le n d$ with $d$ the maximal degree of graph, we can bound \eqref{proof:lemma012} by $n d C_L^2 \| \bbx \|_2^2 \delta^2$. 
	
	\smallskip
	\noindent \textbf{Second term in \eqref{proof:lemma01}.} Since the error matrices $\bbE^{(r_1)}$ and $\bbE^{(r_2)}$ in $\bbZ$ are independent with $r_1\ne r_2$, we have $\mathbb{E}[\bbE^{(r_1)}] = \mathbb{E}[\bbE^{(r_2)}] = (1-\mu) \bbS$. By following the same steps as in \eqref{proof:lemma02}-\eqref{proof:lemma05}, we can rewrite $\bbZ$ as
	\begin{align}\label{proof:lemma014}
		&(1-\mu)^2 \sum_{i=1}^n \hat{x}_i^2 \sum_{r_1\ne r_2}^K\! \sum_{k=r_1}^K\! \sum_{\ell=r_2}^K\! \alpha_k \alpha_\ell \lambda_i^{k\!+\!\ell} \tr \Big(\! \bbv_i \bbv_i^\top\! \Big)\\
		& \le (1-\mu)^2 \sum_{i=1}^n \hat{x}_i^2 \Big|\sum_{r_1\ne r_2}^K\! \sum_{k=r_1}^K\! \sum_{\ell=r_2}^K\! \alpha_k \alpha_\ell \lambda_i^{k+\ell}\Big|. \nonumber
	\end{align}
	The expression inside the absolute value in \eqref{proof:lemma014} can be linked to the Lipschitz gradient of the frequency response over the air $f(\bblambda)$ as well. Specifically, let $\bblambda_1 = \bblambda_2 = [\lambda_i,\ldots,\lambda_i]^\top$ be the multivariate frequency and the partial derivative of $f(\bblambda)$ w.r.t. the $r$th entry $\lambda^{(r)}$ of $\bblambda_{1:2,r}$ is
	\begin{align}
		\label{proof:lemma015}
		\frac{\partial f(\bblambda_{1:2,r})}{\partial \lambda^{(r)}} \!=\! \sum_{k=r}^K \alpha_k \lambda_i^{k-1},\! ~\forall~r\!=\!1,\!\ldots,\!K.
	\end{align}
	By using the triangle inequalities $|a+b| \le |a| + |b|$ and $|ab| \le (a^2 + b^2)/2$, we can upper bound the absolute value in \eqref{proof:lemma014} as
	\begin{align}
		\label{proof:lemma016}
		&\Big|\sum_{r_1\ne r_2}^K\! \sum_{k=r_1}^K\! \sum_{\ell=r_2}^K\! \alpha_k \alpha_\ell \lambda_i^{k+\ell}\Big|\\
		&=\!\!\Big|\! \Big(\!\sum_{r=1}^K\!\! \frac{\partial f(\bblambda_{1:2,r})}{\partial \lambda^{(r)}} \lambda_i\!\!\Big)\!\Big(\!\sum_{r=1}^K\!\frac{ \partial f(\bblambda_{1:2,r})}{\partial \lambda^{(r)}} \lambda_i\!\!\Big) \!\!-\!\! \sum_{r\!=\!1}^K \!\!\Big(\! \frac{\partial f(\bblambda_{1:2,r})}{\partial\lambda^{(r)}} \lambda_i\!\! \Big)^2\!\Big| \nonumber \\
		& \le (K-1) \sum_{r\!=\!1}^K \!\Big(\! \frac{\partial f(\bblambda_{1:2,r})}{\partial\lambda^{(r)}}\lambda_i \! \Big)^2 \le (K-1) C_L^2,\nonumber
	\end{align}
	where the integral Lipschitz condition is used in the last inequality. By substituting \eqref{proof:lemma016} into \eqref{proof:lemma014}, we can upper bound the second term in \eqref{proof:lemma01} as
	\begin{align}\label{proof:lemma017}
		{\rm tr}(\bbZ) \le (K-1) C_L^2 \|\bbx\|_2^2 (1-\mu)^2.
	\end{align}
	
	By using the bounds for \eqref{proof:lemma011}, \eqref{proof:lemma012} and \eqref{proof:lemma017} in \eqref{proof:lemma01}, we complete the proof
	\begin{align}\label{proof:lemma018}
		&\mathbb{E}\!\Big[\! \tr \Big(\! \sum_{k\!=\!1}^K\! \sum_{\ell\!=\!1}^K\!\! \alpha_k \alpha_\ell\!\! \sum_{r_1\!=\!1}^{k}\! \sum_{r_2\!=\!1}^\ell\!\! \bbE^{(\!r_2\!)}\bbS^{k\!+\!\ell\!-\!r_1\!-\!r_2}\bbE^{(\!r_1\!)} \bbS^{r_1\!-\!1}\bbx \bbx^\top\! \bbS^{r_2\!-\!1}\! \Big)\!\Big] \nonumber\\
		& \le (n+K-1) C_L^2 \|\bbx\|_2^2 (1-\mu)^2 + ndC_L^2\|\bbx\|_2^2 \delta^2.
	\end{align}
\end{proof}

\begin{lemma} \label{LemmmaErrorMatrix}
	Consider the graph shift operator over the air $\bbS_{\rm air}$ with the channel coefficient $h$ drawn from a random distribution with expectation $\mu$ and standard deviation $\delta$. Let $\bbS$ be the underlying graph shift operator, $\bbE_{\rm air} = \bbS_{\rm air} - \bbS$ be the deviation matrix, and $\bbD$ be the diagonal degree matrix with $d_i$ the degree of node $i$. Then, it holds that
	\begin{equation}\label{eq:LemmaErrorMatrix}
		\mathbb{E}\left[ \bbE_{\rm air}^2 \right] = (1-\mu)^2 \bbS + \delta^2 \bbD.
	\end{equation}  
\end{lemma}

\begin{proof}
	By definition, the $(i,j)$th entry of $\bbS_{\rm air}$ can be represented as $[\bbS_{\rm air}]_{ij} = h_{ij} s_{ij}$, where $h_{ij}$ is a random channel coefficient with expectation $\mu$ and standard deviation $\delta$ and where $s_{ij}$ is the $(i,j)$th entry of the adjacency matrix $\bbS$.\footnote{In the decentralized setting, we consider the graph shift operator as the adjacency matrix with communication links as graph edges.} By exploiting the matrix multiplication for $\bbS_{\rm air}\bbS_{\rm air}$, the $(i,j)$th entries of $\mathbb{E}[\bbS_{\rm air}^2]$ is given by 
	\begin{align} \label{prlem31}
		\left[\mathbb{E}[\bbS_{\rm air}^2]\right]_{ij} \!=\! \sum_{n\!=\!1}^N \!s_{in}s_{nj}\mathbb{E}[h_{in}h_{nj}].
	\end{align}
	The channel coefficients $\{ h_{ij} \}_{ij}$ are independent except for $h_{ij}=h_{ji}$ since $\bbS_{\rm air}$ is symmetric. Thus, we get 
	\begin{equation}\label{prlem3155}
		\mathbb{E}[ h_{in}h_{nj}] =
		\begin{cases}
			\mu^2,  & \mbox{if }i \ne j, \\
			\delta^2 + \mu^2, & \mbox{if }i=j.
		\end{cases}
	\end{equation}
	By substituting \eqref{prlem3155} into \eqref{prlem31}, we have
	\begin{equation}\label{prlem32}
		\left[\mathbb{E}[\bbS_{\rm air}^2]\right]_{ij} =
		\begin{cases}
			\sum_{n=1}^N \!s_{in}s_{nj}\mu^2,  & \mbox{if }i \ne j, \\
			\sum_{n=1}^N \!s_{in}s_{nj}(\delta^2 + \mu^2), & \mbox{if }i=j.
		\end{cases}
	\end{equation}
	Since $\left[\bbS^2\right]_{ij} \!=\! \sum_{n=1}^N \!s_{in}s_{nj}$ and $\sum_{n=1}^N \!s_{in}s_{ni} = d_i$ is the degree of node $i$, we can rewrite $\mathbb{E}[\bbS_{\rm air}^2]$ as
	\begin{align} \label{prlem33}
		\mathbb{E}[\bbS_{\rm air}^2] = \mu^2 \bbS^2 + \delta^2 \bbD,
	\end{align}
	where $\bbD$ is the degree matrix. 
	
	By substituting the representation $\bbE_{\rm air} = \bbS_{\rm air} - \bbS$ into $\mathbb{E}[\bbE_{\rm air}^2]$ and expanding the multiplication terms, we have
	\begin{align}\label{prlem34}
		\mathbb{E}[\bbE_{\rm air}^2] &= \mathbb{E}\big[ \bbS_{\rm air}^2 - \bbS_{\rm air}\bbS - \bbS\bbS_{\rm air} + \bbS^2 \big]\\
		& = \mathbb{E}\big[ \bbS_{\rm air}^2] - 2 \mu \bbS^2 + \bbS^2, \nonumber 
	\end{align}
	where the linearity of the expectation and the fact $\mathbb{E}[\bbS_{\rm air}] = \mu \bbS$ is used. By further substituting \eqref{prlem33} into \eqref{prlem34}, we get 
	\begin{align}\label{prlem35}
		\mathbb{E}[\bbE_{\rm air}^2] &= \mu^2 \bbS^2 - 2 \mu \bbS^2 + \bbS^2 + \delta^2 \bbD \\
		& = (1 - \mu)^2 \bbS + \delta^2 \bbD, \nonumber
	\end{align}
	which completes the proof.
\end{proof}


\bibliographystyle{IEEEbib}
\bibliography{myIEEEabrv,biblioOp}

\end{document}